\documentclass{iopconfser}
\usepackage{amssymb}
\usepackage{amsfonts}
\usepackage{amsmath}

\newtheorem{theorem}{Theorem}

\newtheorem{definition}[theorem]{Definition}

\newtheorem{proposition}[theorem]{Proposition}

\newenvironment{proof}[1][Proof]{\noindent\textbf{#1.} }{\ \rule{0.5em}{0.5em}}

\begin{document}

\title{Topological, Differential Geometry Methods and Modified Variational
Approach for Calculation of the Propagation Time of a Signal, Emitted by a
GPS-Satellite and Depending on the Full Set of 6 Kepler Parameters}

\author{Bogdan G. Dimitrov$^{1,2}$}

\affil{$^1$Institute of Nuclear Research and Nuclear Energetics, Bulgarian
Academy of Sciences, 72 Tzarigradsko shaussee, 1784 Sofia, Bulgaria} 
\affil{$^2$Institute for Advanced Physical Studies , Sofia Tech Park, 111
Tzarigradsko shaussee, 1784 Sofia, Bulgaria}

\email{$^1$dimitrov.bogdan.bogdan@gmail.com}
\email{$^2$bogdan.dimitrov@iaps.institute}

\begin{abstract}
In preceding publications a mathematical approach has been developed for
calculation of the propagation time of a signal, emitted by a moving along
an elliptical orbit satellite and accounting also for the General Relativity
Theory (GRT) Effects. So far, the formalism has been restricted to one
dynamical parameter (the true anomaly or the eccentric anomaly angle). In
this paper the aim is to extend the formalism to the case, when also the
other five Kepler parameters will be changing and thus, the following
important problem can be stated: if two satellites move on two
space-distributed orbits and they exchange signals, how can the propagation
time be calculated? This paper requires the implementation of differential
geometry and topological methods. In this approach, the action functional
for the propagation time is represented in the form of a quadratic
functional in the differentials of the Kepler elements, consequently the
problem is related to the first and second quadratic forms from differential
geometry. Such an approach has a clear advantage, because if the functional
is written in terms of Cartesian coordinates X, Y, Z, the extremum value
after the application of the variational principle is shown to be the
straight line - a result, known from differential geometry, but not
applicable to the current problem of signal exchange between satellites on
different orbits. So the known mapping from celestial mechanics is used,
when by means of a transformation the 6 Kepler parameters are mapped into
the cartesian coordinates X, Y, Z. This is in fact a submersion of a
manifold of 6 parameters into a manifold of 3 parameters. If a variational
approach is applied with respect to a differential form in terms of the
differentials of the Kepler parameters, the second variation will be
different from zero and the Stokes theorem can be applied, provided that the
second partial derivatives of the Cartesian coordinates with respect to the
Kepler parameters are assumed to be different from zero. ` From topology
this requirement is equivalent to the existence of the s.c. Morse functions
(non-degenerate at the critical points). In the given case it has been shown
that Morse function cannot exist with respect to each one of the Kepler
parameters- Morse function cannot be defined with respect to the omega angle.
\end{abstract}

\section{Introduction}

\subsection{Basic facts about the standard Shapiro delay formulae and its
implementation in contemporary space experiments}

Propagation of signals (radio, light or laser) in the gravitational field of
the near-Earth space is an important theory in applied gravitational theory,
since it is related to many astrophysical experiments in the near-Earth
space and also in the deep space (such as LISA, for example). \bigskip \cite%
{Shapiro1964}. The main ingredient of this theory is the Shapiro delay
formulae 
\begin{equation}
T_{AB}=\frac{R_{AB}}{c}+\frac{2GM_{E}}{c^{3}}\ln \left( \frac{r_{A}\newline
+r_{B}+R_{AB}}{r_{A}+r_{B}-R_{AB}}\right) \ \ ,  \label{AA19}
\end{equation}%
which expresses the signal propagation time $T_{AB}=T_{B}-T_{A}$ between two
space-time points and is calculated from the null cone equation%
\begin{equation}
ds^{2}=-c^{2}\left( 1+\frac{2V}{c^{2}}\right) (dT)^{2}+\left( 1-\frac{2V_{.}%
}{c^{2}}\right) \left( (dx)^{2}+(dy)^{2}+(dz)^{2}\right) =0\text{ \ \ ,}
\label{DOP25}
\end{equation}%
obtained after setting up the infinitesimal metric element equal to zero. In
formulae (\ref{AA19}) $G_{\oplus }M_{\oplus }$ is the geocentric
gravitational constant, $M_{\oplus }$ is the Earth mass, $G_{\oplus }$ is
the gravitational constant, the standard potential of the Earths
gravitational field $V=\frac{G_{\oplus }M_{\oplus }}{r}$ is considered
without taking into account the harmonics. The coordinates of the emitting
and of the receiving satellite are correspondingly $\mid x_{A}(t_{A})\mid
=r_{A}$ and $\mid x_{B}(t_{B})\mid =r_{B}$, $\ R_{AB}=$ $\mid
x_{A}(t_{A})-x_{B}(t_{B})\mid $ is the Euclidean distance between the signal
- emitting satellite and the signal - receiving satellite. We shall denote
also by $t=TCG$ the Geocentric Coordinate Time (TCG). The second term in
formulae Eq. (\ref{AA19}) is the Shapiro time delay term, \ accounting for
the signal delay due to the curved space-time.The Shapiro delay term has an
important physical meaning - it means that in a curved space-time due to the
General Relativity effects, a signal travels a longer time because of the
curved signal trajectory.

Since the discovery of the Shapiro delay formulae it has found numerous
applications in VLBI interferometry (see also the review article \cite%
{Interf:CA7A2B1} by Sovers, Fanselow and Jacobs) and also in creation of
models for relativistic reference frames for the GRAIL mission (Gravity
Recovery and Interior Laboratory) \cite{Turysh1:AAB62} accurate to $1$ $\mu
m/\sec $ ($1$ $micrometer$ per $\sec ond$; $1$ micrometer$=10^{-6}$ $m$ )
and also for the GRACE-Follow-on mission (Gravity Recovery and Climate
Experiment) \cite{Turysh2:AAB63}. For the second model, the relativistic
change of the phase of the signal has been computed due to the changing
(geodesic!) distance between the emitting satellite and the receiving
satellite. In fact, in these papers \cite{Turysh1:AAB62}, \cite%
{Turysh2:AAB63} and \cite{Turysh3:AAB64} the difference between the emitting
time and the receiving time turns out to be a relativistic observable,
dependent on the geodesic distance, travelled by light - in such a case it
contains also the relativistic logarithmic correction, typical for the
Shapiro delay formulae. In the third paper \cite{Turysh3:AAB64} (for the
case of the ACES experiment - Atomic Clock Ensemble in Space) of the series
of these papers the difference between the two times is shown to be 
\begin{equation}
T_{2}-T_{1}=\frac{\mid r_{2}-r_{1}\mid }{c}+(1+\gamma )\sum\limits_{b}\frac{%
GM_{b}}{c^{3}}\ln \frac{r_{1}^{b}+r_{2}^{b}+r_{12}^{b}}{%
r_{1}^{b}+r_{2}^{b}-r_{12}^{b}}+O(c^{-5})\text{ \ ,}  \label{AA20}
\end{equation}%
where the summation is over the "deflecting" bodies (in the case-the Moon,
the Earth and the Sun), causing the Shapiro time delay.

\subsection{Uncertainties in the Shapiro formulae due to some approximations
in the integration}

Now it is important to stress that eq. (\ref{AA20}), as well as (\ref{AA19})
are obtained from the null cone equation after performing an integration
over a variable $r$, related to the distance between the transmitter and the
receiver. However, in the known review article by Ashby from 2003 \cite%
{Ashby2003LivRev} it was proved that at the level of several millimeters,
spatial curvature effects have to be considered, which is evident from the
simple equality 
\begin{equation}
\int dr\left[ 1+\frac{GM_{E}}{c^{2}r}\right] \approx r_{2}-r_{1}+\frac{GM_{E}%
}{c^{2}}\ln \frac{r_{2}}{r_{1}}\text{ \ \ .}  \label{AA21}
\end{equation}%
In \cite{Ashby2003LivRev} the second term was estimated to be $4.43\times
\ln (4.2)\approx 6.3$ $mm$. Also in the thesis \cite{Duchayne2008PhD} the
approximation $cdt=\left( 1+\frac{GM_{E}}{c^{2}r}\right) \parallel
dr\parallel $ has been used, where the infinitesimal vector $dr$ is not
related to the signal trajectory, represented by the vector $\overrightarrow{%
R}=\overrightarrow{r}-\overrightarrow{r}_{A}$, where $\overrightarrow{r}_{A}$
is the vector, associated with the position of emission and $\overrightarrow{%
r}$ - with a position of a point in a Geocentric Reference Frame. In other
words, the integration is not along \ $\parallel dR\parallel $, because the
approximation $\parallel dR\parallel \approx \parallel dr\parallel $ has
been assumed and thus the integration is along $dr$. A similar approximation
has been used in the PhD thesis \cite{Gulklett2003}.

All these approximations are incompatible with the uncertainty $1$ $\mu
m/\sec $ \ of determination of the Relativistic Reference Frame in \cite%
{Turysh1:AAB62}, but also in contemporary investigations of the
gravitational field by means of resolving the gravitational redshift \cite%
{Bothwell2022} across a millimeter scale atomic sample. This means that a
clock will tick differently at the top and bottom of a sample of $100$ $000$
ultra-cold strontium atoms \cite{Aeppli2024}, i.e. the clock ticks slowly at
lower elevations.

\subsection{Brief review of the approach for calculation of the propagation
time by using the more general transformation in celestial mechanics and the
true anomaly angle $f$ as the only dynamical parameter}

The accuracy of determination of the propagation time evidently will depend
also on the choice of the integration variable. For example, in the
monograph \cite{KopeikinBook} and in many other books on celestial mechanics
(see, for example \cite{Gurfil}) a transformation from the six Kepler
parameters $(f,\Omega ,\omega ,i,a,e)$ to the cartesian coordinates $(x,y,z)$
has been used 
\begin{equation}
x=\frac{a(1-e^{2})}{1+e\cos f}\left[ \cos \Omega \cos (\omega +f)-\sin
\Omega \sin (\omega +f)\cos i\right] \text{ \ \ ,}  \label{K1}
\end{equation}%
\begin{equation}
y=\frac{a(1-e^{2})}{1+e\cos f}\left[ \sin \Omega \cos (\omega +f)+\cos
\Omega \sin (\omega +f)\cos i\right] \text{ \ \ ,}  \label{K2}
\end{equation}%
\begin{equation}
z=\frac{a(1-e^{2})}{1+e\cos f}\sin (\omega +f)\sin i\text{ \ \ \ ,}
\label{K3}
\end{equation}

where $(a,e)$ are the large semi-major axis of the satellites elliptical
orbit and eccentricity respectively, $r=\frac{a\left( 1-e^{2}\right) }{%
1+e\cos f}$ is the radius-vector in the orbital plane, the angle $\Omega $
of the longitude of the right ascension of the ascending node is the angle
between the line of nodes and the direction to the vernal equinox, the
argument of perigee (periapsis) $\omega $ is the angle within the orbital
plane from the ascending node to perigee in the direction of the satellite
motion $\left( 0\leq \omega \leq 360^{0}\right) $. The angle $i$ is the
inclination of the orbit with respect to the equatorial plane and the true
anomaly angle $f$ geometrically represents the angle between the line of
nodes and the position vector $\vec{r}$ on the orbital plane. The angle $f$ $%
\ $is related to the motion of the satellite and is assumed to be the only
changing dynamical parameter meaning that the space-distributed orbit does
not change during the motion of the satellite. Under such assumptions and
taking into account that $dr=\alpha (f)df$ (where $\alpha (f)\equiv -\frac{%
a(1-e^{2})e\sin f\text{ }df}{(1+e\cos f)^{2}}$), in the paper \cite{Bog1}
the propagation time $\widetilde{T}$ in the null cone equation (\ref{DOP25}%
)\ has been expressed by a complicated expression \ 
\begin{equation}
\widetilde{T}=G(f)=-\frac{2G_{\oplus }M_{\oplus }}{c^{3}}.\frac{n}{%
(1-e^{2})^{\frac{3}{2}}}\widetilde{T}_{1}+T_{2}^{(2)}\text{ \ , }  \label{K4}
\end{equation}%
where $\widetilde{T}_{1}$ and $T_{2}^{(2)}$ is a combination of elliptic
integrals of the second and of the fourth order. The exact expressions can
be found in the review article \cite{Bog2}, but also will be given in
Appendix B of this paper. Further in this paper they will participate in the
expressions in the proposed variational formalism, based on the
Gauss-Ostrogradsky and the Stokes theorems, but their concrete analytical
forms will not be taken into account. Formulae (\ref{K4}) is the propagation
time for the signal, while the true anomaly angle $f$ changes from some
value $f_{0}$ to some final value $f_{1}$. In other words, this formulae
does not presume that there is any transmission of signals between a
signal-emitting and a signal-receiving satellite because the second point of
signal perception is not related to any satellite, it is determined only by
the distance, travelled by the signal for the calculated propagation time $%
\widetilde{T}$. Formulae (\ref{K4}) expresses the propagation time $%
\widetilde{T}$ as a function only of the dynamical parameter, which is the
true anomaly angle $f$ $\ $\cite{Bog1}.

This means also that all the other (five) remaining parameters of the orbit $%
(\Omega ,\omega ,i,a,e)$ do not change during the motion of the satellite.
For such a case, it can easily be found that

\begin{equation}
\sqrt{(dx)^{2}+(dy)^{2}+(dz)^{2}}=\sqrt{\left( v_{f}^{x}\right) ^{2}+\left(
v_{f}^{y}\right) ^{2}+\left( v_{f}^{z}\right) ^{2}}df=v_{f}df\text{ \ \ ,}
\label{K5}
\end{equation}

where the velocity $v_{f}$, associated to the true anomaly angle $f$ $\ $is
given by

\begin{equation}
v_{f}=v=\frac{na}{\sqrt{1-e^{2}}}\sqrt{1+e^{2}+2e\cos f}\text{ \ \ \ \ .}
\label{K6}
\end{equation}

Making use of the null cone equation (\ref{DOP25}) and also of the
approximation $\beta =\frac{2V}{c^{2}}=\frac{2G_{\oplus }M_{\oplus }}{c^{2}a}%
\ll 1$, one can obtain the general formulae for the propagation time in the
form (\ref{K4}).

\subsection{The other not yet investigated case of transmission of signals
(with account of the General Relativity effects) between satellites on
different space-distributed orbits}

The more interesting case will be when a signal is emitted by a satellite on
a space-distributed orbit and is received (percepted) by a satellite on
another space-distributed orbit. Then there will be two dynamical parameters 
$f_{1}$ and $f_{2}$, related to the motion of both satellites. The
integration in the expression for the propagation time will be over the
infinitesimal distance element $dR_{AB}$ between the two satellites, since
it involves the differentials of both sets of cartesian coordinates $%
(x_{1},y_{1},z_{1})$ and $(x_{2},y_{2},z_{2})$ or the two sets of Kepler
parameters $(f_{1},\Omega _{1},\omega _{1},i_{1},a_{1},e_{1})$ and $%
(f_{2},\Omega _{2},\omega _{2},i_{2},a_{2},e_{2})$ for the transformations (%
\ref{K1}) - (\ref{K3}), written respectively for the indices "$1$" and "$2$".

\subsection{Autonomous navigation as an impetus for developing
intersatellite communications}

\ In the book \cite{Shuai} autonomous navigation is defined as referring to
\textquotedblleft processes in which the spacecraft without the support from
the ground-based TT\&C (Telemetry \& Tracking Control) system for a long
time, relying on its onboard devices, obtains all kinds of measurement data;
determines the navigation parameters like orbit, time and
attitude\textquotedblright . In this monograph DORIS and PRARE navigation
systems are determined as non-autonomous, because \textquotedblleft DORIS
system can determine the spacecraft's orbit with high accuracy, but both
need to exchange information with the ground stations\textquotedblright .
Thus, the main aim of autonomous navigation is to reduce the dependence of
spacecrafts on the ground TT\&C network and in this way to enhance the
capability of the systems anti-jam and autonomous survivability. It is
evident also that the determination of the satellite's exact location
through the reception, processing, and transmitting of ranging signals
between the different satellites.

Autonomous GPS navigation system is necessary in view of the proposal of
researchers from the University of Texas, Aerospace Corporation, National Bureau 
of Standards, International Business Machines Corporation (IBM) and 
Rockwell Automation Inc. to monitor nuclear explosions, based on 
GPS inter-satellite communication link.

The development of inter-satellite laser communication systems in space -
orbit technology that enable super-high-speed data transfers at rates
greater than $1$ $Gbps$ is widely applied also in.cube-/nanosatellite
platforms such as CubeLCT, AeroCube-7B/C, CLICK, LINCS-A/B, SOCRATES and
LaserCube \cite{VISION2023}. However, in order to establish laser
communications, high performance of arc-second level pointing system is
required, and this is a difficulty for the nanosatellite platforms. It can
be supposed why there is such a difficulty - since the data-transmission
rate of $1Gbps$ between the two nanosatellites is at an intersatellite range
of $1000$ $km$, the effects of curving the trajectory of the laser signal
may be considerable. In publications, related to small satellite optical
links \cite{Zaman2020} it has been admitted that the s.c. "pointing error"
arizes not only due to tracking sensors and mechanical vibrations, but also
due to the base motion of the satellite. This fact, together with the
ranging of signals and the pointing error illustrate the assumption that not
only the motion of two satellites is important, but the effect of curving of
the signal trajectory. The paper \cite{Zaman2020} also asserts that for a
GPS constellation with 24 satellites, there are a total of $8-16$ links - 
they can be forward and backward links in the same orbit, but as
well as lateral links between adjacent orbits. Moreover, the distance 
of GPS inter-satellite crosslink can reach $49$ $465km$.

\section{Overview of the aims and approaches of this paper}

The main goal of this paper is to present a theoretical model for
inter-satellite communications between two satellites on different
space-distributed orbits, characterized by two different sets and Kepler
parameters, accounting for the General Relativity effect of curving the
signal trajectory. In future, the aim will be also to find the optimal
position of the satellites for transmission of signals with a minimal
propagation time. From an experimental point of view the importance of
developing such approaches is described in the introduction in view of the
rapidly developing laser broadband and large capacity communication links in
free space, especially image transmissions.

The peculiarity of the proposed approach is that while the first
publications on the implementation of General Relativity in GPS
communications focused on space-ground links for the most simple case of
plane Kepler orbits (characterized only by the eccentric anomaly angle $E$
and constant ellipticity $e$ , here the problem is more complicated and is
related to the communication between satellites on two different orbits.
Each orbit is characterized by six dynamically changing Kepler parameters.

The main proposal in this paper is to implement ideas from differential
geometry and topology for the solution of this problem. Thus it may be
asserted that the problem is also non-trivial from a mathematical point of
view. Namely, if the motion of the first satellite and the emission of the
signal is characterized by the parameters $(f_{1},\Omega _{1},\omega
_{1},i_{1},a_{1},e_{1})$, defining a manifold $M_{1}$, and the motion of the
second satellite and the processes of perception of the signal - by the
parameters $(f_{2},\Omega _{2},\omega _{2},i_{2},a_{2},e_{2})$, defining a
second manifold $M_{2}$, then can the combined motion of the two satellites
and the processes of signal-exchanges be characterized by one manifold? It
will be shown that it is possible to introduce such a manifold,
characterized by the two sets of (total) $12$ Kepler parameters, but an
additional parameter $s$ should be introduced. This turned out to be
possible due to the fulfillment of the Whitney theorem in differential
geometry. This parameter parametrizes the two sets of $12$ Kepler
parameters, which enter the expressions for the two sets of cartesian
coordinates $(x_{1},y_{1},z_{1})$ and $(x_{2},y_{2},z_{2})$.

This theoretical setting clearly demonstrates the necessity to implement
differential geometry and topology approaches in this problem about
intersatellite communications. That is why, particular attention in the
paper is paid to notions such as topological and metric spaces and
manifolds, homeomorphisms (topological equivalence), manifolds with
boundaries. It has been shown also that the standard Shapiro delay formulae
can be modified, and thus the modified formulaes (\ref{Obz10}) and (\ref%
{Obz11}) have been obtained. The boundaries of the manifold $M_{1}$ were
found by applying the concrete definitions. Since the manifold is a
three-dimensional one, the boundaries were found to be three two-dimensional
planes of the variables $(x,y)$, $(x,z)$ and $(y,z).$

Further the aim was to develop a variational formalism, based on the Stokes
and the Gauss - Ostrogradsky theorem and accounting for higher-order
variations. The corresponding variational action was found with derivatives
of the first, second and third order. The idea about higher-order
variational calculus is not new and it comes from older books such as \cite%
{GelfandFomin} and many others. In particular, second-order variations allow
to take into account non-fixed, "moving boundaries" in the corresponding
integrals. This is convenient in view of the fact that these boundaries will
be related to the motion of the two satellites and further to the necessity
of "optimizing" the process of transmission of signals. The variational
model is constructed by means of variation of the parameters $u_{i}$ and $v_{j}$%
, denoting two different subsets of Kepler parameters of one satellite 
\begin{equation}
u_{i}=(f_{.},a_{.},e_{.})\text{ \ \ }i=1,2,3\text{ and }v_{j}=(\Omega
_{.},i_{.},\omega _{.})\text{ \ }j=1,2,3)  \label{K7}
\end{equation}

or the two different sets of Kepler parameters, corresponding to the first
and to the second satellite respectively 
\begin{equation}
u_{i}=(f_{.1},a_{1.},e_{1},_{.}\Omega _{1.},i_{1.},\omega _{1})\text{ \ \ }%
i=1,2,....4,5,6\text{ }  \label{K8}
\end{equation}%
and 
\begin{equation}
v_{j}=(f_{2.},a_{2.},e_{2},\Omega _{2.},i_{2.},\omega _{2.})\text{ \ }%
j=1,2,...,4,5,6\text{ \ \ .}  \label{K9}
\end{equation}

\section{Some basic notions from topology, needed for the development of the
theory in this paper}

This study makes use of some standard notions from differential geometry,
but also from topology \ These two branches from geometry are of course
related, but it is important to keep in mind that there is an important
difference. For example, a topological manifold is a set of points of the
Euclidean space, for which the notion of proximity is defined \cite%
{ShapOlshan}. It should be noted, however, that proximity is not related to
properties of metric spaces, which shall also be defined further. To be more
precise, a topological manifold is the set of points of the Euclidean space,
defined by the equations 
\begin{subequations}
\begin{equation}
F_{i}(x_{1},x_{2},......x_{\overline{n}})=0\text{ \ \ \ , }i=1,2,....m\text{
\ \ \ ,}  \label{Obz1}
\end{equation}%
\end{subequations}
\begin{subequations}
\begin{equation}
f_{j}(x_{1},x_{2},......x\overline{_{n}})\geq 0\text{ \ \ \ , }j=1,2,....l%
\text{ \ \ \ .}  \label{Obz2}
\end{equation}%
This (topological) manifold, defined by the equations (\ref{Obz1}) is a
sub-manifold of a dimension $n=\overline{n}-m$. The notion of a sub-manifold
is also defined as a topological property, and not related to any metric
properties. Another typical notions for topological manifolds are
boundedness, connectedness,compactness, continuity. The notions of a
sub-manifold and a boundary,as further will be explained, are also not
related to any metric properties.

It is evident that topological spaces are generalizations of metric spaces
and as such they are defined as a class of open sets of $E$ (the Euclidean
space), satisfying the following conditions \cite{SchwarzPhys}:

1. The intersection of a finite number of open sets is an open set $X$.

2. The union of any number of open sets is an open set. The class of open
sets is called a topology.

These two conditions are supplemented by a third one in \cite{DubrFomVol2}:
The empty set and the whole set $X$ must be open.

The definition about topological space by itself is sufficient to define the
important notions of continuity and homeomorphicity.

A map $f:X\longrightarrow Y$ of one topological space to another is
continuous if the complete inverse image $f^{-1}(U)\subset X$ of every open
set $U\subset Y$ is open also in $X$ \cite{DubrFomVol2}. Two topological
spaces $E$ and $F$ are homeomorphic (or topologically equivalent) if there
is a one-to-one mapping $\varphi :E\longrightarrow F$ of $E$ onto $F$ such
that both $\varphi $ and the inverse one $\varphi ^{-1}$ are continuous (see
also \cite{SchwarzPhys}). A simple application of homeomorphicity, given in 
\cite{SchwarzPhys} is the following one: Since the hemisphere is
homeomorphic to a disc, two pairs of opposite points on the boundary can be
glued together and thus the $n-$dimensional projective plane $RP^{n}$ can be
obtained. This can be illustrated by the following simple formulae 
\end{subequations}
\begin{equation}
R^{n+1}\backslash \{0\}\rightarrow S^{n}\rightarrow RP^{n}\text{ \ \ .}
\label{Obz2A}
\end{equation}%
This definition makes use of the topological notion about "manifolds with
boundary", which shall be clarified further. Instead of the term "gluing",
often the term "identifying" is used. For example, in order to visualize $%
RP^{2}$, one may take the two-dimensional sphere $S^{2}$ and identify pairs
of opposite points below the equator and above the equator.

Now an interesting property of a topological space $M$ can be defined,which
enables to define a differentiable manifold \cite{NashSen}. If

A. The topological space $M$ is provided with a family of pairs $\left\{
(M_{\alpha },\Phi _{\alpha })\right\} $, where $\Phi _{\alpha }$ are
homeomorphisms from $M_{\alpha }$ to an open subset of $O_{\alpha }$ of $%
R^{n}$ i.e. $\Phi _{\alpha }:M_{\alpha }\mapsto O_{\alpha \text{, }}$and $%
M_{\alpha }$ are a family of open sets, which cover $M:\cup M_{\alpha }=M$.

B. Given $M_{\alpha },M_{\beta }$ such that $M_{\alpha }\cap M_{\beta }\neq
\emptyset $, the map $\Phi _{\beta }\circ \Phi _{\alpha }^{-1}$ from the
subset $\Phi _{\alpha }(M_{\alpha }\cap M_{\beta })$ of $R^{n}$ to the
subset $\Phi _{\beta }(M_{\alpha }\cap M_{\beta })$ of $R^{n}$ is infinitely
differentiable (written as $C^{\infty }$),

then the family \ $\left\{ (M_{\alpha },\Phi _{\alpha })\right\} $
satisfying A. and B. is called an atlas.

An important feature of the topological space $M$ is the Hausdorff property,
which however should be additionally imposed and is independent from the
metric space. The Hausdorff property turns out to be important, if one would
like to determine the dimension of $M$, in view of the following

\begin{definition}
\cite{Tu}: A topological manifold is a Hausdorff, second countable (meaning
that it has a countable basis) and locally Euclidean space.
\end{definition}

Here the two important notions of Hausdorff space and locally Euclidean
space have to be clarified. The Hausdorff property \cite{NashSen} means that
if there are two distinct points $x,y$ belonging to the topological space $M$%
, then there exists a pair of open sets $O_{x}$ and $O_{y}$, such that $%
O_{x} $ $\cap $ $O_{y}=\emptyset $. This means also that there exist small
enough disjoint open sets $O_{x}$ and $O_{y}$, which contain $x$ and $y$
respectively. The respective definition, given in \cite{Tu} is analogous.

Further, a topological space $M$ is locally Euclidean of dimension $n$ if
every point $p$ in $M$ has a neighborhood $U$ such that there is a
homeomorphism $\Phi $ from $U$ onto an open subset of $R^{n}$ \cite{Tu}. \
The pair $(U,\Phi :U\rightarrow R^{n})$ is called a chart, $\Phi $-a
coordinate map or a coordinate system on $U$ \ and $U$- a coordinate
neighborhood or a coordinate open set. \ This definition is the essence of
the s.c. invariance of dimension, according to which for $n\neq m$ an open
subset of $R^{n}$ is not homeomorphic to an open set of $R^{m}$. This topic
will be essential for the topological notions of embedding, submersion and
immersion, which shall further be explained with reference to the different
mappings $f:(f,\Omega ,\omega ,i,a,e)\rightarrow (x,y,z)$ which can be
defined, concerning the transformations (\ref{K1}) - (\ref{K3}). In
particular, these mappings will depend on the number of the Kepler elements,
which will change during the process of transmission of a signal between two
satellites on two space-dependent orbits.

Unlike the topological \ notion of \ a homeomorphism, which generates
equivalence classes between topological spaces, the notion of a metric space
is a particular kind of a topological space, where the open sets are
provided by a distance function $d(x,y)$, which corresponds to our intuitive
notion for the distance between points. Of course, $d(x,y)$ can be
generalized also to the $3D-$case. The distance function $d(x,y)$ has the
following properties \cite{NashSen}, \cite{Tu} and is expressed by the
simple formulae \ $d(x,y)=\sqrt{(x_{1}-x_{2})^{2}+(y_{1}-y_{2})^{2}}$, for
the $n-$dimensional space $d(p,q)=\sqrt{\sum%
\limits_{i=1}^{n}(p^{i}-q^{i})^{2}}$:

1A. $d(x,y)\geq 0$ .

2A. $d(x,y)=0\leftrightarrow x=y$.

3A.If the point $z\in M$, then the triangle inequality $d(x,z)\leq
d(x,y)+d(y,z)$ is fulfilled. A topological space, satisfying these three
conditions is a metric space. In addition, if $M$ is also\ a Riemannian
manifold, then the distance is $ds^{2}=g_{ij}dx^{i}dx^{j}$ ($i,j=1,2....4)$
and the points $a,b$ are connected by a piecewise differentiable curve $C$,
parametrized by the functions $x^{i}(s)$, then the distance $L_{ab}$ between
the points will be defined as 
\begin{equation}
L_{ab}=\int\limits_{a}^{b}g_{ij}\frac{dx^{i}}{ds}\frac{dx^{j}}{ds}ds\text{ \
.}  \label{Obz3}
\end{equation}%
Then the distance function in the above definition $d(x,y)$\ will be defined
as the infimum or least upper bound of $L_{ab}$ as it varies over all such
curves $x^{i}(s)$ i.e. $d(a,b)=\inf L_{ab}(C)$.

An important clarification will be given below, which is related to the
present research.

\subsection{Application of metric and topological spaces in the present
formalism}

The analogy of the distance function $d(x,y)$ is the introduced in Appendix
A formulae (\ref{A00}) and formulae (\ref{A1}) for the differential $%
dR_{AB}^{2}$ for the case, when $f_{1}$ and $_{f_{2}}$ are the only
dynamical parameters. Since formulae (\ref{A00}) can be written also as 
\begin{equation}
R_{AB}^{2}=r_{1}^{2}+r_{2}^{2}-2(\overrightarrow{\mathbf{r}}_{1}.%
\overrightarrow{\mathbf{r}}_{2})\text{ \ \ \ ,}  \label{Obz4}
\end{equation}%
where the vectors $\overrightarrow{\mathbf{r}}_{1}$ and $\overrightarrow{%
\mathbf{r}}_{2}$ are defined in the natural way , $\overrightarrow{\mathbf{r}%
}_{1}=(x_{1},y_{1},z_{1})$ and $\overrightarrow{\mathbf{r}}%
_{2}=(x_{2},y_{2},z_{2})$, the following problem arizes: suppose that for
the two topological manifolds one can define the two open balls $%
B_{1}(p_{1},r_{1})$ and $B_{2}(p_{2},r_{2})$ with centers at the points $%
p_{1}\in R_{1}^{3}$ and $p_{2}\in R_{2}^{3}$ (the lower indices "$1$" and "$%
2 $" mean that these are two different Euclidean manifolds) such that 
\begin{equation}
B_{1}(p_{1},r_{1})=\{x_{1},y_{1},z_{1}\in R_{1}^{3}\mid
d(x_{1},y_{1},z_{1},p_{1})<r_{1}\}\text{ \ ,}  \label{Obz5}
\end{equation}%
\begin{equation}
B_{2}(p_{2},r_{2})=\{x_{2},y_{2},z_{2}\in R_{1}^{3}\mid
d(x_{2},y_{2},z_{2},p_{2})<r_{2}\}\text{. }  \label{Obz6}
\end{equation}%
Then, in order to apply the definition for an open set in the meaning of an
open ball (see \cite{Tu}), can one define the open ball $B(p,r)=\{x,y,z\in
R_{1}^{3}\times R_{2}^{3}\mid d(x,y,z,p)<R_{AB}\}$? This problem is outside
the scope of the present paper, but should be investigated also for the
general case when the two manifolds $R_{1}^{3}$ and $R_{2}^{3}$ are
expressed in terms of the two sets of Kepler parameters $(f_{1},\Omega
_{1},\omega _{1},i_{1},a_{1},e_{1})$ and $(f_{2},\Omega _{2},\omega
_{2},i_{2},a_{2},e_{2})$.

\subsection{A new proposition for a generalized Shapiro delay formulae,
based on the differential geometry formalism}

\subsubsection{The additional condition for equality of the two differentials%
}

Using the definition about the metric space, let us derive the general
formulae for the differential of the propagation time $T$ for the signal
between two satellites on two different space-distributed orbits.
Remembering the null cone equation (\ref{DOP25}) and formulae (26) from the
review paper \cite{Bog2}, one can obtain 
\begin{equation}
dT\approx \frac{1}{c}\int\limits_{path}\sqrt{\delta _{ij}dx^{i}dx^{j}}+\frac{%
1}{c^{3}}\int\limits_{path}2U\sqrt{\delta _{ij}dx^{i}dx^{j}}\text{ \ \ ,}
\label{Obz7}
\end{equation}%
where the integration is performed over the s.c. "slant range" $d\rho =\sqrt{%
\delta _{ij}dx^{i}dx^{j}}$, called also the "geometric path". The Earths
gravitational potential $U$ will be taken in the standard form $U=\frac{%
G_{\oplus }M_{\oplus }}{R}$. But for this case, we may define the as the
infinitesimal geometric path $dR_{AB}$ on the line, joining the
signal-emitting and the signal-receiving satellite. Further argumentation
will be provided in favour of the statement that imposing the equality $%
d\rho =dR_{AB}$ is very convenient.

If the coordinates $x,y,z$ are parametrized by the parameter $s$, i.e. $%
x=x(s)$, $y=y(s)$ and $z=z(s)$ and remembering the previously introduced
notations $x=x_{2}-x_{1}$, $y=y_{2}-y_{1}$ and $z=z_{2}-z_{1}$, one can
represent $d\rho $ as 
\begin{equation}
d\rho =\sqrt{\left( \overset{.}{x}\right) ^{2}+\left( \overset{.}{y}\right)
^{2}+\left( \overset{.}{z}\right) ^{2}}ds  \label{Obz8}
\end{equation}%
and 
\begin{equation*}
dR_{AB}^{2}=2(x_{2}-x_{1})(dx_{2}-dx_{1})+2(y_{2}-y_{1})(dy_{2}-dy_{1})+2(z_{2}-z_{1})(dz_{2}-dz_{1})=
\end{equation*}%
\begin{equation}
=\frac{1}{R_{AB}}\left[ xdx+ydy+zdz\right] =\frac{1}{2R_{AB}}\frac{d}{ds}%
\left( x^{2}+y^{2}+z^{2}\right) =\frac{1}{2R_{AB}}\frac{dR_{AB}^{2}}{ds}=%
\frac{dR_{AB}}{ds}\text{ \ \ . \ }  \label{Obz9}
\end{equation}%
The last simple formulae $dR_{AB}^{2}=\frac{dR_{AB}}{ds}$ is related to the
fact that the infinitesimal Euclidean distance $dR_{AB}^{2}$ is equal to the
infinitesimal distance $\frac{dR_{AB}}{ds}$ along the curve $R_{AB}(s)$,
parametrized by the parameter $s$. In fact, formulae (\ref{Obz9}) can be
written in the same way as in the monograph by Postnikov \cite{Postnikov3}
(see also \cite{DubrFomVol1})%
\begin{equation}
l(s)=\int\limits_{s_{0}}^{s}\mid \overset{.}{\overrightarrow{\mathbf{R}}}%
_{AB}(s)\mid ds\text{ \ \ \ ,}  \label{Obz9B}
\end{equation}%
where the coordinates of the vector $\overset{.}{\overrightarrow{\mathbf{R}}}%
_{AB}(s)$ are $\overset{.}{\overrightarrow{\mathbf{R}}}_{AB}(s)=(\overset{.}{%
x}(s),\overset{.}{y}(s),\overset{.}{z}(s))$ and thus $\mid \overset{.}{%
\overrightarrow{\mathbf{R}}}_{AB}(s)\mid =\sqrt{\overset{.}{x}^{2}(s)+%
\overset{.}{y}^{2}(s)+\overset{.}{z}^{2}(s)}$. It is known that the
parameter $s$ in the general case is not a natural parameter, meaning that
it is not equal to the length of the curve. It would have been a natural
parameter only if $\mid \overset{.}{\overrightarrow{\mathbf{R}}}_{AB}(s)\mid
=1$ and thus $l(s)=s-s_{0}$. This basic fact from the differential
geometry of curves and surfaces will be of primary importance, when further
some theorems about the geodesic and the null cone equations will be given.

Now one can see the subtle difference - in the infinitesimal sense the
distance between two points along the curve is the same as the distance
between the two points on the straight line, which connects the two points.
In particular, this is the simple intuitive reasoning for requiring the
fulfillment of the equality 
\begin{equation}
d\rho =\sqrt{\delta _{ij}dx^{i}dx^{j}}=dR_{AB}\text{ \ \ \ ,}  \label{Obz9A}
\end{equation}%
which can be written also as 
\begin{equation}
\sqrt{\left( \overset{.}{x}\right) ^{2}+\left( \overset{.}{y}\right)
^{2}+\left( \overset{.}{z}\right) ^{2}}ds=\frac{1}{2R_{AB}}%
d(x^{2}+y^{2}+z^{2})=\text{ \ \ \ }  \label{Obz10}
\end{equation}%
\begin{equation}
=dR_{AB}^{2}=2r_{1}dr_{1}+2r_{2}dr_{2}-2(\overrightarrow{\mathbf{dr}}_{1}.%
\overrightarrow{\mathbf{r}}_{2})-2(\overrightarrow{\mathbf{r}}_{1}.%
\overrightarrow{d\mathbf{r}}_{2})\text{ \ \ ,}  \label{Obz11}
\end{equation}%
where expression (\ref{Obz4}) for $dR_{AB}^{2}$ has been used. Note also
that formulae (\ref{Obz7}) is not necessarily to be written on the base of
the assumption (\ref{Obz9A}) $d\rho =dR_{AB}$ - at the end of this section
such a possibility will also be discussed.

\subsubsection{Modified Shapiro delay formulae with the additional condition}

It should be kept in mind that $2r_{1}dr_{1}$ and $2r_{2}dr_{2}$ are scalar
multiplication of functions, while $(\overrightarrow{\mathbf{dr}}_{1}.%
\overrightarrow{\mathbf{r}}_{2})$ and $(\overrightarrow{\mathbf{r}}_{1}.%
\overrightarrow{d\mathbf{r}}_{2})$ are vector multiplications of two
vectors, for example 
\begin{equation}
(\overrightarrow{\mathbf{dr}}_{1}.\overrightarrow{\mathbf{r}}%
_{2})=x_{2}dx_{1}+y_{2}dy_{1}+z_{2}dz_{1}\text{ \ , }(\overrightarrow{%
\mathbf{r}}_{1}.\overrightarrow{d\mathbf{r}}%
_{2})=x_{1}dx_{2}+y_{1}dy_{2}+z_{1}dz_{2}\text{\ \ \ }  \label{Obz12}
\end{equation}%
and the two set of functions $(x_{1},y_{1},z_{1})$ and $(x_{2},y_{2},z_{2})$
are given by formulaes (\ref{K1}) - (\ref{K3}), supplemented by the
corresponding indices "$1$" and "$2$". A known property from celestial
mechanics (see for example \cite{KopeikinBook} and \cite{Gurfil}) is that
the scalar functions $r_{1}$ and $r_{2}$ represent the corresponding
Euclidean distances 
\begin{equation}
x_{1}^{2}+y_{1}^{2}+z_{1}^{2}=r_{1}^{2}\text{ }\ \ \text{,}\ \text{\ }%
x_{2}^{2}+y_{2}^{2}+z_{2}^{2}=r_{2}^{2}\text{ \ \ \ ,}  \label{Obz13}
\end{equation}%
which depend on the corresponding true anomaly angles $f_{1}$ and $f_{2}$
and also on $(e_{1},a_{1})$ and $(e_{2},a_{2})$ 
\begin{equation}
r_{1}=\frac{a_{1}(1-e_{1}^{2})}{1+e_{1}\cos f_{1}}\text{ \ \ , \ \ \ \ }%
r_{2}=\frac{a_{2}(1-e_{2}^{2})}{1+e_{2}\cos f_{2}}\text{ \ \ .\ }
\label{Obz14}
\end{equation}%
If all the orbital elements (i.e. the two sets of orbital parameters) $%
(f_{1},\Omega _{1},\omega _{1},i_{1},a_{1},e_{1})$ and $(f_{2},\Omega
_{2},\omega _{2},i_{2},a_{2},e_{2})$ are varying during the process of
emission and perception of the signal between the two satellites (i.e.all
these two sets of Kepler elements depend on the parameter $s$ and thus \
represent functions), this has also to be taken in account when calculating
the differentials $dr_{1}$ and $dr_{2}$ in formulae (\ref{Obz11}) and also $(%
\overrightarrow{\mathbf{r}}_{1}.\overrightarrow{d\mathbf{r}}_{2})$ and $(%
\overrightarrow{\mathbf{dr}}_{1}.\overrightarrow{\mathbf{r}}_{2})$ ). We
shall not present in details these lengthy calculations, these formulaes and
the further development of  the new approach will be given in another
publication. .

However, when an integration along the parameter $s$ is performed along a
curve, joining two non-infinitesimal points, corresponding to the parameter
values $s_{1}$ and $s_{2}$, the integration will not give an "infinitesimal
result". This would mean that the first term in expression (\ref{Obz7}) for
the differential of the propagation time $dT$ 
\begin{equation}
dT\approx \frac{1}{c}\int\limits_{path}\sqrt{\left( \overset{.}{x}\right)
^{2}+\left( \overset{.}{y}\right) ^{2}+\left( \overset{.}{z}\right) ^{2}}ds+%
\frac{2G_{\oplus }M_{\oplus }}{c^{3}}\int\limits_{R_{0A}}^{R_{0B}}\frac{dR}{%
R}\text{ \ \ }  \label{Obz15}
\end{equation}%
will no longer give a constant Euclidean distance, divided by $c$, as is in
the standard Shapiro formulae. In (\ref{Obz15}) we have denoted the
under-integral variable in the second integral as \ $d\overline{\rho }=\sqrt{%
\delta _{ij}dx^{i}dx^{j}}ds=dR$ and the path integral will be from $R_{01}$
to $R_{02}$ ($R_{01}$ and $R_{02}$ are the distances from the first point
and the second point to the zero-point, the origin of the coordinate system
and thus $R_{AB}=R_{0B}-R_{0A}$) 
\begin{equation}
\frac{2G_{\oplus }M_{\oplus }}{c^{3}}\int\limits_{R_{0A}}^{R_{0B}}\frac{dR}{%
R}=\frac{2G_{\oplus }M_{\oplus }}{c^{3}}\ln (R_{0B}-R_{0A})=\frac{2G_{\oplus
}M_{\oplus }}{c^{3}}\ln (R_{AB})\text{ \ \ \ .\ \ }  \label{Obz16}
\end{equation}
This second term might be considered to be a generalization of the second
term in the standard Shapiro formulae $T_{AB}=\frac{R_{AB}}{c}+\frac{2GM_{E}%
}{c^{3}}\ln \left( \frac{r_{A}\newline
+r_{B}+R_{AB}}{r_{A}+r_{B}-R_{AB}}\right) \ \ $(\ref{AA19}), but this should
be understood in the most general sense. The reason is that formulae (\ref%
{Obz7}) is for the differential $dT$, which depends on the two sets of
Euclidean coordinates (in total $6$ coordinates-or in the case, on the
coordinates $x=x_{2}-x_{1}$, $y=y_{2}-y_{1}$) or one might also consider a
case, when $dT$ will depend on the two sets of Kepler parameters. Moreover,
the formulae (\ref{Obz15}) is valid when the condition (\ref{Obz9A}) is also
considered, and it represents a complicated quadratic expression in the
differentials $dx^{i}$ and $dx^{j}$.

\subsubsection{Modified Shapiro delay formulae without any additional
conditions}

Equation (\ref{Obz7}) can be written also without applying the condition $%
d\rho =\sqrt{\delta _{ij}dx^{i}dx^{j}}=dR_{AB}$ (\ref{Obz9A}).Then
expression (\ref{Obz15}) will be rewritten as 
\begin{equation}
dT\approx \frac{1}{c}\int\limits_{path}\sqrt{\left( \overset{.}{x}\right)
^{2}+\left( \overset{.}{y}\right) ^{2}+\left( \overset{.}{z}\right) ^{2}}ds+%
\frac{2G_{\oplus }M_{\oplus }}{c^{3}}\int \frac{\sqrt{\left( \overset{.}{x}%
\right) ^{2}+\left( \overset{.}{y}\right) ^{2}+\left( \overset{.}{z}\right)
^{2}}}{R_{AB}}dR_{AB}\text{ \ \ . \ \ }  \label{Obz17}
\end{equation}%
Now making use of the formulaes (\ref{Obz11}) and (\ref{Obz12}), it can
easily be proved that 
\begin{equation}
dR_{AB}=\frac{1}{R_{AB}}\left[ xdx+ydy+zdz\right] \text{ }  \label{Obz18}
\end{equation}%
Thus, the whole expression (\ref{Obz17}) can be written as 
\begin{equation*}
dT\approx \frac{1}{c}\int\limits_{path}\sqrt{\left( \overset{.}{x}\right)
^{2}+\left( \overset{.}{y}\right) ^{2}+\left( \overset{.}{z}\right) ^{2}}ds
\end{equation*}
\begin{equation}
+\frac{2G_{\oplus }M_{\oplus }}{c^{3}}\int \frac{\sqrt{\left( \overset{.}{x}%
\right) ^{2}+\left( \overset{.}{y}\right) ^{2}+\left( \overset{.}{z}\right)
^{2}}}{\left[ x^{2}+y^{2}+z^{2}\right] }\left[ xdx+ydy+zdz\right] \text{ \ \
. \ \ }  \label{Obz19}
\end{equation}%
This is the expression in the general case. If the parametrization $x=x(s)$, 
$y=y(s)$ and $z=z(s)$ is taken into account (also $dT=\int dTds$), then the
propagation time $T=T_{2}-T_{1}$ will be given by an integral of the type $%
T\approx \int F(x,y,z,\overset{.}{x},\overset{.}{y},\overset{.}{z})ds$,
where the under-integral function acquires the form 
\begin{equation*}
F(x,y,z,\overset{.}{x},\overset{.}{y},\overset{.}{z})=\frac{1}{c}\sqrt{%
\left( \overset{.}{x}\right) ^{2}+\left( \overset{.}{y}\right) ^{2}+\left( 
\overset{.}{z}\right) ^{2}}
\end{equation*}
\ 
\begin{equation}
+\frac{2G_{\oplus }M_{\oplus }}{c^{3}}\frac{\sqrt{\left( \overset{.}{x}%
\right) ^{2}+\left( \overset{.}{y}\right) ^{2}+\left( \overset{.}{z}\right)
^{2}}}{\left[ x^{2}+y^{2}+z^{2}\right] }\left[ x\overset{.}{x}+y\overset{.}{y%
}+z\overset{.}{z}\right] \text{ \ \ . \ \ }  \label{Obz20}
\end{equation}

The first integral in (\ref{Obz19}) is known from differential geometry \cite%
{DubrFomVol1}, \cite{DubrFomVol2}, \cite{MishtFomCourse}, \cite%
{PrasolovDifferential} (and many others) and its relation to the physical
essence of the propagation of signals will be investigated in the next
section, related also to some basic theorems from differential geometry.

The second term in (\ref{Obz20}) is more complicated and it will be
investigated in another paper, since it will be interesting to include this second 
term in the usual variational formalism, when the action functional derived from the null cone equation 
will be minimized and the Euler Lagrange equations with account of this second term have to be solved. 
Then it will be interesting to compare the results from the usual variational formalism with the results 
from the higher-order variational formalism, which will be proposed in the following sections of 
this paper. 

\section{Minimization of the first term in the action functional for the
propagation time}

\subsection{Mathematical derivation of the extremal function - the straight
line}

Now we shall make the first step towards solving the following physical
problem: suppose that a satellite on one space-distributed orbit sends a
signal towards another satellite (also on a space-distributed orbit). Then
what is the optimal configuration (i.e. disposition of the satellites), so
that the transmission will be for the least time.

This problem requires the minimization of the functional (\ref{Obz19}), in
particular we shall investigate the minimization only of the first term.
Since (\ref{Obz19}) originally has been obtained by the null cone equation (%
\ref{DOP25}), this minimization will require the implementation of a
variational principle, based on the Euler-Lagrange equations. So in this
section we shall minimize the action functional 
\begin{equation}
T=\int \left[ \sqrt[.]{\left( \overset{.}{x}\right) ^{2}+\left( \overset{.}{y%
}\right) ^{2}+\left( \overset{.}{z}\right) ^{2}}+\lambda (x,y,z)R_{AB}(x,y,z)%
\right] ds\text{ \ \ ,}  \label{Obz21A}
\end{equation}
where we have denoted $x=x_{1}-x_{2}$, $y=y_{1}-y_{2}$, $z=z_{1}-z_{2}$, $\ $%
supplemented by the Lagrange multiplier $\lambda (x,y,z)$ and the distance
between the space points $(x_{1},y_{1},z_{1})$ and $(x_{2},y_{2},z_{2})$ is $%
R_{AB}=\sqrt[.]{(x_{1}-x_{2})^{2}+(y_{1}-y_{2})^{2}+(z_{1}-z_{2})^{2}}$ as a
constraint function. Additionally we introduce the notation for the length
function $L_{1}$ 
\begin{equation}
L_{1}=\sqrt[.]{\left( \overset{.}{x}\right) ^{2}+\left( \overset{.}{y}%
\right) ^{2}+\left( \overset{.}{z}\right) ^{2}}\text{ \ ,}  \label{Obz22}
\end{equation}
where, as usual, the dot "$.$" denotes a differentiation along the parameter 
$s$. The length function $L_{1}$ measures the
distance between the two points along the arc of two points on a given
curve, not necessarily assuming that $s$ is a "natural parameter".

The action functional will have a minimum if and only if the under-integral
function satisfies the Euler-Lagrange equations 
\begin{equation}
\frac{\delta L}{\delta x}=0=\lambda _{x}R_{AB}+\lambda \frac{x}{R_{AB}}%
+\left( \frac{\overset{..}{x}}{L_{1}}-\frac{1}{L_{1}^{3}}\frac{d}{ds}\frac{%
\overset{.}{x}^{3}}{3}\right) \text{ \ \ \ \ \ \ \ ,}  \label{Obz23}
\end{equation}%
\begin{equation}
\frac{\delta L}{\delta y}=0=\lambda _{y}R_{AB}+\lambda \frac{y}{R_{AB}}%
+\left( \frac{\overset{..}{y}}{L_{1}}-\frac{1}{L_{1}^{3}}\frac{d}{ds}\frac{%
\overset{.}{y}^{3}}{3}\right) \text{ \ \ \ \ \ \ \ ,}  \label{Obz24}
\end{equation}%
\begin{equation}
\frac{\delta L}{\delta z}=0=\lambda _{z}R_{AB}+\lambda \frac{z}{R_{AB}}%
+\left( \frac{\overset{..}{z}}{L_{1}}-\frac{1}{L_{1}^{3}}\frac{d}{ds}\frac{%
\overset{.}{z}^{3}}{3}\right) \text{ \ \ \ \ \ \ \ .}  \label{Obz25}
\end{equation}%
Multiplying the first equation by $\overset{.}{x}$, the second equation by $%
\overset{.}{y}$ and the third equation by $\overset{.}{z}$ and summing up,
one can obtain 
\begin{equation}
\left( \lambda _{.}R_{AB}\right) _{s}+\frac{dL_{1}}{ds}-\frac{1}{L_{1}^{3}}%
\left[ \overset{.}{x}^{3}\overset{..}{x}+\overset{.}{y}^{3}\overset{..}{y}+%
\overset{.}{z}^{3}\overset{..}{z}\right] =0\text{ \ \ .}  \label{Obz26}
\end{equation}
From the expressions for the functions $R_{AB}$ and $L_{1}$ 
after differentiating one can obtain 
\begin{equation}
L_{1}\frac{dL_{1}}{ds}=\overset{.}{x}\overset{..}{x}+\overset{.}{y}\overset{%
..}{y}+\overset{.}{z}\overset{..}{z}\text{ , \ \ }R_{AB}\frac{dR_{AB}}{ds}%
\text{\ }=x\overset{.}{x}+y\overset{.}{y}+\text{\ }z\overset{.}{z}\text{ \ \
.}  \label{Obz27}
\end{equation}
If the first differentiated equation in (\ref{Obz27}) is multiplied by $x$ ,
the second differentiated equation - by $\overset{..}{x}$ and the resulting
equations are equated, then 
\begin{equation}
xL_{1}\frac{dL_{1}}{ds}-\overset{..}{x}R_{AB}\frac{dR_{AB}}{ds}=x\overset{.}{%
y}\overset{..}{y}+x\overset{.}{z}\overset{..}{z}-\overset{..}{x}y\overset{..}%
{y}-\overset{..}{x}z\overset{.}{z}\text{ \ ,}  \label{Obz28}
\end{equation}
which can be rewritten as 
\begin{equation}
xL_{1}\frac{dL_{1}}{ds}-\overset{..}{x}R_{AB}\frac{dR_{AB}}{ds}=x(L_{1}\frac{%
dL_{1}}{ds}-\overset{.}{x}\overset{..}{x})-\overset{..}{x}(R_{AB}\frac{%
dR_{AB}}{ds}-x\overset{.}{x})\text{ \ \ }\Longrightarrow x\overset{.}{x}%
\overset{..}{x}=0\text{ \ . }  \label{Obz29}
\end{equation}%
Similarly one can obtain the equations $y\overset{.}{y}\overset{..}{y}=0$
and $z\overset{.}{z}\overset{..}{z}=0$. From the three equations it follows 
\begin{equation}
xyz\frac{d}{ds}\left[ \left( \overset{.}{x}\right) ^{2}+\left( \overset{.}{y}%
\right) ^{2}+\left( \overset{.}{z}\right) ^{2}\right] =0\text{ \ }%
\Rightarrow L_{1}^{2}=const=C_{1}\text{\ }  \label{Obz30}
\end{equation}%
Now we rewrite the three dynamical equations (\ref{Obz23}), (\ref{Obz24})
and (\ref{Obz25}): 
\begin{equation}
\left( \lambda R_{AB}\right) _{x}+\frac{1}{C_{1}}\overset{..}{x}=0\text{ \ , 
}\left( \lambda R_{AB}\right) _{y}+\frac{1}{C_{1}}\overset{..}{y}=0\text{ \
, \ \ }\left( \lambda R_{AB}\right) _{z}+\frac{1}{C_{1}}\overset{..}{z}=0%
\text{\ }  \label{Obz31}
\end{equation}%
Multiplying the first equation by $\overset{.}{x}$, the second equation by $%
\overset{.}{y}$ and the third equation by $\overset{.}{z}$ and summing up,
one can obtain 
\begin{equation}
\left( \lambda R_{AB}\right) _{s}+\frac{1}{C_{1}}\left( \overset{.}{x}%
\overset{..}{x}+\overset{.}{y}\overset{..}{y}+\overset{.}{z}\overset{..}{z}%
\right) =0\text{ \ \ \ .}  \label{Obz32}
\end{equation}%
Since we have proved that $L_{1}\frac{dL_{1}}{ds}=\overset{.}{x}\overset{..}{%
x}+\overset{.}{y}\overset{..}{y}+\overset{.}{z}\overset{..}{z}=const$, it
follows that $\lambda R_{AB}=const_{1}s+const_{2}$, i.e. $R_{AB}(x,y,z)=\sqrt%
[.]{x^{2}+y^{2}+z^{2}}$ is the straight line! In the case, the straight line
will be between the signal-emitting satellite and the signal-receiving
satellite and 
\begin{equation}
R_{AB}=\sqrt[.]{%
(x_{1.}^{.}-x_{2})^{2}+(y_{1.}^{.}-y_{2})^{2}+(z_{1.}^{.}-z_{2})^{2}}%
=const_{1}s+const_{2}\text{ \ \ ,}  \label{Obz33}
\end{equation}%
where $x_{1.}^{.},x_{2},y_{1.}^{.},y_{2},z_{1.}^{.}z_{2}$ are expressed from
(\ref{K1}), (\ref{K2}) and (\ref{K3}) in terms of the two sets of Kepler
parameters $(f_{.1},a_{1.},e_{1.},\Omega _{1.},i_{1.},\omega _{1.})$ and $%
(f_{2.},a_{2.},e_{2.},\Omega _{2.},i_{2.},\omega _{2.})$. So in terms of all
these variables, this will be a complicated $12-$dimensional surface.

\subsection{Geometrical justification of the obtained result}

The obtained formulae (\ref{Obz33}) may seem elementary and trivial, but it
has an important physical meaning, which will justify the application of a
more specific variational formalism.

Let us first understand what does this formulae implies: it means that the
manifold defined in terms of the space variables $x=x_{1}-x_{2}$, $%
y=y_{1}-y_{2}$, $z=z_{1}-z_{2}$ contains a straight line. But then one can
notice that there is a special theorem in differential geometry (see the
contemporary textbook \cite{PrasolovDifferential}):

\begin{theorem}
If a surface contains a straight line then this line (endowed with the
natural parameter $s$) is a geodesic.
\end{theorem}

In fact, we started from the action functional (\ref{Obz21A}) $T=\frac{1}{c}%
\int [\sqrt[.]{\left( \overset{.}{x}\right) ^{2}+\left( \overset{.}{y}%
\right) ^{2}+\left( \overset{.}{z}\right) ^{2}}$

$+\lambda (x,y,z)R_{AB}(x,y,z)]ds$ (supplemented with the constraint function, which however did
not influence the variational formalism), then we applied a variational
formalism and we found as extremals the straight lines! Since the action
functional $T=\frac{1}{c}\int \sqrt{\left( \overset{.}{x}\right) ^{2}+\left( 
\overset{.}{y}\right) ^{2}+\left( \overset{.}{z}\right) ^{2}}ds$ is a
consequence of the original null cone equation (\ref{DOP25}), there is an
analogy with a theorem, cited in the review paper \cite{Bog2} and proved in
the monograph by Fock \cite{Fock}.The theorem is the following: $\ $

Let $F=\frac{1}{2}g_{\alpha \beta }\frac{dx^{\alpha }}{ds}$ $\frac{dx^{\beta
}}{ds}$ , $F=0$ is the null cone equation and the extremum of the integral $%
s=\int\limits_{s_{1}}^{s_{2}}Lds$ is searched $(L=\sqrt{2F})$, then the
Lagrange equation of motion will give 
\begin{equation}
\frac{dF}{ds}=0\text{ \ }\Longrightarrow F=const\Longrightarrow \frac{d}{ds}%
\frac{\partial F}{\partial \overset{.}{x}_{\alpha }}-\frac{\partial F}{%
\partial x_{\alpha }}=0\text{ \ \ .}  \label{Obz34}
\end{equation}%
From the last equation the geodesic equations $(\alpha ,\beta ,\nu =0,1,2,3)$
\begin{equation}
\frac{d^{2}x_{\nu }}{ds^{2}}+\Gamma _{\alpha \beta }^{\nu }\frac{dx_{\alpha }%
}{ds}\frac{dx_{\beta }}{ds}=0\text{ \ \ \ \ }  \label{Obz35}
\end{equation}%
can be obtained. \ Since the constant in (\ref{Obz34}) can be set up to
zero, this would mean that the geodesic equation (\ref{Obz35}) is compatible
with the null-cone equation $F=\frac{1}{2}g_{\alpha \beta }\frac{dx^{\alpha }%
}{ds}$ $\frac{dx^{\beta }}{ds}=0$. It is important to acquire the proper
understanding of this statement - it means that every solution of the null
equation (including in the proper parametrization of the differentials,
related to the trajectory of motion of the satellite) is also a solution of
the light-like geodesic equation. In the "backward direction", however the
statement is not true, meaning that there might be another solutions of the
geodesic equation, which are not solutions of the null-cone equation. This
problem needs further investigation.

The meaning of the geodesic as the shortest path, joining two given points
(but only in the local sense, i.e a sufficiently short part of the
geodesics) is confirmed also by the following theorem in the monograph by
Prasolov \cite{PrasolovDifferential}. Let us present the exact formulation
of this theorem: \ 

\begin{theorem}
Any point $p$ on a surface $S$ has a neighborhood $U$ such that the length
of a geodesic $\gamma $ going from the point $p$ to some point $q$ and lying
entirely in the neighborhood $U$ does not exceed the length of any curve $%
\alpha $ on $S$, joining the points $p$ and $q$. Moreover, if the lengths of
the curves $\gamma $ and $\alpha $ are equal, then $\gamma $ and $\alpha $
coincide as non-parametrizable curves.
\end{theorem}

In fact, what does it mean that "the length of the geodesics from point $p$
to point $q$ does not exceed the length of any other curve" (but only
locally, i.e. in the vicinity of the given neighborhood)? This means that
this geodesic line is minimal. So this confirms the theorem in another,
third well-known book on differential geometry by Mishtenko and Fomenko \cite%
{MishtFomCourse}:

\begin{theorem}
The geodesic line $\gamma :\left[ a,b\right] \rightarrow M^{n}$ is minimal
if it is not longer than any smooth path, joining the endpoints $\gamma (a)$
and $\gamma (b)$.
\end{theorem}

So if we have for example the functional $E(\gamma (s))=\int \mid \overset{.}%
{\gamma }^{2}(s)\mid ds=\int g_{ij}(x)\overset{.}{x}^{i}(s)\overset{.}{x}%
^{j}(s)ds$ for the null cone equation, and we are searching for the
extremals, these extremals will be derived from the geodesics- the multitude
of the solutions for the extremals from the geodesics turns out to be larger
in comparison with the one for the other local curves in the vicinity, in
particular for the functional $E(\gamma (s))$ for the null-cone equation.

On the base of combining the results in the previously mentioned monographs
by Fock \cite{Fock} and \cite{PrasolovDifferential}, one easily establishes
the validity of the following theorem in \cite{MishtFomCourse}:

\begin{theorem}
The extremals of the length functional $L(\gamma (s))=\int \sqrt[.]{g_{ij}(x)%
\overset{.}{x}^{i}(s)\overset{.}{x}^{j}(s)}ds$ are smooth functions, derived
from the geodesics (by means of smooth change of the parameter along the
curve). Then each extremal of the functional $E(\gamma (s))=\int \mid 
\overset{.}{\gamma }^{2}(s)\mid ds=\int g_{ij}(x)\overset{.}{x}^{i}(s)%
\overset{.}{x}^{j}(s)ds$ (called also action functional of the trajectory)
is also an extremal of the functional $L(\gamma (s))=$ $\int \sqrt[.]{%
g_{ij}(x)\overset{.}{x}^{i}(s)\overset{.}{x}^{j}(s)}ds$. However, each
extremal of the functional \ $L(\gamma (s))$ is not necessarily an extremal
of the functional $E(\gamma (s))$.
\end{theorem}

There is a simple inequality $L^{2}(\gamma (s))\leq $ $E(\gamma (s))$ \cite%
{MishtFomCourse} in confirmation of the above theorem, which can be
immediately proved by means of the Schwartz inequality 
\begin{subequations}
\begin{equation}
\left( \int\limits_{0}^{1}f(s)g(s)ds\right) \leq \left(
\int\limits_{0}^{1}f^{2}(s)ds\right) \left(
\int\limits_{0}^{1}g^{2}(s)ds\right)  \label{Obz36}
\end{equation}%
for the following functions \cite{MishtFomCourse} 
\end{subequations}
\begin{equation}
f(s)\equiv 1\text{ \ \ }g(s)\equiv \mid \overset{.}{\gamma }(s)\mid \text{ \
\ .}  \label{Obz37}
\end{equation}%
It should be noted that the parameter $s$ in the book \cite{Fock} is defined
simply as a "space parameter along a curved line". In the above formulaes, $%
s $ is also a space parameter. But as noted in \cite{MishtFomCourse}, when $%
g(s)=const$, i.e. $\gamma (s)=const$ and the parameter $s$ is proportional
to the length of the arc of the curved line, an equality can be derived \ $%
L^{2}(\gamma (s))=E(\gamma (s))$.

\subsection{Physical interpretation of the obtained constancy of the
Euclidean distance and the length function as a result of the applied
variational formalism}

The constancy of the length function $L_{1}=\sqrt[.]{\left( \overset{.}{x}%
\right) ^{2}+\left( \overset{.}{y}\right) ^{2}+\left( \overset{.}{z}\right)
^{2}}$ and of $R_{AB}(x,y,z)$ is non-acceptable from a physical point of
view, because the signal trajectory should not be a straight line due to the
General Relativity effects of curving the trajectory. Of course, for the
moment we have minimized only the first term $\frac{1}{c}\int\limits_{path}%
\sqrt{\left( \overset{.}{x}\right) ^{2}+\left( \overset{.}{y}\right)
^{2}+\left( \overset{.}{z}\right) ^{2}}ds$ in the action functional (\ref%
{Obz19}) for the generalized Shapiro delay formulae. Due to this reason it
is necessary to investigate also the second term in (\ref{Obz17}) and (\ref%
{Obz19}). \ 

However, investigation of the second term will not give an answer to some
fundamental questions, which naturally appear, concerning the Shapiro delay
formulae (\ref{AA19}). On purely physical considerations it is a very crude
approximation to assume that there will exist a (even not so small)
space-time area, where the light (or radio or laser) ray will move along a
straight line (the first term $\frac{R_{AB}}{c}$ in (\ref{AA19})) and so, it
will not experience any influence from the gravitational field. This
statement follows of course from the Equivalence Principle in General
Relativity Theory. However, the Equivalence Principle is local in its
essence, while in the original derivation of the Shapiro delay formulae (\ref%
{AA19}) the distance $R_{AB}$ is not small. The approximation, from a
fundamental point of view is that the gravitational field causes only a
small deviation from the straight line trajectory of the signal. In the
papers \cite{Bog3} and \cite{Bog4} \ a formulae was obtained for the square
of the geodesic distance $\widetilde{R}_{AB}^{2}$ as compared to the square
of the Euclidean distance $R_{AB}^{2}$ (see also formulae (163) in the
review \cite{Bog2})%
\begin{equation*}
\widetilde{R}_{AB}^{2}-R_{AB}^{2}=-\frac{1}{2}%
(a_{1}^{2}+a_{2}^{2})+e_{1}e_{2}a_{1}a_{2}
\end{equation*}%
\begin{equation}
-\frac{1}{4}(a_{1}^{2}e_{1}^{2}+a_{2}^{2}e_{2}^{2})+2a_{1}a_{2}\sqrt{%
(1-e_{1}^{2})(1-e_{2}^{2})}\text{ \ .}  \label{Obz37A1}
\end{equation}%
Although only for the partial case of two satellites on two plane elliptical
orbits, in \cite{Bog3} and \cite{Bog4} it was proved that the right - hand
side of (\ref{Obz37A1}) is positive, consequently $\widetilde{R}%
_{AB}^{2}>R_{AB}^{2}$ and $\widetilde{R}_{AB}^{2}$ really represents the
geodesic distance, travelled by the signal. What is also interesting is that
the last $4$ numerical terms depend only on constant parameters $%
e_{1},e_{2},a_{1},a_{2}$ and not on the dynamical parameters - the eccentric
anomaly angles $E_{1}$ and $E_{2}$. This means that these terms describe a
deviation from the straight line and thus are not related to the process of
"curving" the trajectory, which will be described by the dynamical
parameters $E_{1}$ and $E_{2}$ (if this was the case)

If not in the infinitesimal sense and according to formulae (\ref{Obz9B}) $%
l(s)=\int\limits_{s_{0}}^{s}\mid \overset{.}{\overrightarrow{\mathbf{R}}}%
_{AB}(s)\mid ds$, the distance along the curve $\overset{.}{\overrightarrow{%
\mathbf{R}}}_{AB}(s)=(\overset{.}{x}(s),\overset{.}{y}(s),\overset{.}{z}(s))$
will be greater than the distance $R_{AB}$ between the points on the
straight line. So the "non-local" character of the action of the
gravitational field will be more correctly "approximated" by the term $\frac{%
1}{c}\int\limits_{path}\sqrt{\left( \overset{.}{x}\right) ^{2}+\left( 
\overset{.}{y}\right) ^{2}+\left( \overset{.}{z}\right) ^{2}}ds$ rather than
the simple term $\frac{R_{AB}}{c}$. In support of this assertion, the
following citation from the monograph by Fock \cite{Fock} can be reminded:
"The essence of the Equivalence Principle is that by means of introducing a
locally-geodesic (freely-falling) reference system, the Galilean space-time
is restored in an infinitely small space. However, in no way this enables
one to make any conclusions how to discern acceleration and gravitation in a
finite space-time area."

It has been also asserted by Fock \cite{Fock} that: "In any non-local
interpretation the approximate equivalence of the gravitational field and
the acceleration is restricted. This "equivalence" can take place only for
weak and homogeneous gravitational fields and also for slow motions." The
reason is that the formulae (\ref{Obz7}) $dT\approx \frac{1}{c}%
\int\limits_{path}\sqrt{\delta _{ij}dx^{i}dx^{j}}+\frac{1}{c^{3}}%
\int\limits_{path}2U\sqrt{\delta _{ij}dx^{i}dx^{j}}$ , from which we started
for the derivation of the modified Shapiro formulae, in its original form
has been derived under the assumption $\frac{U}{c^{2}}\ll 1$. It is
equivalent to the following two inequalities \cite{Fock} \ 
\begin{equation}
U\ll c^{2}\text{ \ \ ,}  \label{Obz38}
\end{equation}%
\begin{equation}
\left( \frac{dx}{dt}\right) ^{2}+\left( \frac{dy}{dt}\right) ^{2}+\left( 
\frac{dz}{dt}\right) ^{2}=v^{2}\ll c^{2}\text{ \ \ \ .}  \label{Obz39}
\end{equation}%
Under the fulfillment of these inequalities, from the original metric (in
fact, this is the metric, used in this paper, from which the null cone
equation (\ref{DOP25}) was derived) 
\begin{equation}
ds^{2}=(c^{2}-2U)\left( dt\right) ^{2}-(1+\frac{2U}{c^{2}}%
)((dx)^{2}+(dy)^{2}+(dz)^{2})\text{ \ \ ,}  \label{Obz40}
\end{equation}%
only the approximate metric 
\begin{equation}
ds^{2}=(c^{2}-2U)\left( dt\right) ^{2}-(1+\frac{2U}{c^{2}})\left(
(dx)^{2}+(dy)^{2}+(dz)^{2}\right)  \label{Obz41}
\end{equation}%
can be derived, and this is the mathematical expression of the
approximation, pointed out by Fock \cite{Fock}. In our notations, instead of
"$t$" we have written the propagation time $T$. Now let us see whether it
will be correct to keep the physical meaning of (\ref{Obz39}) as the
Newtonian velocity. If this is so, (\ref{Obz39}) can be written as $v^{2}=%
\left[ \left( \overset{.}{x}\right) ^{2}+\left( \overset{.}{y}\right)
^{2}+\left( \overset{.}{z}\right) ^{2}\right] (\frac{ds}{dT})^{2}\ll c^{2}$
, or 
\begin{equation}
\left[ \left( \overset{.}{x}\right) ^{2}+\left( \overset{.}{y}\right)
^{2}+\left( \overset{.}{z}\right) ^{2}\right] \ll c^{2}\left( \frac{dT}{ds}%
\right) ^{2}\text{ \ \ \ .}  \label{Obz42}
\end{equation}%
But $\frac{dT}{ds}$ is the function $F(x,y,z,\overset{.}{x},\overset{.}{y},%
\overset{.}{z})=\frac{1}{c}\sqrt{\left( \overset{.}{x}\right) ^{2}+\left( 
\overset{.}{y}\right) ^{2}+\left( \overset{.}{z}\right) ^{2}}+\frac{%
2G_{\oplus }M_{\oplus }}{c^{3}}\frac{\sqrt{\left( \overset{.}{x}\right)
^{2}+\left( \overset{.}{y}\right) ^{2}+\left( \overset{.}{z}\right) ^{2}}}{%
\left[ x^{2}+y^{2}+z^{2}\right] }\left[ x\overset{.}{x}+y\overset{.}{y}+z%
\overset{.}{z}\right] $ in (\ref{Obz20}) and neglecting the second term in
the expansion for $c^{2}F^{2}$, which is of order $O(c^{-4})$, we obtain 
\begin{equation}
1\ll 1+\frac{G_{\oplus }M_{\oplus }}{c^{2}}\frac{d}{ds}\ln
(x^{2}+y^{2}+z^{2})\text{ \ \ \ \ \ .}  \label{Obz43}
\end{equation}%
This is however not realistic, because the second term should be a large
number. Since this cannot be the case, since numerous calculations have
proved that only because of the coefficient ,$\frac{G_{\oplus }M_{\oplus }}{%
c^{3}}\approx 10^{-9}$ (see the paper \cite{Bog1} and the cited papers) it
turns out that the definition (\ref{Obz39}) is incorrect. The reason is that
the velocity should be defined only as a tangent vector to the curve, and
then the locality condition of Fock should be defined as 
\begin{equation}
\left( \overset{.}{x}\right) ^{2}+\left( \overset{.}{y}\right) ^{2}+\left( 
\overset{.}{z}\right) ^{2}\ll c\text{ \ \ .}  \label{Obz44}
\end{equation}%
Then if the second term in (\ref{Obz20}) is much less than the first one and
inequality (\ref{Obz44}) is fulfilled, then from 
\begin{equation}
\frac{2G_{\oplus }M_{\oplus }}{c^{3}}\frac{\sqrt{\left( \overset{.}{x}%
\right) ^{2}+\left( \overset{.}{y}\right) ^{2}+\left( \overset{.}{z}\right)
^{2}}}{\left[ x^{2}+y^{2}+z^{2}\right] }\left[ x\overset{.}{x}+y\overset{.}{y%
}+z\overset{.}{z}\right] \ll \frac{1}{c}\sqrt{\left( \overset{.}{x}\right)
^{2}+\left( \overset{.}{y}\right) ^{2}+\left( \overset{.}{z}\right) ^{2}}
\label{Obz45}
\end{equation}%
it will follow that $\frac{2G_{\oplus }M_{\oplus }}{c^{3}}\frac{\sqrt{\left( 
\overset{.}{x}\right) ^{2}+\left( \overset{.}{y}\right) ^{2}+\left( \overset{%
.}{z}\right) ^{2}}}{\left[ x^{2}+y^{2}+z^{2}\right] }\left[ x\overset{.}{x}+y%
\overset{.}{y}+z\overset{.}{z}\right] \ll 1$ or 
\begin{equation}
\frac{G_{\oplus }M_{\oplus }}{c^{3}}\frac{d}{ds}\ln (x^{2}+y^{2}+z^{2})\ll 
\frac{1}{c}\text{ \ \ .}  \label{Obz46}
\end{equation}%
This is much more realistic and compatible with the fact that the second
Shapiro delay term should be a very small one. 

\section{Differential geometry notions about manifolds with boundary - basic
definitions and related notions and new results about the boundaries of the
manifold (\protect\ref{K1})-(\protect\ref{K3}) of the celestial mechanics
transformations}

The application of the Shapiro formulae in this paper is based on the
transformations (\ref{K1}) - (\ref{K3}) for the space coordinates $x,y,z$,
depending on the $6$ Kepler parameters. These transformations have not been
investigated from a differential geometry point of view, and they have
interesting geometrical properties, depending on the fact which variables
will be considered dynamical ones and which will be kept constant in the
process of signal propagation from one satellite to another. For example, if
only the two true anomaly angles $f_{1}$ and $f_{2}$ are the dynamical
parameters, we will have the mapping $F_{1}:(f_{1},f_{2})\rightarrow (x,y,z)$
(again, here we denote $x=x_{2}-x_{1},y=y_{2}-y_{1},z=z_{2}-z_{1}$), which
is a mapping from a two-dimensional manifold to a three-dimensional
manifold. We can define, however, also the mapping $F_{2}:$ $(f_{1},\Omega
_{1},\omega _{1},i_{1},a_{1},e_{1};f_{2},\Omega _{2},\omega
_{2},i_{2},a_{2},e_{2})\rightarrow (x,y,z)$ or $F_{2}:$ $(f_{1},\Omega
_{1},\omega _{1},i_{1},a_{1},e_{1};f_{2},\Omega _{2},\omega
_{2},i_{2},a_{2},e_{2})\rightarrow (x_{1},y_{1},z_{1};x_{2},y_{2},z_{2})$%
-these are mappings from a $12-$ dimensional manifold to a three-dimensional
manifold $(x,y,z)$ or to a six-dimensional manifold. So in the first case
the mapping $F_{1}$ is from a lower-dimensional ($2-$dimensional manifold)
to a higher-dimensional manifold, while in the second case of the mapping $%
F_{2}$ the mapping is from a higher-dimensional to a lower-dimensional
manifold. These different kinds of mappings are related to the notions of
immersions, submersions, embeddings, which will be clarified in the next
sections and are related to the Stokes and the Gauss - Ostrogradsky theorems.

\subsection{Basic definitions about the boundary of a manifold and some
relations to other definitions}

One of the basic definitions is given in the textbook \cite{MishtFomCourse}:

\begin{definition}
A manifold $M$ with a boundary is characterized by the property that there
exists an atlas $\{U_{\alpha }\}$ with local coordinates $(x_{\alpha
}^{1},x_{\alpha }^{2},......,x_{\alpha }^{n})$ so that for each map the
strict inequality $x_{\alpha }^{n}>0$ is fulfilled for the inner points and
the equality $x_{\alpha }^{n}=0$ is fulfilled at the boundary points. The
set of boundary points $\partial M$ is a smooth manifold with a dimension $%
(n-1)$ (i.e. one dimension less than the dimension of the original
manifold). \ 

In the more strict mathematical sense, a manifold with a boundary is a
metric space $M$, for which an atlas of charts $\{U_{\alpha }\}$ and
coordinate homeomorphisms $\varphi _{\alpha }:U_{\alpha }\rightarrow
V_{\alpha }\subset 
%TCIMACRO{\U{211d} }%
%BeginExpansion
\mathbb{R}
%EndExpansion
_{+}^{n}$ ($V_{\alpha }$ is an open set in $%
%TCIMACRO{\U{211d} }%
%BeginExpansion
\mathbb{R}
%EndExpansion
_{+}^{n}$ and $%
%TCIMACRO{\U{211d} }%
%BeginExpansion
\mathbb{R}
%EndExpansion
_{+}^{n}$ is the semi-space, denoted as $x^{n}\geq 0$, $%
(x^{1},x^{2},....x^{n})\in 
%TCIMACRO{\U{211d} }%
%BeginExpansion
\mathbb{R}
%EndExpansion
^{n}$) exist such that  the functions for coordinate changes 
\begin{equation}
\varphi _{\beta }\varphi _{\alpha }^{-1}:V_{\alpha \beta }=\varphi _{\alpha
}(U_{\alpha }\cap U_{\beta })\rightarrow V_{\beta \alpha }=\varphi _{\beta
}(U_{\alpha }\cap U_{\beta })  \label{Obz47}
\end{equation}%
are smooth functions of class $C^{r}$, $r=1,2,......\infty $.
\end{definition}

This definition however does not reveal the relation of the boundary of the
manifold with the inverse function theorem as a basic requirement for the
existence of the boundary. This means that if a vector function $%
f:R^{n}\rightarrow R^{k}$ is defined, then it should satisfy the inverse
function theorem as a necessary requirement (condition) for the existence of
a boundary. \cite{MishtFomCourse} Namely, if at all points $P\in
M=\{f(P)=0\} $, the differential $df$ has a maximal rank.

This means also that if the system of $n-$equations is investigated 
\begin{equation}
f^{1}(x^{1},x^{2}...x^{n})=0\text{ }..\text{\ , \ }%
f^{k-1}(x^{1},x^{2}...x^{n})=0\text{ , ........}f^{k}(x^{1},x^{2}...x^{n})=0
\label{Obz48}
\end{equation}%
and $\overline{M}$ is the set of solutions of this system, and if also at
each point $P\subset $. $\overline{M}$ the rank of the determinant $%
\parallel \frac{\partial f^{i}}{\partial x^{s}}\parallel $, $1<i\leq k-1$, $%
1\leq s\leq n$ is equal to $(k-1)$, then $\overline{M}$ is a manifold with a
boundary.

Now let us give a simple example about a manifold with a boundary - this is
the closed ball $D^{m}\subset R^{m}$, defined by the inequality $%
(x^{1})^{2}+(x^{2})^{2}+(x^{m})^{2}\leq 1$. The $m-$cube $[0,1]^{m}$ is
homeomorphic to $D^{m}$ and is thus also a manifold with a boundary.

\subsubsection{Definition for a boundary of a manifold in terms of
semi-spaces}

The definition in another monographs, for example in \cite{Boothby} implies
a slightly different formulation.

\begin{definition}
A manifold with a boundary is a Hausdorff space $M$ with a countable basis
of open sets and the subspace of $%
%TCIMACRO{\U{211d} }%
%BeginExpansion
\mathbb{R}
%EndExpansion
^{n}$
\end{definition}

\begin{equation}
H^{n}=\left\{ x=(x_{1},x_{2},.....x_{n})\in R^{n})\mid x^{n}\geq 0\right\}
\label{Obz49}
\end{equation}
with the relative topology of $R^{n}$ and 
\begin{equation}
\partial H^{n}=.\left\{ x\in H^{n}\mid x^{n}=0\right\}  \label{Obz50}
\end{equation}
denotes the boundary of $H^{n}$. The boundary $\partial H^{n}$ carries the
metric topology derived from the metric of $R^{n}$ and $\partial H^{n}$ is
homeomorphic to $R^{n-1}$ by the map 
\begin{equation}
(x_{1},x_{2},.....x_{n-1})\mapsto (x_{1},x_{2},.....,x_{n-1},0)\text{ \ .}
\label{Obz51}
\end{equation}
\ $\ \ $
The diffeomorphisms of open sets of $H^{n}$ take boundary points to
boundary points and interior points to interior points. This means also that
by means of the homeomorphisms $\varphi _{n}$ a point $p\subset M$ is
contained either 1. In an open set $U$ of $H^{n}-\partial H^{n}$ or 2. to an
open set $U$ of $H^{n}$ with  $\varphi (p)\subset \partial H^{n}$ (i.e. a
boundary point of $H^{n}$) The set of points in $H^{n}-\partial H^{n}$ is in
fact $\left\{ x=(x_{1},x_{2},.....x_{n})\in R^{n})\mid x^{n}\geq 0\right\} $
and in \cite{SchwarzPhys} it is called "a closed half-space". This is the
set of vectors $x\in R^{n}$, satisfying the positive vector multiplication $%
x.y\geq 0$, where $y\in R^{m}$ is a fixed non-zero vector. Now one can see
the close similarity of the definition for $H^{n}-\partial H^{n}$ or $H^{n}$
with $x^{n}\geq 0$ with the definition of convex spaces, which exist in
optimization problems \cite{Suharev}. In \cite{Suharev} sets of the kind $%
H_{p\beta }=\left\{ x\in R^{n}\mid <p.x>=\beta \right\} $ are called
hyperplanes and they cause the semi-spaces 
\begin{equation}
H_{p\beta }^{+}=\left\{ x\in R^{n}\mid <p.x>\text{ }\geq \beta \right\} 
\text{ \ \ , \ \ }H_{p\beta }^{--}=\left\{ x\in R^{n}\mid <p.x>\text{ }\leq
\beta \right\} \text{ \ \ . }  \label{Obz52}
\end{equation}
Optimization problems are investigated in mathematics and apart from gravity
theory, but in the previous sections it was demonstrated that they may
appear in problems, related to the optimal configuration for the propagation
of a signal between two moving satellites on different space-distributed
orbits, with account also of the General Relativity effect of curving the
trajectory of the signal.

\subsubsection{The notion of convexity in differential geometry and its
relation to the notion of a boundary}

It is very interesting that the notion of convexity appears not only in
optimization theory, but also in differential geometry and it turns out that
it is related also to the notion of a boundary. In the course of Thorpe \cite%
{Thorpe} on differential geometry the following definition is given for a
convex surface

An oriented $n-$dimensional hyper-surface $S$ in $R^{n+1}$is called convex
(or a convex one as a whole), if at each point $p\in S$ the hyper-surface is
contained in one of the following semi-spaces 
\begin{equation}
H_{p}^{+}=\left\{ q\in R^{n+1}\text{ \ \ }:(q-p).N(p)\geq 0\right\} \text{ \
\ \ \ ,}  \label{Obz53}
\end{equation}%
or 
\begin{equation}
H_{p}^{-}=\left\{ q\in R^{n+1}\text{ \ \ }:(q-p).N(p)\leq 0\right\} \text{ \
\ .}  \label{Obz54}
\end{equation}%
On the base of this definition in \cite{Thorpe}, another definition about a
strictly convex surface is given, which evidently has a relation to the
notion of a boundary, given for example by (\ref{Obz50}) $\partial
H^{n}=\left\{ x\in H^{n}\mid x^{n}=0\right\} $.

\begin{definition}
If the surface $S$ is convex at each point $p$ and $S\cap H_{p}=\{p\}$ at
each point $p\in S$, where 
\begin{equation}
H_{p}=\left\{ q\in R^{n+1}\text{ \ \ }:(q-p).N(p)=0\right\} \text{ \ \ \ \ ,}
\label{Obz55}
\end{equation}%
then the hyper-surface $S$ is called strictly convex.
\end{definition}

The vector $N(p)$ is an important characteristic of the surface and is
called the Gaussian mapping \cite{Thorpe}. This is a mapping $N:S\rightarrow
S^{n}$, which juxtaposes to each point $p\in S$ a point in $R^{n+1}$,
derived after the translation of the unit normal vector $p$ to the origin of
the coordinate system. The image in the Gaussian mapping $N(s)=\left\{ q\in
S^{n}:q=N(p)\right\} $ is called also the spherical image of the oriented
surface.

By means of $N(p)$, two other characteristics are defined \cite{Thorpe}.

\begin{definition}
The Weingarten mapping $L_{p}:S_{p}\rightarrow S_{p}$ is determined as $%
L_{p}(v):=-\nabla _{v}N$, $\overrightarrow{\mathbf{v}}\in S_{p}$. It
measures the turning angle of the normal to $S$ at the point $p$, moving at
different speeds $v$ and thus $L_{p}$ is a measure of the curvature of the
space $R^{n+1}$ at the point $p$.
\end{definition}

The second characteristics is the s.c. normal curvature of $S$ at the point $%
p$ in the direction of the vector $\overrightarrow{\mathbf{v}}$.

\begin{definition}
In the case of $\parallel \overrightarrow{\mathbf{v}}\parallel =1$, the
normal curvature of $S$ at the point $p$ in the direction of the vector $%
\overrightarrow{\mathbf{v}}$ \ is defined as the number $k(p):=%
\overrightarrow{L_{p}}(v).\overrightarrow{\mathbf{v}}$ and it represents the
normal component of the acceleration at the point $p$, which is caused by
the curvature of $S$ in $R^{n+1}$. This means that if $k(\overrightarrow{%
\mathbf{v}})>0$, then the surface $S$ is "curving" in the direction of $%
\overrightarrow{N}$ and in the direction of $\overrightarrow{\mathbf{v}}$,
if $k(\overrightarrow{\mathbf{v}})<0$ the surface $S$ is curving in the
backward direction of $\overrightarrow{N}$ and in the direction of $%
\overrightarrow{\mathbf{v}}$.
\end{definition}

\subsection{New result - calculation of the boundary manifold for the $3D$
manifold of the celestial mechanics transformations (\protect\ref{K1})-(%
\protect\ref{K3})}

Further the Stokes theorem will be used, which for the concretely
investigated case will relate an integral on a $3D$ space $M$ to an integral
on the two-dimensional boundary $\partial M$. Since the boundary manifold is
one dimension less than the original $3D$ manifold, it will be in terms of
the variables $(x,y)$ (the first boundary, when $z=0$), or in terms of the
variables $(x,z)$ (the second boundary, when $y=0$) or in terms of the
variables $(y,z)$ (the third boundary, when $x=0$).

\subsubsection{The first boundary in terms of $(x,y)$ and $z=0$}

In order to find the boundary, one should combine the two equations (\ref{K1}%
) for $x$ and (\ref{K2}) for $y$ so that to obtain a relation between the
coordinates. Let us multiply both sides of (\ref{K1}) with $\cos \Omega $
and of equation (\ref{K2}) with $\sin \Omega $ and sum up the two equations.
The obtained equation will be

\begin{equation}
x\cos \Omega +y\sin \Omega =r\cos (\omega +f)\text{ \ \ \ \ .}  \label{Obz56}
\end{equation}%
Similar equation may be obtained if (\ref{K1}) for $x$ is multiplied by $%
\sin \Omega $ and (\ref{K2}) for $y$ by $(-\cos \Omega )$ and the two
equations are summed up. The resulting equation will be 
\begin{equation}
x\sin \Omega -y\cos \Omega =-r\sin (\omega +f)\cos i\text{ \ \ . \ \ \ \ }
\label{Obz57}
\end{equation}%
Next, according to the definition $\partial H^{n}=\left\{ x\in H^{n}\mid
x^{n}=0\right\} $ in (\ref{Obz50}), we have to find when the condition 
\begin{equation}
z=r\sin (\omega +f)\sin i=0  \label{Obz58}
\end{equation}%
is fulfilled. The first case is when 
\begin{equation}
\sin i=0\text{ \ }\Rightarrow i=\pi k=\pi (m+1)\text{ \ }  \label{Obz59}
\end{equation}%
and $\sin (\omega +f)\neq 0$ is also fulfilled. Then when 
\begin{equation}
k=0,2,4.......\Rightarrow k=m+1\text{ }\Rightarrow \text{\ \ }%
m=1,3..........k\text{ - even; }m\text{-odd}  \label{Obz60}
\end{equation}%
we have for 
\begin{equation}
\cos i=\cos (\pi (m+1))=(-1)^{m+1}\text{ \ \ \ .}  \label{Obz61}
\end{equation}%
Indeed, for odd numbers $m=1,3.....(2p+1)$, $\cos (\pi
(m+1))=(-1)^{m+1}=(-1)^{2p+2}=1$, for even numbers $m=2,4,6....2p$ one has $%
\cos (\pi (m+1))=(-1)^{m+1}=(-1)^{2p+1}=-1$, which can be checked. Note for $%
k=0$ one cannot define the corresponding value for $m$, then $\cos (\pi
.0)=\cos 0=1$. However, instead of \ $i=\pi k=\pi (m+1)$ in (\ref{Obz59})
one can also define $i=\pi k=\pi (m-1)$ - then $k=0$ will correspond to the
value $m=1$.

The other condition $\sin (\omega +f)\neq 0$ will be fulfilled when
similarly to (\ref{Obz59}) one has 
\begin{equation}
\sin (\omega +f)\neq 0\text{ \ }\Rightarrow \omega +f\neq \pi k=\pi (m\pm 1)%
\text{ }\ \ \ \text{.\ }  \label{Obz61A1}
\end{equation}

Next, dividing the two equations (\ref{Obz56}) and (\ref{Obz57}) and making
use of (\ref{Obz61}), one can derive the following simple equation 
\begin{equation}
\frac{x\cos \Omega +y\sin \Omega }{x\sin \Omega -y\cos \Omega }=-\frac{%
tg(\omega +f)}{(-1)^{m+1}}\text{ \ \ \ .}  \label{Obz62}
\end{equation}%
For odd numbers $m=1,3.....(2p+1)$ after some transformations and also using
the simple trigonometric equalities 
\begin{equation}
\cos (\omega +f+\Omega )=\cos (\omega +f)\cos \Omega -\sin (\omega +f)\sin
\Omega \text{ \ \ \ ,}  \label{Obz63}
\end{equation}%
\begin{equation}
\sin (\omega +f+\Omega )=\sin (\omega +f)\cos \Omega +\cos (\omega +f)\sin
\Omega \text{ \ \ \ ,}  \label{Obz64}
\end{equation}%
the equation (\ref{Obz62}) can be rewritten as 
\begin{equation}
x\frac{\sin (\omega +f+\Omega )}{\cos (\omega +f)\cos \Omega }-y\frac{\cos
(\omega +f+\Omega )}{\cos (\omega +f)\cos \Omega }=0\text{ \ \ .}
\label{Obz65}
\end{equation}%
Provided that the inequalities $\cos (\omega +f)\neq 0$ and also $\cos
\Omega \neq 0$ the last expression represents the equation of the straight
line \ 
\begin{equation}
x\sin (\omega +f+\Omega )-y\cos (\omega +f+\Omega )=0\text{ \ \ .}
\label{Obz66}
\end{equation}%
The two inequalities can be fulfilled when $\omega +f\neq (2q-1)\frac{\pi }{2%
}$, $\Omega \neq (2p-1)\frac{\pi }{2}$. The numbers $p$ and $q$ can be odd
or even, i.e. $p,q=0,1,2,.........$. It is easy to check that $(2q-1)\frac{%
\pi }{2}=0,-\frac{\pi }{2},\frac{\pi }{2},\frac{3\pi }{2},\frac{5\pi }{2}%
.......$for $q=0,1,2,3.......$ and also $(2p-1)\frac{\pi }{2}=0,-\frac{\pi }{%
2},\frac{\pi }{2},\frac{3\pi }{2},\frac{5\pi }{2}.......$for $p=0,1,2,3$,
correspondingly when $\omega +f\neq (2q-1)\frac{\pi }{2}$ and $\Omega \neq
(2p-1)\frac{\pi }{2}$, $\cos (\omega +f)\neq 0$ and $\cos \Omega \neq 0$
will be fulfilled.

Similarly, for even numbers $m=2,4,6....2p......$it can easily be proved
that the analogous to (\ref{Obz66}) equation will be 
\begin{equation}
x\sin (\omega +f-\Omega )+y\cos (\omega +f-\Omega )=0\text{ \ \ ,}
\label{Obz67}
\end{equation}%
where the analogous to (\ref{Obz63}) trigonometric equality 
\begin{equation}
\cos (\omega +f-\Omega )=\cos (\omega +f)\cos \Omega +\sin (\omega +f)\sin
\Omega \text{ \ \ \ ,}  \label{Obz68}
\end{equation}%
has been used. The second case when $z=r\sin (\omega +f)\sin i=0$, i.e. $%
\sin (\omega +f)=0$ should be considered in the same way. The calculations
shall not be performed here. Note also that the number $m$ can be either odd
or even, consequently either equation (\ref{Obz66}) or equation (\ref{Obz67}%
) is fulfilled, but not both, which means that the corresponding lines as
boundaries should not intersect.

\subsubsection{The second boundary in terms of $(x,z)$ and $y=0$}

It might be expected that the calculation for the second boundary $(x,z)$ is
of the same type as for the previous case. Surprisingly, the calculation
turns out to have some peculiarities and the derived expressions will be
more complicated.

Let us multiply equation (\ref{K1}) for $x$ by $(\sin i)$ and (\ref{K3}) for 
$z$ by $(\cos i\sin \Omega )$ and then sum up the two equations. The
obtained equation will be 
\begin{equation}
x\sin i+\cos i\left( \sin \Omega \right) \text{ }z=r\cos \Omega \cos (\omega
+f)\sin i\text{ \ \ \ .}  \label{Obz69}
\end{equation}%
Another equation can be obtained if (\ref{K1}) is multiplied by $(\sin
(\omega +f)\sin i)$ and (\ref{K3}) is multiplied by \ $(-\cos \Omega \cos
(\omega +f))$ and the resulting equations are summed up 
\begin{equation*}
\sin (\omega +f)\sin i\text{ }x-\cos \Omega \cos (\omega +f)z
\end{equation*}%
\begin{equation}
=r\sin \Omega \sin ^{2}(\omega +f)\cos I\sin i\text{ \ \ .}  \label{Obz71}
\end{equation}%
After dividing the two equations, the following equation is obtained 
\begin{equation}
\frac{x\sin i+\cos i\sin \Omega \text{ }z}{\sin (\omega +f)\sin i\text{ }%
x-\cos (\omega +f)\cos \Omega \text{ }z}=-\frac{\cot g\Omega \cot g(\omega
+f)}{\cos i\sin (\omega +f)}\text{ \ \ .}  \label{Obz72}
\end{equation}%
Now we set up $y=0$ in expression (\ref{K2}) for $y$, express $\cot g\Omega $
as%
\begin{equation}
\cot g\Omega =-\frac{1}{tg(\omega +f)\cos i}  \label{Obz73}
\end{equation}%
and substitute it in the preceding equation (\ref{Obz72}) 
\begin{equation*}
x.\sin i\frac{\left[ \sin ^{2}(\omega +f)\left( 1+\cos ^{2}i\right) -1\right]
}{\sin ^{2}(\omega +f)\cos ^{2}i}
\end{equation*}%
\begin{equation}
+z.\frac{\cos \Omega \left[ 1-\cos ^{4}i\text{ }tg^{4}(\omega +f)\right] }{%
\cos ^{2}i\text{ }tg^{3}(\omega +f)}=0\text{ \ \ .}  \label{Obz74}
\end{equation}%
This is again an equation for a straight line, but in terms of the variables 
$x$ and $z$ and with complicated coefficient functions in terms of
trigonometric functions. It can be proved that they contain functions with $%
\cos 2(\omega +f)$ and $\cos 2i.\cos 4i$. \bigskip 

\subsubsection{The third boundary in terms of $(y,z)$ and $x=0$}

The equation (\ref{K2}) for $y$ is multiplied by $\left[ \sin (\omega
+f)\sin i\right] $, equation (\ref{K3}) for $z$ is multiplied by $\left[
-\sin \Omega \cos (\omega +f)\right] $ and after that the obtained equations are
summed up, the following equation is obtained 
\begin{equation*}
y\sin (\omega +f)\sin i-z\text{ }\cos (\omega +f)\sin \Omega
\end{equation*}%
\begin{equation}
=r\cos \Omega \sin ^{2}(\omega +f)\sin i\cos i\text{ \ .}  \label{Obz76}
\end{equation}%
The second equation is obtained after (\ref{K2}) for $y$ is multiplied by $%
\left[ \sin i\right] $, (\ref{K3}) for $z$ is multiplied by $\left[ -\cos
i\cos \Omega \right] $ and after that the obtained equations are summed up. The
obtained equation is 
\begin{equation}
y\sin i-z\text{ }\cos i\cos \Omega =r\sin \Omega \sin i\cos (\omega +f)\text{
\ \ .}  \label{Obz77}
\end{equation}%
Now setting up equation (\ref{K1}) for $x$ to zero, i.e. $x=0$, one can
obtain 
\begin{equation}
\cos i=\cot g\Omega \cot g(\omega +f)\text{ \ \ \ .}  \label{Obz78}
\end{equation}%
One can calculate also $\sin i=\sqrt{1-\cos ^{2}i}$ by using the
trigonometric identities (\ref{Obz64}) for $\cos (\omega +f+\Omega )$ and (%
\ref{Obz68}) for $\cos (\omega +f-\Omega )$. The result is 
\begin{equation}
\sin i=\frac{\sqrt{-\cos (\omega +f+\Omega )\cos (\omega +f-\Omega )}}{\sin
\Omega \sin (\omega +f)}\text{ \ \ \ \ \ \ .}  \label{Obz79}
\end{equation}%
After dividing the equations (\ref{Obz76}) and (\ref{Obz77}) and
substituting the obtained expressions (\ref{Obz78}) for $\cos i$ and (\ref%
{Obz79}) for $\sin i$, after some transformations the following equation is
obtained 
\begin{equation}
\frac{\sin (\omega +f)\left[ \sin ^{2}\Omega -\cos ^{2}\Omega \right] }{\sin
^{2}\Omega }\left[ Ny-z\cos (\omega +f)\right] =0\text{ \ ,}  \label{Obz80}
\end{equation}%
where we have denoted 
\begin{equation}
N(\omega ,f,\Omega )\equiv \sqrt{-\cos (\omega +f+\Omega )\cos (\omega
+f-\Omega )}\text{ \ \ \ \ \ .}  \label{Obz81}
\end{equation}%
In order to obtain the equation for the straight line 
\begin{equation}
Ny-z\cos (\omega +f)=0  \label{Obz81A}
\end{equation}%
in $(y,z)$ coordinates, one should have also the fulfillment of the three
conditions, so that the straight line equation (\ref{Obz81A}), obtained from
(\ref{Obz80}) is correct 
\begin{equation}
\sin (\omega +f)\neq 0\text{ \ \ \ , \ \ \ }\sin ^{2}\Omega \neq 0\text{, \ }%
\sin ^{2}\Omega -\cos ^{2}\Omega \neq 0\text{ \ \ .}  \label{Obz82}
\end{equation}%
The fulfillment of the first condition was established in (\ref{Obz61A1}),
the second condition is investigated analogously, the third condition means
that 
\begin{equation}
\sin ^{2}\Omega -\cos ^{2}\Omega =2\sin ^{2}\Omega -1\neq 0\text{ \ }%
\Rightarrow \sin \Omega \neq \pm \frac{\sqrt{2}}{2}\text{, i.e. }\Omega \neq
45^{\circ }\text{ \ .}  \label{Obz83}
\end{equation}%
Much more nontrivial is the problem whether the expression $N(\omega
,f,\Omega )$ in (\ref{Obz81}) should be real-valued or complex. On physical
grounds, for the moment one cannot assert anything about the real-valuedness
or imaginary value of $N(\omega ,f,\Omega )$. However, on purely geometrical
grounds, it is possible for $N(\omega ,f,\Omega )$ to be imaginary - this
can happen when both $\cos (\omega +f+\Omega )$ and $\cos (\omega +f-\Omega
) $ are positive or when they are both negative. Then there will be the
imaginary straight line $i\overline{N}y-z\cos (\omega +f)=0$ - it will be
one of the imaginary intersecting lines (the other is $i\overline{N}y+z\cos
(\omega +f)=0$) on the cone $\overline{N}^{2}y^{2}+z^{2}\cos ^{2}(\omega
+f)=0$. A theorem, proved in the monograph \cite{Fedorchuk} states that
there exist exactly $5$ classes of projective equivalences among the lines of
second order, namely:

1. a real-valued oval $x_{1}^{2}+x_{2}^{2}-x_{3}^{2}=0$ ;

2. a complex oval $x_{1}^{2}+x_{2}^{2}+x_{3}^{2}=0$ ;

3. intersecting straight lines $x_{1}^{2}-x_{2}^{2}=0$ ;

4. imaginary intersecting straight lines $x_{1}^{2}+x_{2}^{2}=0$ ;

5.coinciding straight lines $x_{1}^{2}=0$.

\section{Differential geometry notions about immersions, embeddings,
submersions, critical and regular points and Morse functions - a review of
the different definitions and some new applications in the current problem}

We have already defined the notion of a boundary of a manifold, which
represents in fact a submanifold with dimension $(n-1)$ of the original $n-$%
dimensional manifold. As clarified in the preceding chapters, a new
variational formalism will be applied, based on the Stokes and
Gauss-Ostrogradsky theorems. These theorems relate the integration on a $n-$%
manifold to an integration on the boundary. But in fact, the notion of the
boundary manifold, concrete examples about which were given, are closely
related with some other notions in differential geometry about immersions,
submersions, embeddings, also critical and regular points and Morse
functions.

\subsection{Immersions and the sub-manifold property}

Before giving the strict definitions, let us first define what is the rank
(denoted by $rank\varphi $) of the mapping $\varphi :N\rightarrow M$ between
two manifolds-this is simply the dimension of the manifold $N$ at all its
points. The essence of this assertion is given in the following

\begin{definition}
\cite{Sternberg} Let the mapping $\varphi :M_{1}\rightarrow M_{2}$ is a
differentiable mapping with $(U,h)$-coordinate charts in the neighborhood of
the point $\varphi (p)$, where $p\in M_{1}$. The differentiable functions on 
$M_{1}$ are denoted by $y^{i}\circ \varphi $ . Let $h$ is a local chart in $%
M_{2}$, $g$ is a local chart in $M_{1}$ and in terms of the coordinate chart 
$(V,g=(x^{1},x^{2},.....x^{n_{1}}))$ the mapping $\varphi \equiv h\circ
\varphi \circ g^{-1}$ is given by the transformation 
\begin{equation}
y^{i}=\varphi ^{i}(x^{1},x^{2},.....x^{n_{1}})\text{ \ \ \ \ \ }%
(i=1,2,.......n_{2})\text{ \ \ .}  \label{Obz84}
\end{equation}%
Then the rank of the mapping $\varphi $ is equal to $n_{1}$, if the rank of
the Jacobian $(n_{1}\times n_{2})$ matrix $\left( \frac{\partial \varphi ^{i}%
}{\partial x^{j}}\right) _{(x^{1},x^{2},.....x^{n_{1}})}$ is equal to $n_{1}$%
.Such a case, when at all points of the manifold $M_{1}$ the $rank\varphi
=n_{1}$ (the dimension of the manifold $M_{1}$) and when the dimension of
the manifold $M_{2}$ is $n_{2}\geq n_{1}$ is called an immersion. (see also 
\cite{Boothby}) \ The subset $\varphi (M_{1})=\overline{M}_{1}\subset M_{2}$
is called an immersed sub-manifold, but the peculiar fact is that it depends
on both manifolds $M_{1}$ and $\overline{M}_{1}$- $\overline{M}_{1}$ is not
even a subspace of $M_{2}$ (but, can be called a subset). Thus it is not
possible for the mapping $\varphi $ and its inverse one $\varphi ^{-1}$ to
be one-to-one, even in the local sense.
\end{definition}

Naturally the following question appears: how can the image subset $%
\overline{M}_{1}$ be endowed with the property to be a sub-manifold?

\begin{definition}
\cite{Boothby} The subset $\ \overline{M}_{1}$ can be a sub-manifold if it
is the image in $M_{2}$ of a one-to-one immersion and $\varphi :M_{1}\mapsto 
\overline{M}_{1}$ establishes an one-to-one correspondence between $M_{1}$
and the subset $\varphi (M_{1})=\overline{M}_{1}$, i.e. the mapping $\varphi
:M_{1}\rightarrow \overline{M}_{1}$ is a diffeomorphism. Now $\overline{M}%
_{1}$ is called an immersed sub-manifold \ $\overline{M}_{1}\in M_{2}$ and
is endowed with the topology and $C^{\infty }$ (infinitely differentiable)
structure of the manifold $M_{2}$.
\end{definition}

Let us remind also that a mapping $\varphi $ is diffeomorphism if $\varphi $
and its inverse one $\varphi ^{-1}$ are smooth \cite{ArnoldOrdDiff}.

The sub-manifold property of $\overline{M}_{1}$ also is related to the
existence of preferred coordinates relative to $\overline{M}_{1}$ \cite%
{Boothby}. This means that each point $p\in \overline{M}_{1}$ has a
coordinate neighborhood $U,\varphi $ on $M_{2}$ such that $1$. $\varphi
(p)=(0,0......0)$; $2$. $\varphi (U)=C_{\epsilon }^{m}(0)$ and $3$. $\varphi
(U\cap \overline{M}_{1})=\{x\in C_{\epsilon }^{m}(0)\mid
x^{n_{1}+1}=.........x^{n_{2}}=0\}$. The last equality means that the
immersion $\varphi :M_{1}\rightarrow M_{2}$ on $V=$ $U\cap \overline{M}_{1}$
in terms of the preferred local coordinates will be given by 
\begin{equation}
(x^{1},x^{2},......x^{n_{1}})\rightarrow
(x^{1},x^{2},......x^{n_{1}},0,0,.......0)  \label{Obz85A1}
\end{equation}%
and the zeroes correspond to $x^{n_{1}+1}=.........x^{n_{2}}=0$.

Evidently, immersions and embeddings are related to the different notions of
diffeomorphism and homeomorphism. But in some sense they can also be related
- if $\varphi :M_{1}\mapsto M_{2}$ is an immersion, then each point $p\in
M_{1}$ has a neighborhood $U$ such that $\varphi \mid U$ (the restriction of
the mapping $\varphi $ on $U$) is an embedding of $U$ in $M_{2}$ \cite%
{Boothby}.

An example about an immersion is the mapping $F:R\rightarrow R^{3}:$ \ $%
F(t)=(\cos 2\pi t,\sin 2\pi t,t)$, representing the helix, lying on the unit
cylinder, whose axis is the $x_{3}-$axis in $R^{3}$ \cite{Boothby}. The
equation of the helix is $\cos ^{2}2\pi t+\sin ^{2}2\pi t=1$ \ \ , \ $%
x_{3}=t $. It can be checked that the rank of the Jacobian matrix $\frac{%
dF(t)}{dt}\mid _{t=1}$is equal to one. In some sense, this equation is a
simplified model of the equation $x^{2}+y^{2}+z^{2}=r^{2}=\left( \frac{%
a\left( 1-e^{2}\right) }{1+e\cos f}\right) ^{2}$, which can be derived from
the equations (\ref{K1})-(\ref{K3}) in this paper.

\subsection{Embeddings, regular points and sub-manifolds, critical points and
zero measure manifolds, Sards theorem}

\subsubsection{Embedding as a partial case of the immersion}

Next the embedding shall be defined as a partial case of the immersion, the
only difference will be that now the notion of "diffeomorphism" shall be
replaced by "homeomorphism". Homeomorphicity is related to the notion of
continuity, according to the definition in \cite{SchwarzPhys}.

\begin{definition}
\cite{Boothby} \ An embedding is a one-to-one immersion $\varphi
:M_{1}\rightarrow M_{2}$, which is a homeomorphism $M_{1}$ onto its image $%
\varphi (M_{1})=\overline{M}_{1}$ and the image has a topology as a subspace
of $M_{2}$. The image $\overline{M}_{1}$ is called an embedded manifold.
\end{definition}

Now, after the definition of the notion of sub-manifold, it turns out to be
possible to define the notion about a regular manifold, which later on will
be defined in a different way, by means of its relation to the critical
points.

\subsubsection{Regular sub-manifolds}

\begin{definition}
\cite{Boothby} A regular sub-manifold of a $C^{\infty }$ manifold $M_{2}$ is
any subspace of $M_{1}$ with the sub-manifold property and with the $%
C^{\infty }$ structure that the corresponding preferred coordinate
neighborhoods determine on it.
\end{definition}

Thus, if the $C^{\infty }$ mapping $\varphi :M_{1}\rightarrow M_{2}$ maps a
sub-manifold of $M_{1}$ with the rank of the mapping $k<n_{1}$, and $q$ is a
point $q\in \varphi (M_{1})$, then $\varphi ^{-1}(q)$ is a regular
sub-manifold of $M_{1}$ of dimension $n_{1}-k$.

In other words, the inverse mapping $\varphi ^{-1}(q)$ maps the sub-manifold
of $M_{2}$ not into the initial sub-manifold of $M_{1}$ with a dimension $k$%
, but to a regular sub-manifold of $M_{1}$ of dimension $(n_{1}-k)$. But
then, how the initial sub-manifold of $M_{1}$ with dimension $k$ can be
characterized, if the inequality between the ranks $k<n_{1}$ is fulfilled?

We shall consider now the general case of the immersion $\varphi
:M_{1}\rightarrow M_{2}$ with the dimensions of the corresponding manifolds $%
n_{1}$ and $n_{2}$, so that $n_{1}<n_{2}$. Note that here the sign "$=$" is
purposefully omitted, compared to the definition of immersion (see (\ref%
{Obz84})).

\subsubsection{Zero measure sub-manifolds and manifolds}

\begin{definition}
\cite{Golubit} Let $\varphi :M_{1}\rightarrow M_{2}$ is a differentiable
mapping of class $C^{1}$. Then the set $\varphi (M_{1})$ has a zero measure
in $M_{2}$. This also means that the image $\varphi (M_{1})$ does not fill
out the whole manifold $M_{2}$ \cite{DubrFomVol2}.
\end{definition}

The notion of "zero measure" is closely related also with the Lebesque
measure theory, necessary also for introduction of the Sard theorem. For the
purpose the definition of an $n-$dimensional cube is necessary.

\begin{definition}
\cite{Golubit} Let $a=(a_{1},a_{2},....,a_{n})$ and $%
b=(b_{1},b_{2},.....b_{n})$ are points in the $n-$dimensional space $R^{n}$,
such that $a_{i}<b_{i}$ $(1\leq i\leq n)$ and let $C(a,b)$ denotes the open
cube 
\begin{equation}
\left\{ \left( t_{1},t_{2},....t_{n}\right) \in R^{n}\mid a_{i}<t_{i}<b_{i}%
\text{ \ , \ }1\leq i\leq n\right\} \text{ \ .}  \label{Obz86}
\end{equation}%
Let the volume of the cube $C(a,b)$ is determined as 
\begin{equation}
vol\left[ C(a,b)\right] =(b_{1}-a_{1})(b_{2}-a_{2}).....(b_{n}-a_{n})\text{
\ \ .}  \label{Obz87}
\end{equation}
\end{definition}

\begin{definition}
\cite{Golubit} (about a set with zero measure) Let $S$ is a subset of $R^{n}$%
. The subset $S$ is said to have a zero measure if the following two
conditions are fulfilled:

A. For every $\epsilon >0$ a covering of the subset $S$ exists with a
countable number of open cubes $C_{1},C_{2}.......$such that $%
\sum\limits_{i=1}^{\infty }vol[C_{i}]<\epsilon $.

B. Let $X$ is a differentiable $n-$dimensional manifold and $S$ is a subset
of $X$. Let also a countable number of charts $\varphi
_{i}:U_{i}\longrightarrow R^{n}$ exists such that $U_{1},U_{2}....$cover $S$
and for every value of $i$ the set $\varphi _{i}(U_{i}\cap S)$ has a zero
measure. Then an unification of sets with zero measure will again give a set
with zero measure.
\end{definition}

This definition is used in order to define the notion about the critical
points and also the well-known Sards theorem in differential geometry. For
the definition of zero measure manifolds, see also the monograph by
Sternberg \cite{Sternberg}.

\subsubsection{Critical points, Sards theorem in differential geometry and
Morse functions}

\begin{theorem}
\cite{DubrFomVol2} (Sards theorem) Let a smooth differentiable mapping $%
\varphi :M_{1}\rightarrow M_{2}$ is defined with $\dim M_{1}=n_{1}<$ $\dim
M_{2}=n_{2}$ and let $C(\varphi )$ denotes the set of those points $x\in
M_{1}$, for which the differential $d\varphi _{x}:T_{x}\rightarrow
T_{\varphi (x)}$ has a rank, less than \ $n_{2}=\dim M_{2}$. The rank of the
differential $d\varphi _{x}$ is in fact the rank of the Jacobian matrix $%
\left( \frac{\partial \varphi ^{\alpha }}{\partial x^{\beta }}\right) $ (see
also (\ref{Obz84})), $T_{x}$ and $T_{\varphi (x)}$ denote the vectors in the
tangent spaces $TM_{1}$ and $TM_{2}$ of the manifolds $M_{1}$ and $M_{2}$.
The set $C(\varphi )$ will be called the set of the critical points for the
mapping $\varphi $ and the set $C(\varphi )$ has a zero measure in $M_{2}$.
\end{theorem}

Also, let in the neighborhood of the point $p\in M_{1}$ a local coordinate
system $x^{1},x^{2},.....x^{n}$ is chosen, and $f$ $\ $is a real smooth
function on $M_{1}$. Then the point $p$ is said to be a critical point for
the function $f$ $\ $if the following condition is fulfilled \cite{Milnor} 
\begin{equation}
\frac{\partial f}{\partial x^{1}}(p)=\frac{\partial f}{\partial x^{2}}%
(p)=.........=\frac{\partial f}{\partial x^{n}}(p)=0  \text{ \
\ .}  \label{Obz88}
\end{equation}%
Further this definition shall be applied with respect to the mapping $%
F:(f_{1},f_{2})\rightarrow (x,y,z)=(x_{2}-x_{1},y_{2}-y_{1},z_{2}-z_{1})$
for the celestial transformation (\ref{K1}) - (\ref{K3}), when a
two-dimensional manifold $M_{1}$ of the variables $f_{1},f_{2}$ (i.e.$\dim
M_{1}=2$) will be mapped to the three-dimensional manifold $M_{2}$ with
dimension $\dim M_{2}=3$, defined by the cartesian coordinates $(x,y,z)$.

Critical points can be degenerate and non-degenerate ones.

\begin{definition}
\cite{DubrFomVol2} A critical point $x\in M$ of the smooth function $f(x)$,
defined on the manifold $M$ is called a non-degenerate one, if the matrix $%
a_{ij}=\frac{\partial ^{2}f(x)}{\partial x_{i}\partial x_{j}}$ is a
non-degenerate one, i.e. $\det \parallel a_{ij}\parallel \neq 0$. The
quadratic form is invariant with respect to the choice of the point $x$ and
is called a second differential at the given point. In local coordinates $%
x^{i}$ with a centre at a critical point the function $f$ can be written as 
\cite{Vasiliev} 
\begin{equation}
f(x)=f(0)+\sum\limits_{i\neq j}a_{ij}x_{i}x_{j}+O(\mid x\mid ^{3})\text{ \ \
\ .}  \label{Obz88AA1}
\end{equation}%
A critical point is called a Morse one, if the quadratic form \ $a_{ij}$ is
a non-degenerate one. If all the critical points are Morse ones, the
function $f$ is called a Morse function.
\end{definition}

By means of a suitable choice of the local coordinates $%
(x_{1},x_{2},....x_{n})$ the quadratic form can be transformed to a
canonical form 
\begin{equation}
x_{1}^{2}+x_{2}^{2}+....+x_{k}^{2}-x_{k+1}^{2}-.......-x_{n}^{2}\text{ \ \ \ 
}  \label{Obz88AA2}
\end{equation}%
in the neighborhood of the Morse critical point. The index of the \ critical
point is equal to the number of negative terms (i.e. the negative inertial
index))

This is the exact mathematical definition for a critical point and a Morse
function. However, the problem is whether this definition is sufficient in
view of the determination of Morse functions.

The definition above is given because critical points (degenerate and
non-degenerate) will be studied further with respect to the transformations (%
\ref{K1}) - (\ref{K3}) from Kepler parameters $(f,a,e,\Omega ,i,_{.}\omega
_{.})$ to the cartesian coordinates $(x,y,z)$. Here we shall point out only
some aspects of the problem, the detailed study will be presented in
subsequent publications.

The literature on critical points and Morse functions is rather extensive
-only a part of the references on the subject are \cite{DubrFomVol2}, \cite%
{Vasiliev},\cite{TrofimFom}, \cite{SchwarzPhys}, \cite{NashSen}, \cite%
{ShapOlshan}, \cite{Golubit} and especially the monograph \cite%
{NovikovTaimanov}.

\subsubsection{Critical points in the transformation (\protect\ref{K1}) - (%
\protect\ref{K3}). Are there Morse functions with respect to some of the
Kepler parameters?}

All the critical points and Morse functions will not be found in this
paragraph, but we shall pay a special attention to the case, concerning the
derivatives of $x,y$ and $z$ with respect to the $\Omega -$ angle. We shall
restrict ourselves only to the case of one manifold $M$, represented in
cartesian coordinates as $(x,y,z)$ and also in terms of the $6$-Kepler
parameters. From expressions (\ref{K1}) - (\ref{K3}) the derivatives $\frac{%
\partial x}{\partial \Omega }$,$\frac{\partial y}{\partial \Omega }$ and$%
\frac{\partial z}{\partial \Omega }$ can be calculated as follows 
\begin{equation}
\frac{\partial x}{\partial \Omega }=r\left[ -\sin \Omega \cos \left( \omega
+f\right) -\cos \Omega \sin \left( \omega +f\right) \cos i\right] =-y\text{
\ \ ,}  \label{Obz88AA3}
\end{equation}

\begin{equation}
\frac{\partial y}{\partial \Omega }=r\left[ \cos \Omega \cos \left( \omega
+f\right) -\sin \Omega \sin \left( \omega +f\right) \cos i\right] =x\text{ \
\ ,}  \label{Obz88AA4}
\end{equation}%
\begin{equation}
z=\frac{a(1-e^{2})}{1+e\cos f}\sin (\omega +f)\sin i\text{ \ \ \ }%
\Rightarrow \frac{\partial z}{\partial \Omega }=0\text{ \ \ \ \ .}
\label{Obz88AA5}
\end{equation}%
Consequently, one will have $\frac{\partial x}{\partial \Omega }=\frac{%
\partial y}{\partial \Omega }=\frac{\partial z}{\partial \Omega }=0$ at the
critical point (in terms of the cartesian coordinates) $%
(x,y,z)_{crit}=(0,0,z)$. Let us see which are the critical values in terms
of the Kepler parameters. From the first equation (\ref{Obz88AA3}) one can
express 
\begin{equation}
\cos i=-\frac{\sin \left( \omega +f\right) }{tg\Omega \cos \left( \omega
+f\right) }  \label{Obz88AA6}
\end{equation}%
and substitute in the second equation $\frac{\partial y}{\partial \Omega }=0$%
. The result will be 
\begin{equation}
\cos \Omega \cos \left( \omega +f\right) +\sin \Omega \sin \left( \omega
+f\right) \frac{\sin \left( \omega +f\right) }{tg\Omega \cos \left( \omega
+f\right) }  \label{Obz88AA7}
\end{equation}%
\begin{equation}
=\cos \Omega \left[ \frac{\cos ^{2}\left( \omega +f\right) +\sin ^{2}\left(
\omega +f\right) }{\cos \left( \omega +f\right) }\right] =\frac{\cos \Omega 
}{\cos \left( \omega +f\right) }=0\text{ \ \ .}  \label{Obz88AA8}
\end{equation}%
So if $\cos \left( \omega +f\right) \neq 0$, to the critical point $(0,0,z)$
will correspond the critical point $(\cos \Omega =0,\cos i=0)$ This might
seem to be the critical point, but in fact it is not if besides the
condition $\cos \left( \omega +f\right) \neq 0$, the condition $\sin (\omega
+f)\neq 0$ should also be imposed. So the critical points are $\Omega =k%
\frac{\pi }{2}$, but additionally $\omega +f\neq m\frac{\pi }{2}$ and also $%
\omega +f\neq p\pi $ ($m=1,2,,,,;p=0,1,2$).

Note that in the definition for critical points it was not required that the
zero point $(x,y,z)=(0,0,0)$ should not be a critical point. If it is
required (i.e. postulated) that $(x,y,z)\neq (0,0,0)$ (in the sense of all
the three coordinates to be different from zero), then $(0,0,z)$ will be the
critical point for $\Omega =k\frac{\pi }{2}$.

However, because of the equalities $\frac{\partial ^{2}x}{\partial \Omega
^{2}}=-\frac{\partial y}{\partial \Omega }$ and $\frac{\partial ^{2}y}{%
\partial \Omega ^{2}}=\frac{\partial x}{\partial \Omega }$ it is seen that
the critical points are degenerate ones and consequently the function $%
f=(x,y,z)$ is not a Morse function. This is admissable in view of the Sards
theorem, since if in a manifold the set of critical points can have a zero
measure, then every smooth mapping $f:(f,a,e,\Omega ,i,_{.}\omega
_{.})\rightarrow (x,y,z)$ should admit non-critical values.

The last assertion about the function $f$ $\ $being a non-Morse function is
consistent with a theorem in \cite{DubrFomVol2} that each Morse function on
a compact manifold has only a finite number of critical points. But in the
given case, if one assumes that $f=(x,y,z)$ is a Morse function, this will
contradict with the infinite number of values for $\Omega =k\frac{\pi }{2}$,
since the number $k$ can be arbitrary large.

\subsubsection{Regular points}

In contrast to the set $C(\varphi )$ of critical points, a point $p\in M_{1}$
will be called a regular point for the smooth mapping $\varphi $, if it is
not a critical one, i.e. the rank of the mapping between the vectors in the
corresponding tangent spaces $d\varphi :TM_{1}\rightarrow TM_{2}$ is equal
to $n_{2}=\dim M_{2}$. Also, the point $q\in M_{2}$ will be a regular point
if all its inverse images are regular points in the manifold $M_{1}$ (for
the case when $\varphi ^{-1}(q)=\varnothing $ \ the point $q$ will also be
considered to be a regular point) \ \cite{DubrFomVol2}.

For the definition of critical and regular points, see also \cite{Sternberg}.

\subsubsection{Submersions}

Let again a differentiable mapping $\varphi :M_{1}\rightarrow M_{2}$ is
given, but contrary to the case of immersion, when the ranks of the
manifolds were defined by the inequality $rankM_{1}=n_{1}\leq
rankM_{2}=n_{2} $, for the case of submersion the inequality for the ranks
is defined as $rankM_{1}=n_{1}\geq rankM_{2}=n_{2}$ \cite{Boothby}, \cite%
{DubrFomVol2}. As noted in \cite{Boothby}, the case $%
rankM_{1}=n_{1}=rankM_{2}=n_{2}=rank\varphi $ is a necessary condition for
establishing a diffeomorphism between the two manifolds, for the case of
submersion it is impossible.

\subsubsection{Critical points defined for both cases of immersions and
submersions}

\begin{definition}
\cite{ArnoldVarchenko} Let $M$ and $N$ are manifolds. The point $x\in M$ \
is called a critical point for the smooth mapping $f:M\rightarrow N$ if the
rank of the mapping $f_{\ast x}$ between the tangent spaces of $M$ and $N$ 
\begin{equation}
f_{\ast x}:T_{x}M\rightarrow T_{f(x)}N  \label{Obz88A1}
\end{equation}%
at this point is smaller than the maximum possible, i.e. small than the
smaller between the dimensions of the manifolds $M$ and $N$ 
\begin{equation}
rank\text{ }f_{\ast x}<\min (\dim M,\dim N)\text{ \ \ \ \ .}  \label{Obz88A2}
\end{equation}
\end{definition}

Evidently, when $\dim M>\dim N$ (the case of submersion) this definition can
be applied for the definition of critical points.

\subsubsection{Submersions, the Whitney theorem and some applications to the
current research}

The case of submersion is important for the presently investigated case of
the mapping 
\begin{equation}
F:M_{1}(f_{.},a_{.},e_{.},\Omega _{.},i_{.},\omega _{.})\rightarrow N(x,y,z)%
\text{ \ \ \ \ ,}  \label{Obz89}
\end{equation}
when the celestial transformations (\ref{K1}) - (\ref{K3}) map the manifold $%
M_{1}$, related to the six Kepler parameters $(f,a,e_{.},\Omega
_{.},i_{.},_{.}\omega _{.})$ to the manifold $N$, related to the three
Cartesian coordinates $(x,y,z)$.

In the present investigation we would like to consider two sets of Kepler
parameters $(f_{1},a_{1},e_{1.},\Omega _{1.},i_{1.},_{.}\omega _{1.})$ and $%
(f_{2},a_{2},e_{2.},\Omega _{2.},i_{2.},_{.}\omega _{2.})$, characterizing
the motion respectively of the signal-emitting and signal-receiving
satellite. 

These two sets will define two manifolds $\ \
M_{1}(f_{1},a_{1},e_{1.},\Omega _{1.},i_{1.},_{.}\omega _{1.})$ and  $%
M_{2}(f_{2},a_{2},e_{2.},\Omega _{2.},i_{2.},_{.}\omega _{2.})$. Since the
aim will be to apply a variational formalism, based on all the variables,
the question is whether it will be possible to define a mapping $\overline{F}
$ of a manifold $M$, related to all the $12$ Kepler parameters to the
manifold $N$, related to the two sets of cartesian coordinates $%
(x_{1},y_{1},z_{1})$ and $(x_{2},y_{2},z_{2})$? These two mappings will be
the submersions 
\begin{equation}
\overline{F}_{1}:\overline{M}_{1}(f_{1},a_{1},e_{1.},\Omega
_{1.},i_{1.},_{.}\omega _{1},f_{2},a_{2},e_{2.},\Omega
_{2.},i_{2.},_{.}\omega _{2.})\rightarrow \overline{N}%
_{1}(x_{1},y_{1},z_{1},x_{2},y_{2},z_{2})  \label{Obz90}
\end{equation}%
or 
\begin{equation}
\overline{F}_{1}:\overline{M}_{1}(f_{1},a_{1},e_{1.},\Omega
_{1.},i_{1.},_{.}\omega _{1},f_{2},a_{2},e_{2.},\Omega
_{2.},i_{2.},_{.}\omega _{2.})\rightarrow \overline{N}%
_{2}(x=x_{1}-x_{2},y=y_{1}-y_{2},z=z_{1}-z_{2})\text{ \ }.  \label{Obz91}
\end{equation}%
In the first case (\ref{Obz90}) this is a submersion of a $12-$dimensional
manifold $\overline{M}_{1}$ into a $6-$dimensional manifold $\overline{N}_{1}
$ and in the second case (\ref{Obz91}) - a submersion of the manifold $%
\overline{M}_{1}$ into a $3-$dimensional manifold $\overline{N}_{2}$. Naturally, the $12-$%
dimensional manifold $\overline{M}_{1}$ will exist if it is possible to
embed or immerse both manifolds $M_{1}(f_{1},a_{1},e_{1.},\Omega
_{1.},i_{1.},_{.}\omega _{1.})$ and $M_{2}(f_{2},a_{2},e_{2.},\Omega
_{2.},i_{2.},_{.}\omega _{2.})$ into $\overline{M}_{1}$. Such a possibility
is guaranteed by the known Whitney theorem, which can be proved by using the
Sards theorem .

\begin{theorem}
(Whitney) Each connected smooth and closed $n$-dimensional manifold \ can be
smoothly embedded in an $R^{2n+1}-$dimensional Euclidean space or it can be
immersed in an $R^{2n}-$dimensional Euclidean space. .
\end{theorem}

The proof of this theorem in \cite{DubrFomVol2} is rather complicated and is
based on the subsequent projection of a smooth embedded manifold $M\subset
R^{n}$ on hyperplanes, so that at each projection the dimension of the
enclosing space is diminished by one. As a result of an orthogonal
projection $\pi _{l_{0}}:M\rightarrow R_{t_{0}}^{N-1}$ a smooth immersion is
obtained so that a point and a forbidden direction (straight line) with a
total of $n+(n-1)$ parameters should have a zero measure. According to Sards
theorem this is possible when $2n-1<N-1$ or $2n<N$ - this is the second case
of immersion. Analogous considerations are presented in \cite{DubrFomVol2}
for the other case of an embedding in a $R^{2n+1}-$ dimensional space.

In the concrete case, we have $n=6$ (the dimension of the manifold $M_{1}$
or $M_{2}$), so $M_{1}$ (or $M_{2}$) is possible to be smoothly embedded to
a $2\times 6+1=13$ dimensional manifold or to be smoothly immersed in an $%
2\times 6=12-$dimensional manifold. This will be exactly the manifold,
defined on the two sets of Kepler parameters  
\begin{equation*}
\overline{M}_{1}(f_{1},a_{1},e_{1.},\Omega _{1.},i_{1.},_{.}\omega
_{1},f_{2},a_{2},e_{2.},\Omega _{2.},i_{2.},_{.}\omega _{2.})\text{ \ \ .}
\end{equation*}
For the first case of $2\times 6+1=13-$dimensions, the manifold is
supplemented by one more function - this can be, for example, the distance $%
R_{AB}$ between two points $A\subset M_{1}$ and $B\subset M_{2}$. Also, this
distance will be related to the parameter $s$, because each of the Kepler
parameters in the two sets will depend on it, respectively the two sets of
coordinates $(x_{1},y_{1},z_{1},x_{2},y_{2},z_{2})$ will also depend on it.
So the mapping (\ref{Obz90}) will be 
\begin{equation}
\overline{F}_{1}:\overline{M}_{1}(\Gamma ^{(1)}(s),\Gamma
^{(2)}(s),R_{AB}(s))\rightarrow \overline{N}%
(x_{1},y_{1},z_{1},x_{2},y_{2},z_{2})  \label{Obz92}
\end{equation}%
and $\Gamma ^{(1)}(s),\Gamma ^{(2)}(s)$ denote the two sets of Kepler
parameters 
\begin{equation}
\Gamma ^{(1)}(s)\equiv (f_{1}(s),a_{1}(s),e_{1.}(s),\Omega
_{1.}(s),i_{1.}(s),_{.}\omega _{1}(s))\text{ , \ }  \label{Obz92A1}
\end{equation}%
\begin{equation}
\Gamma ^{(2)}(s)\equiv (\text{\ }f_{2}(s),a_{2}(s),e_{2.}(s),\Omega
_{2.}(s),i_{2.}(s),_{.}\omega _{2.}(s))\text{ \ \ .}  \label{Obz92A2}
\end{equation}

Since the coordinates $x_{1},y_{1},z_{1},x_{2},y_{2},z_{2}$ do not introduce
any new parameters or functions, the dimension of the manifold $\overline{N}$
remains $6$.

\subsubsection{Homotopy theory and combining the approaches of algebraic
topology and variational calculus}

The mapping (\ref{Obz92}) is in fact an attempt to combine the formalism of
variational calculus with the topological ideas. Let us remind the standard
definition about homotopy \cite{Boothby}, \cite{DubrFomVol2} (see also the
modern textbook \cite{Nakahara} on the application of topology in physics)
in algebraic topology, which on the one hand represents a mapping between
topological spaces $X$ and $Y$, but one the other hand \ represents also an
immersion $X\times I\rightarrow Y$, where $I$ is an interval.

\begin{definition}
(\cite{Boothby}) Let $F,G$ are continuous mappings from the topological
space $X$ to the topological space $Y$ and $I=[0,1]$ is an unit interval.
Then $F$ is homotopic to $G$ if \ there is a continuous mapping (homotopy) 
\begin{equation}
H:X\times I\mapsto Y\text{ \ \ \ \ ,}  \label{Obz93}
\end{equation}%
which satisfies the condition $F(x)=H(x,0)$ and $G(x)=H(x,1)$ for all $x\in
X $.
\end{definition}

A slightly different definition is given in \cite{Nakahara}, where $%
a,b:I\rightarrow X$ are defined to be loops (i.e. the interval $I$ is a
two-parametric one - the parameters are $s$ and $t$), $F:I\times
I\rightarrow X$ is a continuous mapping and it is required that 
\begin{equation}
F(s,0)=a(s)\text{, \ \ \ \ \ \ \ \ \ \ \ \ }F(s,1)=b(s)\text{ \ \ \ \ }%
\forall s\in I\text{ \ \ \ \ , }  \label{Obz94}
\end{equation}%
\begin{equation}
F(0,t)=F(1,t)=x\text{ \ \ \ \ \ \ \ }\forall t\in I\text{ \ \ .}
\label{Obz95}
\end{equation}%
Then the connecting map $F$ is called a homotopy between $a$ and $b$
(denoted as $a\sim b$). 

The differences in the presently applied approach are the following ones:

1. The interval $I$ will not be the unit interval, but the interval $I=\left[
s_{a},s_{b}\right] $, where the values $s=s_{a}$ and $s=s_{b}$ will
correspond to the values of the parameter $s$ for the Kepler parameters 
\begin{equation}
\Gamma ^{(1)}(s)\equiv (f_{1}(s_{a}),a_{1}(s_{a}),e_{1.}(s_{a}),\Omega
_{1.}(s_{a}),i_{1.}(s_{a}),_{.}\omega _{1}(s_{a}))  \label{Obz96}
\end{equation}
of the signal-emitting satellite and 
\begin{equation}
\Gamma ^{(2)}(s)\equiv (f_{2}(s_{b}),a_{2}(s_{b}),e_{2.}(s_{b}),\Omega
_{2.}(s_{b}),i_{2.}(s_{b}),_{.}\omega _{2.}(s_{b}))  \label{Obz97}
\end{equation}
of the signal-receiving satellite (both are moving along different elliptic,
space - distributed orbits).

2. Further it will also be supposed that both ends of the signal trajectory
are dynamically moving and will be required to find the optimal trajectory
of the signal (i.e.the optimal positions of both satellites) for
transmission between the satellites. If finding the signal trajectory takes
into account the General Relativity effects and the Fermats principle for
minimal time of propagation of the signal, finding the optimal position of
the satellites on both orbits will require the application of a
multi-parameter optimization theory. In other words, both ends of the signal
trajectory will be dynamical, not static and fixed ones. This formalism will
be developed in subsequent papers, but it may only be mentioned that for a
long time ago a formalism in variational calculus has been developed - the
s.c. "variational calculus with moving boundaries" \cite{GelfandFomin} and
many other monographs.

3. It will be impossible to use the direct product $X\times I$, because if $%
I=[s_{a},s_{b}]$, the manifold $X$ will depend on the parameter $s$. A new
problem here is formulated when an immersion or embedding is realized
between two manifolds $f:X\rightarrow Y$, if they are considered at the same
time to be topological spaces and each of their elements depends also on the
parameter $s$. 

\section{Higher order variational formalism, based on the Stokes theorem and
differential forms}

\bigskip The Stokes theorem is the basic theoretical tool, which shall be
used further. It is based on three important notions - differential forms of
degree $n$ (or simply $n-$form), orientation of an $n-$manifold and the
previously defined notion about the boundary $\partial M$ of a manifold $M$.

\subsection{Definition of the Stokes theorem}

\begin{definition}
\cite{Boothby} An exterior differential form of degree $n$ ($n-$form) is an
alternating covariant tensor field of order $n$.
\end{definition}

\begin{definition}
\cite{Boothby} The manifold $M$ is orientable if a $C^{\infty }$ $n-$form $%
\omega =\sum\limits_{i}f_{i}\varphi _{i}^{\ast }(dx^{1}\Lambda
dx^{2}\Lambda .....dx^{n})$ can be defined, which is not zero at any point.
Then it is said that $M$ is oriented by the choice of $\omega $. This means
that the orientation of $M$ determines the orientation of the boundary $%
\partial M$. This holds for regular domains (sub-manifolds) of $M$.
\end{definition}

Then the Stokes theorem, relating an integral on the manifold $M$ to an
integral on the boundary is defined as follows

\begin{theorem}
\cite{Boothby} Let $M$ is an $n-$dimensional oriented manifold with a
boundary $\partial M$, $\omega $ is a differential form of the $(n-1)-$%
degree with a compact support. Then the generalized Stokes formulae can be
expressed as 
\begin{equation}
\int_{M}d\omega =\int_{\widetilde{\partial }M}i^{\ast }\omega \text{ \ \ \ . 
}  \label{Obz98}
\end{equation}%
In the last equality $\widetilde{\partial }M$ denotes the \ boundary $%
\partial M$ with its different possible orientations. This means that when $%
n $ is even, the orientation, induced by $M$ is positive, i.e. $\partial M$,
but when $n$ is odd, then the orientation will be the opposite one, i.e. $%
-\partial M$. Both cases can be united as $\widetilde{\partial }%
M=(-1)^{n}\partial M$. In (\ref{Obz98}) $i^{\ast }\omega $ denotes the
restriction of $\omega $ to $\partial M$ since $i:\partial M\rightarrow M$
is an inclusion mapping. In a slightly different manner, the above equality
is written as $(-1)^{n}\int_{M}d\omega =\int_{\partial M}\omega $ (see also
the definitions in \cite{DubrFomVol2}, \cite{MishtFomCourse}, \cite{NashSen}%
, \cite{SchwarzPhys} and many others). In the second term, instead of $%
i^{\ast }\omega $, many authors write only $\omega $.
\end{theorem}

\subsection{Application of a variational formalism, based on higher order
differential forms and the Stokes theorem}

Now let us describe the variational formalism, developed in this paper. The
idea is to start from the one - form 
\begin{equation}
\omega _{1}=Adx+Bdy+Cdz\text{ \ \ \ \ ,}  \label{Obz99}
\end{equation}%
where $A,B,C$ are the coefficient functions in the variational action (one can
take the action (\ref{Obz20})). In other words, $A,B,C$ will be the
Euler-Lagrange equations as a result of the variation in terms of the \ $x,y$
and $z$ variables 
\begin{equation}
A\equiv \frac{\delta L}{\delta x}+\lambda _{1}(\widetilde{T}-G(f))\text{ \ ;
\ }B\equiv \frac{\delta L}{\delta y}+\lambda _{2}(\widetilde{T}-G(f))\text{
\ ; \ \ }C\equiv \frac{\delta L}{\delta z}+\lambda _{3}(\widetilde{T}-G(f))%
\text{ \ \ ,}  \label{Obz100}
\end{equation}%
supplemented by the Lagrange function multipliers $\lambda
_{1}(x,y,z),\lambda _{2}(x,y,z),\lambda _{3}(x,y,z)$. \ In (\ref{Obz100}) $%
\frac{\delta L}{\delta x},\frac{\delta L}{\delta y},\frac{\delta L}{\delta z}
$ are the generalized variational derivatives 
\begin{equation*}
\frac{\delta L}{\delta x}\equiv \frac{\partial L}{\partial x}-\frac{d}{ds}%
\frac{\partial L}{\partial \overset{.}{x}}\text{ \ , \ }\frac{\delta L}{%
\delta y}\equiv \frac{\partial L}{\partial y}-\frac{d}{ds}\frac{\partial L}{%
\partial \overset{.}{y}}\text{ \ ,}
\end{equation*}%
$\ $%
\begin{equation}
\text{ \ }\frac{\delta L}{\delta z}\equiv \frac{\partial L}{\partial z}-%
\frac{d}{ds}\frac{\partial L}{\partial \overset{.}{z}}\text{ \ \ .}
\label{Obz101}
\end{equation}%
The function $\widetilde{T}$ is equal to the function $G(f)$ (\ref{K4}),
calculated in the paper \cite{Bog1} by means of higher-order elliptic
integrals - the corresponding expressions are given in Appendix B. The
motivation to include the function $(\widetilde{T}-G(f))$ as constraints is
that if in this higher order variational formalism the propagation time is
calculated, then in the partial case of only $f$ $\ $being the dynamical
parameter, the result should coincide with $\widetilde{T}$. This would mean
also that the derivatives of $\widetilde{T}$ \ with respect to the other $5$
Kepler parameters should be zero. The notation for $\overset{.}{x}$, $\overset{.}{y}
$ and $\overset{.}{z}$ means, as usual, derivatives with respect to the
parameter $s$.

In the usual variational formalism, $A,B$ and $C$ are the usual
Euler-Lagrange equations (plus the constraint function). In this case $A,B,C$
will not be equal to zero (see also \cite{TrofimFom}) and they will be used
for the derivation of the two-form by means of the Stokes theorem for the
three-dimensional case 
\begin{equation}
\int\limits_{\partial M}\left[ Adx+Bdy+Cdz\right] =\int\limits_{M}\left[
Pdy\Lambda dz+Qdz\Lambda dx+Rdx\Lambda dy\right]  \label{Obz102}
\end{equation}%
where $P=P(x,y,z)$, $Q=Q(x,y,z)$, $R=,R(x,y,z)$ are the expressions 
\begin{equation*}
P(x,y,z)\equiv \frac{\partial C}{\partial y}-\frac{\partial B}{\partial z}%
\text{ \ , }Q(x,y,z)\equiv \frac{\partial A}{\partial z}-\frac{\partial C}{%
\partial x}\text{ \ \ ,}
\end{equation*}%
\begin{equation}
\text{ \ \ }R(x,y,z)\equiv \frac{\partial B}{\partial x}-\frac{\partial A}{%
\partial y}\text{ \ .}  \label{Obz103}
\end{equation}%
It should be mentioned that these expressions represent a partial case in the
definition of the absolute integral invariant of the second order \cite%
{TrofimFom} 
\begin{equation}
J=\int\limits_{\widetilde{\Omega }}\sum\limits_{i,j}\left( \frac{\partial
F_{i}}{\partial x_{j}}-\frac{\partial F_{j}}{\partial x_{i}}\right) \delta
x_{i}\delta x_{j}\text{ \ \ \ ,}  \label{Obz104}
\end{equation}%
where 
\begin{equation}
J=\int\limits_{\partial \widetilde{\Omega }}\left( F_{1}\delta
x_{1}+......+F_{n}\delta x_{n}\right)  \label{Obz105}
\end{equation}%
is called the relative integral invariant. Let us also remind that $%
x=x_{2}-x_{1},\ y=y_{2}-y_{1},z=z_{2}-z_{1}$.

Further, let us introduce the notations 
\begin{equation}
(q_{x},q_{y},q_{z})=(\frac{\partial L}{\partial x},\frac{\partial L}{%
\partial y},\frac{\partial L}{\partial z})\text{ \ \ , \ \ }\overrightarrow{p%
}=(p_{x},p_{y},p_{z})=(\frac{\partial L}{\partial \overset{.}{x}},\frac{%
\partial L}{\partial \overset{.}{y}},\frac{\partial L}{\partial \overset{.}{z%
}})\text{ \ \ ,}  \label{Obz106}
\end{equation}%
\begin{equation}
\overset{.}{\overrightarrow{p}}=(\overset{.}{p_{x}},\overset{.}{p_{y}},%
\overset{.}{p_{z}})=(\frac{d}{ds}\frac{\partial L}{\partial \overset{.}{x}},%
\frac{d}{ds}\frac{\partial L}{\partial \overset{.}{y}},\frac{d}{ds}\frac{%
\partial L}{\partial \overset{.}{z}})\text{ \ \ .}  \label{Obz107}
\end{equation}%
Then it can easily be found that 
\begin{equation}
P_{1}\equiv \frac{\partial }{\partial z}\left( \frac{d}{ds}\frac{\partial L}{%
\partial \overset{.}{y}}\right) -\frac{\partial }{\partial y}\frac{d}{ds}%
\frac{\partial L}{\partial \overset{.}{z}}=-\frac{\partial \overset{.}{p_{z}}%
}{\partial y}+\frac{\partial \overset{.}{p_{y}}}{\partial z}=(rot\overset{.}{%
\overset{.}{\overrightarrow{p}})_{x}}\text{ \ \ \ .}  \label{Obz108}
\end{equation}%
Analogously 
\begin{equation}
Q_{1}\equiv \left\{ z\rightarrow x;y\rightarrow z\right\} =(rot\overset{.}{%
\overset{.}{\overrightarrow{p}})_{y}\text{ \ , \ }}R_{1}\equiv \left\{
x\rightarrow y;z\rightarrow x\right\} =\left( rot\text{ }\overset{.}{%
\overrightarrow{p}}\right) _{z}\text{ \ .}  \label{Obz109}
\end{equation}%
Then, after taking into account the expressions (\ref{Obz100}) - (\ref%
{Obz102}) and also (\ref{Obz106}) - (\ref{Obz109}), each of the functions $%
P(x,y,z),Q(x,y,z),R(x,y,z)$ in (\ref{Obz103}) can be represented as a sum of
four components 
\begin{equation}
P\equiv \widetilde{P}_{0}+P_{1}+T^{P}+\widetilde{G}_{1}\text{ \ ,}.Q\equiv 
\widetilde{Q}_{0}+Q_{1}+T^{Q}+\widetilde{G}_{2}\text{ \ ,}  \label{Obz110}
\end{equation}%
\begin{equation}
\text{ \ }R\equiv \widetilde{R}_{0}+R_{1}+T^{R}+\widetilde{G}_{3}\text{ \ \
\ .}  \label{Obz110A1}
\end{equation}%
Each one of the components is calculated to be 
\begin{equation}
\widetilde{P}_{0}\equiv \frac{\partial ^{2}L}{\partial y\partial z}-\frac{%
\partial ^{2}L}{\partial z\partial y}\text{, \ \ }P_{1}=(rot\overset{.}{%
\overrightarrow{p}})_{x}\text{ \ \ , \ \ }  \label{Obz110A2}
\end{equation}%
\begin{equation}
T^{P}\equiv \lambda _{3}\widetilde{T}_{y}-\lambda _{2}\widetilde{T}_{z}\text{
\ , \ }\widetilde{G}_{1}\equiv \lambda _{2}\frac{\partial G}{\partial f}%
\frac{1}{z_{f}}-\lambda _{3}\frac{\partial G}{\partial f}\frac{1}{y_{f}}%
\text{ \ . }  \label{Obz111}
\end{equation}%
Similar are the decompositions for $\widetilde{Q}_{0},Q_{1},T^{Q},\widetilde{%
G}_{2}\ $and$\ \widetilde{R}_{0},R_{1},T^{R},\widetilde{G}_{3}$, taking into
consideration the symmetry properties $(y,z)\rightarrow (z,x)\rightarrow
(x,y)$, $x\rightarrow y\rightarrow z$ and $(\lambda _{3},\lambda
_{2})\rightarrow (\lambda _{1},\lambda _{3})\rightarrow (\lambda
_{2},\lambda _{1})$ (circular symmetry $1\rightarrow 2\rightarrow
3\rightarrow 1$). The functions $\widetilde{T}_{x},\widetilde{T}_{y},%
\widetilde{T}_{z}$ are the partial derivatives of the propagation time $%
\widetilde{T}$ \ (\ref{K4}), according to the formulaes in Appendix B. The
functions $\widetilde{P}_{0}$, $\widetilde{Q}_{0}\equiv \frac{\partial ^{2}L%
}{\partial z\partial x}-\frac{\partial ^{2}L}{\partial x\partial z}$, $%
\widetilde{R}_{0}\equiv \frac{\partial ^{2}L}{\partial x\partial y}-\frac{%
\partial ^{2}L}{\partial y\partial x}$ in (\ref{Obz110}), (\ref{Obz110A1})
and (\ref{Obz110A2}) usually are put equal to zero, since the partial
derivatives $\frac{\partial }{\partial x},\frac{\partial }{\partial y},\frac{%
\partial }{\partial z}$ commute. However, if they are expressed through the
Kepler parameters $(f_{.},a_{.},e_{.},\Omega _{.},i_{.},\omega )$, the
corresponding commutators $\left[ \frac{\partial }{\partial x},\frac{%
\partial }{\partial y}\right] $, $\left[ \frac{\partial }{\partial y},\frac{%
\partial }{\partial z}\right] $ and $\left[ \frac{\partial }{\partial x},%
\frac{\partial }{\partial z}\right] $ will turn out to be nonzero. Let us
remind also that the coordinates of the vector fields $\overrightarrow{u}$
and $\overrightarrow{v}$ can be chosen to be%
\begin{equation}
u_{i}=(f_{.},a_{.},e_{.})\text{ \ \ \ }i=1,2,3\text{ and \ }v_{k}=(\Omega
_{.},i_{.},\omega _{.})\text{ \ \ \ \ \ \ }k=1,2,3\text{ \ \ ,}
\label{Obz112}
\end{equation}%
but also can be chosen to be the two sets of Kepler parameters for the first
(signal-emitting) and the second (signal-receiving) satellite 
\begin{equation}
u_{i}=(f_{.1},a_{1.},e_{1},_{.}\Omega _{1.},i_{1.},\omega _{1})\text{ \ \ \ }%
i=1,2,3,4,5,6\text{ \ }  \label{Obz113}
\end{equation}%
and 
\begin{equation}
\text{\ }v_{j}=(f_{.2},a_{2.},e_{2},_{.}\Omega _{2.},i_{2.},\omega _{2.})%
\text{ \ \ \ \ \ \ }j=1,2,3,4,5,6\text{ \ \ \ \ \ .}  \label{Obz114}
\end{equation}%
We shall give first the proof about the non-commutativity of the partial
derivatives and afterwards we shall remind why the result is natural from
the point of view of the Frobenius theorem, which is presented in almost
every monograph on differential geometry.

\subsection{Non-commutativity of the partial derivatives, when expressed by
derivatives, related to the Kepler parameters}

Let us investigate the case of non-commutativity of the partial derivatives $%
\frac{\partial }{\partial y}$ and $\frac{\partial }{\partial z}$.

\begin{proposition}
The commutator $\left[ \frac{\partial }{\partial y},\frac{\partial }{%
\partial z}\right] \neq 0$ is different from zero, provided that the partial
derivatives are expressed as 
\begin{equation}
\frac{\partial }{\partial z}=\frac{1}{3}\sum\limits_{k=1}^{3}\frac{1}{%
z_{u_{k}}}\frac{\partial }{\partial u_{k}}+\frac{1}{3}\sum\limits_{l=1}^{3}%
\frac{1}{z_{v_{l}}}\frac{\partial }{\partial v_{l}}\text{ \ \ \ ,}
\label{Obz115}
\end{equation}%
\begin{equation}
\text{ \ }\frac{\partial }{\partial y}=\frac{1}{3}\sum\limits_{i=1}^{3}\frac{%
1}{y_{u_{i}}}\frac{\partial }{\partial u_{i}}+\frac{1}{3}\sum%
\limits_{j=1}^{3}\frac{1}{y_{v_{j}}}\frac{\partial }{\partial v_{j}}\text{\
\ \ \ .}  \label{Obz116}
\end{equation}
\end{proposition}

\begin{proof}
The second derivative can easily be calculated to be 
\begin{equation}
\frac{\partial ^{2}}{\partial y\partial z}=\frac{1}{3}\sum\limits_{k=1}^{3}%
\frac{\partial }{\partial y}\left( \frac{1}{z_{u_{k}}}\frac{\partial }{%
\partial u_{k}}\right) +\frac{1}{3}\sum\limits_{l=1}^{3}\frac{\partial }{%
\partial y}\left( \frac{1}{z_{v_{l}}}\frac{\partial }{\partial v_{l}}\right) 
\text{ \ .}  \label{Obz117}
\end{equation}%
Let us denote by $M_{ik}(u_{i},u_{k},y\Longleftrightarrow z)$ the operator 
\begin{equation}
M_{ik}(u_{i},u_{k},y\Longleftrightarrow z)\equiv
M_{ik}(u_{i},u_{k},y,z)-M_{ik}(u_{i},u_{k},z,y)\text{ \ \ ,}  \label{Obz118}
\end{equation}%
where $M_{ik}(u_{i},u_{k},y,z)$ denotes the operator 
\begin{equation}
M_{ik}(u_{i},u_{k},y,z)\equiv \frac{1}{y_{u_{i}}}\frac{\partial }{\partial
u_{i}}\left( \frac{1}{z_{u_{k}}}\frac{\partial }{\partial u_{k}}\right) 
\text{ \ \ .}  \label{Obz118A}
\end{equation}%
Respectively%
\begin{equation}
M_{ik}(u_{i},u_{k},z,y)\equiv \frac{1}{z_{u_{i}}}\frac{\partial }{\partial
u_{i}}\left( \frac{1}{y_{u_{k}}}\frac{\partial }{\partial u_{k}}\right) 
\text{ \ \ .}  \label{Obz118B}
\end{equation}%
Then the operator $M_{ik}(u_{i},u_{k},y\Longleftrightarrow z)$ (\ref{Obz118}%
), which is antisymmetric with respect to the variables $y$ and $z$ can be
rewritten as 
\begin{equation}
M_{ik}(u_{i},u_{k},y\Longleftrightarrow z)\equiv \frac{1}{y_{u_{i}}}\frac{%
\partial }{\partial u_{i}}\left( \frac{1}{z_{u_{k}}}\right) -\frac{1}{%
z_{u_{i}}}\frac{\partial }{\partial u_{i}}\left( \frac{1}{y_{u_{k}}}\right)
\label{Obz119}
\end{equation}%
\begin{equation}
=\frac{1}{y_{u_{i}}}\frac{1}{z_{u_{k}}}\left[ \frac{1}{z_{u_{k}}}\frac{%
\partial ^{2}z}{\partial u_{i}\partial u_{k}}+\frac{1}{y_{u_{i}}}\frac{%
\partial ^{2}y}{\partial u_{i}\partial u_{k}}\right]  \label{Obz120}
\end{equation}%
\begin{equation}
=\frac{1}{y_{u_{i}}z_{u_{k}}}\left[ (\ln z_{u_{k}})_{,u_{i}}+(\ln
y_{u_{i}})_{,u_{k}}\right] \text{ \ \ \ .}  \label{Obz121}
\end{equation}%
Evidently, the operator $M_{ik}(u_{i},u_{k},y\Longleftrightarrow z)$ is
different from zero! Now the commutator $\left[ \frac{\partial L}{\partial y}%
,\frac{\partial L}{\partial z}\right] $ for $\widetilde{P}_{0}$ in (\ref%
{Obz110A2}) between the two partial derivatives can be calculated as 
\begin{equation*}
\widetilde{P}_{0}\equiv \frac{\partial ^{2}L}{\partial y\partial z}-\frac{%
\partial ^{2}L}{\partial z\partial y}=\left[ \frac{\partial }{\partial y},%
\frac{\partial }{\partial z}\right] L
\end{equation*}%
\begin{equation*}
=\frac{1}{9}\sum\limits_{i,k=1}^{3}\left[ M_{ik}(u_{i},u_{k},y%
\Longleftrightarrow z)+M_{ik}(v_{i},u_{k},y\Longleftrightarrow z)\right] 
\frac{\partial L}{\partial u_{k}}
\end{equation*}%
\begin{equation}
+\frac{1}{9}\sum\limits_{i,k=1}^{3}\left[ M_{ik}(u_{i},v_{k},y%
\Longleftrightarrow z)+M_{ik}(v_{i},v_{k},y\Longleftrightarrow z)\right] 
\frac{\partial L}{\partial v_{k}}\text{ \ \ \ \ .}  \label{Obz122}
\end{equation}%
In the same way, the other commutators $\left[ \frac{\partial L}{\partial z},%
\frac{\partial L}{\partial x}\right] $ and $\left[ \frac{\partial L}{%
\partial x},\frac{\partial L}{\partial y}\right] $ can be calculated. Note
that when we have the case of the defining equalities (\ref{Obz113}) and (%
\ref{Obz114}) for $u_{i}$ and $v_{j}$, $i,j=1,2,......6$, the summation in (%
\ref{Obz122}) will be from $1$ to $6$ instead from $1$ to $3$. In any case,
all the four constituents in the two sums will be different one from
another, so no cancellation of terms with opposite signs will be possible.
\end{proof}

\subsubsection{A partial case of the commutativity of the partial derivatives%
}

Suppose now that the inverse function theorem can be applied. This means
that if for example the derivative $\frac{\partial y_{i}}{\partial u_{k}}$
or \bigskip $\frac{\partial y_{i}}{\partial v_{k}}$ can be calculated, then
the derivatives $\frac{\partial u_{k}}{\partial y_{i}}$ and $\frac{\partial
v_{k}}{\partial y_{i}}$ can also be calculated. Now it is important to
remind briefly the formulation of the inverse function theorem, so that it
is evident that in the concrete case this theorem can be applied because the
dimension of the initial manifold $M_{1}$ is $6$ (equal to the number of the
Kepler parameters in the right-hand side of the equations (\ref{K1}) - (\ref%
{K3})) and the dimension of the second manifold $M_{2}$ is equal to $3$ -
the number of the equations.

\begin{theorem}
\cite{MishtFomCourse} Let $f:R^{n}\rightarrow R^{k}$ is a mapping of
smoothness class $C^{\infty }$ of the system of equations and $M_{%
\overrightarrow{C}}$ is the set of solutions of the system of equations $f(%
\overrightarrow{x})=\overrightarrow{C}$, $(x^{1},x^{2},....x^{n})$, $%
\overrightarrow{C}=(C^{1},C^{2}......C^{k})$ and the mapping $f:$ $%
R^{n}\rightarrow R^{k}$ is determined by the functions $%
(f^{1},f^{2},.......f^{k})$, where 
\begin{equation}
f^{1}(x^{1},x^{2},....x^{n})=C^{1}\text{,.............,}%
f^{k}(x^{1},x^{2},....x^{n})=C^{k}\text{ \ \ .}  \label{Obz123}
\end{equation}%
If the rank of the Jacobi matrix for the mapping $f$ is maximal at each
point $P_{0}\in $ $M_{\overrightarrow{C}}$ (i.e.$rk$ $df(P_{0})=k$), then $%
M_{\overrightarrow{C}}$ is a $(n-k)$ dimensional smooth manifold of class $%
C^{\infty }$. Then in the neighborhood of each point $P_{0}\in $ $M_{%
\overrightarrow{C}}$ one can choose as local coordinates some $(n-k)$
cartesian coordinates of the embedding Euclidean space.
\end{theorem}

In other words, in the neighborhood of the non-degenerate points such local
coordinates can be introduced and the "inverse derivatives" such as $\frac{%
\partial u_{k}}{\partial y_{i}}$ and $\frac{\partial v_{k}}{\partial y_{i}}$
can be found \ 

From the definition of the function $M_{ik}(u_{i},u_{k},y\Longleftrightarrow
z)$ in (\ref{Obz118}) and (\ref{Obz119}) it will follow that

\begin{equation}
M_{ik}(u_{i},v_{k},y\Longleftrightarrow z)=\frac{\partial u_{i}}{\partial y}.%
\frac{\partial }{\partial u_{i}}\left( \frac{\partial u_{k}}{\partial z}%
\right) -\frac{\partial u_{i}}{\partial z}.\frac{\partial }{\partial u_{i}}%
\left( \frac{\partial u_{k}}{\partial y}\right) =  \label{Obz124}
\end{equation}%
\begin{equation}
=\frac{\partial }{\partial y}\left( \frac{\partial u_{k}}{\partial z}\right)
-\frac{\partial }{\partial z}\left( \frac{\partial u_{k}}{\partial y}\right)
=0\text{ \ \ \ .}  \label{Obz125}
\end{equation}%
However, the conditions of the inverse function theorem, when the "inverse
partial derivatives" can be found \ are not always fulfilled, and these are
the critical (singular) points.

Definitions and proofs of the inverse function theorem can be found also in 
\cite{DubrFomVol1}, \cite{Sternberg}, \cite{Boothby}, \cite{Narasimhan}, a
particular clear formulation is presented in \cite{TrofimFom} and other
books.

Consequently, for the cases when the inverse function theorem cannot be
applied (in particular for critical points), such a non-commutativity may
appear. In order to avoid such situations, some authors assert \cite{Boothby}
that "the Stokes theorem holds only for smooth manifolds with a smooth
boundary". But at the same time, it has been noted also that \ \cite{Boothby}
"the search for reasonable domains of integration to validate the Stokes
theorem leads to the concept of symplexial or polyhedral complex, i.e. a
space made up of fastening together along their faces a number of symplexes 
(line segments, triangles, tetrahedra) and their generalizations (polyhedra,
cubes). Then every $C^{\infty }$ manifold can be "triangulated", which means
that even with considerable smoothness it can be homeomorphic to such a
complex and thus, the integral becomes a sum of integrals over the pieces,
which are images of symplexes, cubes and etc."

The reason for introducing this citation is that we have already found three
boundaries of the manifold - the first one is the $(x,y)$ boundary,
represented by the plane equation (\ref{Obz66}) or (\ref{Obz67}), then the $%
(x,z)$ plane, represented by the plane equation (\ref{Obz74}) and the third $%
(y,z)$ boundary, which is represented by the plane (\ref{Obz81A}). So in a
further research it will be interesting to see whether the three planes will
intersect and thus the intersecting points will determine the boundaries of
integration in the Stokes theorem.

\bigskip In fact, the "commutative case" is related to the s.c. "integrable"
or "involutive" case, which is the essence of the Frobenius theorem.

\subsubsection{ Different formulations of the Frobenius theorem}

The Frobenius theorem has several definitions, which differ in their
formulation.

Let us give first the most concise one 

\begin{definition}
\cite{ArnoldGeom}. Let $M$ is a smooth manifold with a set of tangent
hyperplanes. In the neighborhood of a point with local coordinates $%
(x_{1},x_{2},......x_{n})$ the set of hyperplanes is given by the
differential $1-$form $\alpha =du-pdx$ (where $u=f(x)$), different from zero
and where $pdx=p_{1}dx_{1}+....+p_{n}dx_{n}$, $p_{i}=\frac{\partial f}{%
\partial x^{i}}(x)$.The function $f(x)$ is a $1-$jet function or is called
also a standard contact $1-$form on the manifold $J^{1}(V,%
%TCIMACRO{\U{211d} }%
%BeginExpansion
\mathbb{R}
%EndExpansion
)$, determined by the set of $(2n+1)$ numbers $%
(x_{1},x_{2},......x_{n};u;p_{1},p_{2},.....p_{n})$. Then the Frobenius
theorem states the condition when a set of hyperplanes in $%
%TCIMACRO{\U{211d} }%
%BeginExpansion
\mathbb{R}
%EndExpansion
^{3}$ will have integral hypersurfaces. In other words, this set of
hypersurfaces is fully integrable when the form $d\alpha $ in the
corresponding hyperplane is zero. 
\end{definition}

\begin{definition}
\cite{ArnoldGeom} The hyperplane $\alpha =0$ is fully integrable if and only
if the equality $\alpha \Lambda d\alpha =0$ is fulfilled.
\end{definition}

However, the notion of integrability in \cite{Narasimhan} is defined in a
slightly different and more general context.

\begin{definition}
\cite{Narasimhan} The differential system $D$ of class $C^{r}$ $(r\geq 2)$,
defined on a manifold $V$ is completely integrable if and only if it is
involutive, which means that at each point $a\in V$ there exists a
neighborhood $U$ and $C^{r}$ vector fields $X_{1},X_{2}.....X_{p}$
(determining $D$ on $U$) such that $\left[ X_{\mu },X_{\nu }\right] (b)\in
D(b)$.\ 
\end{definition}

In case the differential form system $D$ is specified to be 
\begin{equation}
D(a)=\left\{ X\in T_{a}V\text{ \ }:\omega _{\nu }(a)(X)=0\text{ \ , }p<\nu
\leq n\right\} \text{ \ \ }  \label{Obz126}
\end{equation}%
and $\omega _{p+1},.....,\omega _{n}$ are linearly independent at each point 
$1-$forms on $V$. Then again in \cite{Narasimhan} the following theorem has
been proved:

\begin{theorem}
If in suitable coordinates the one-forms $\omega _{\nu }$ are expressed as 
\begin{equation}
\omega _{\nu }=dx_{\nu }+\sum\limits_{\mu \leq p}a_{\nu \mu }dx_{\mu }\text{
\ \ , \ }\nu >p\text{ \ , }a_{\nu \mu }\in C(u)\text{ \ \ \ ,}
\label{Obz127}
\end{equation}%
then $D$ is a differential form system, formed by the vector fields 
\begin{equation}
X_{\mu }=\frac{\partial }{\partial x_{\mu }}-\sum\limits_{\nu >p}a_{\nu \mu
}\frac{\partial }{\partial x_{\nu }}\text{ \ \ ,}  \label{Obz128}
\end{equation}%
so that $\omega _{\nu }(X_{\mu })=0$ .
\end{theorem}

If $X_{\mu }$ and $\frac{\partial }{\partial x_{\mu }}$ are the
corresponding vectors in the two tangent space, we can write also the
relations 
\begin{equation}
X_{\mu }=\sum\limits_{\nu =1}^{p}b_{\mu \nu }Y_{\nu }\text{ \ , \ \ \ \ }%
Y_{\nu }(x)=\sum\limits_{\alpha =1}^{n}a_{\nu \alpha }(x)\left( \frac{%
\partial }{\partial x_{\alpha }}\right) _{x}\text{\ }  \label{Obz129}
\end{equation}%
Then it can easily be calculated that 
\begin{equation}
\left[ X_{\mu },X_{\nu }\right] =\sum\limits_{m=1}^{p}\lambda
_{m}X_{m}=\sum\limits_{m=1}^{n}\xi _{m\beta }\frac{\partial }{\partial
x_{\beta }}\text{ \ \ .}  \label{Obz130}
\end{equation}%
Concretely for our case, instead of the local coordinates $x_{\mu }$ we have 
$u_{\mu }$ and $v_{\nu }$, the analog of (\ref{Obz130}) is (\ref{Obz122}),
instead of $X_{\mu }$ and $X_{\nu }$ we have $\frac{\partial }{\partial
x_{\mu }}$ and $\frac{\partial }{\partial x_{\nu }}$. Let us introduce also
the notations  
\begin{equation}
\xi _{ik}^{(1)}\equiv M_{ik}(u_{i},u_{k},y\Longleftrightarrow
z)+M_{ik}(v_{i},u_{k},y\Longleftrightarrow z)\text{ \ \ \ ,}
\label{Obz130A1}
\end{equation}%
\begin{equation}
\xi _{ik}^{(2)}\equiv M_{ik}(u_{i},v_{k},y\Longleftrightarrow
z)+M_{ik}(v_{i},v_{k},y\Longleftrightarrow z)\text{ \ \ \ .}
\label{Obz130A2}
\end{equation}
The commutator $\left[ \frac{\partial }{\partial x_{\mu }},\frac{\partial }{%
\partial x_{\nu }}\right] $ acts on the Lagrangian $L$ and since 
\begin{equation}
\left[ \frac{\partial }{\partial x_{\mu }},\frac{\partial }{\partial x_{\nu }%
}\right] L=\frac{1}{9}\sum\limits_{i,k=1}^{3}\left[ \xi _{ik}^{(1)}\frac{%
\partial }{\partial u_{k}}+\xi _{ik}^{(2)}\frac{\partial }{\partial v_{k}}%
\right] L\neq 0\text{ \ \ \ ,}  \label{Obz131}
\end{equation}%
the system is not involutive.

The involutive meaning of the Frobenius theorem in \cite{Boothby} is
formulated with the purpose of finding the condition for the integrability
of the system of nonlinear equations 
\begin{equation}
\frac{\partial z}{\partial x}=g(x,y,z)\text{ \ \ , \ }\frac{\partial z}{%
\partial y}=h(x,y,z)\text{ \ \ \ .}  \label{Obz132}
\end{equation}%
Then the solution will exist if the following pair of vector fields $X$ and $%
Y$ can be constructed so that 
\begin{equation}
X=\frac{\partial }{\partial x}+f_{x}(x,y)\frac{\partial }{\partial z}\text{
\ \ , \ \ }Y=\frac{\partial }{\partial y}+f_{y}(x,y)\frac{\partial }{%
\partial z}\text{ \ , }  \label{Obz133}
\end{equation}%
where $z=f(x,y)$. Then the necessary condition for the existence of an
independent pair of vector fields $X$ and $Y$, defining a tangent plane to
the surface \ $z=f(x,y)$, can be expressed as 
\begin{equation}
\left[ X,Y\right] =\left( \frac{\partial ^{2}f}{\partial x\partial y}-\frac{%
\partial ^{2}f}{\partial x\partial y}\right) =0\text{ \ \ . }  \label{Obz134}
\end{equation}%
However, it should be kept in mind that this formulation can be applied if
it is possible to express the $z-$coordinate as a function $%
z=f(x,y)$ of $x$ and $y$ . In the case of the transformations (\ref{K1}) - (%
\ref{K3}) this is impossible to the following reason: from the already used
relation $x^{2}+y^{2}+z^{2}=r^{2}=\left( \frac{a\left( 1-e^{2}\right) }{%
1+e\cos f}\right) ^{2}$(see also (\ref{Obz13}) and (\ref{Obz14})) $\cos f$
can be expressed and substituted in the equations (\ref{K1})- (\ref{K3}) for 
$x,y$ and $z$. Then a much more complicated algebraic equation will be
obtained, which even does not guarantee that it will be possible to find any
dependence \ $z=f(x,y)$.

\section{Second stage of the variational formalism, based on the
Gauss-Ostrogradsky theorem}

\subsection{Formulation of the Gauss-Ostrogradsky theorem}

Previously we started from the Stokes theorem (\ref{Obz102}) $%
\int\limits_{\partial M}\left[ Adx+Bdy+Cdz\right] $%
\begin{equation*}
=\int\limits_{M}\left[ Pdy\Lambda dz+Qdz\Lambda dx+Rdx\Lambda dy\right] 
\text{ \ \ ,}
\end{equation*}
where in the right-hand side the integration is  over the $2-$form $\omega
_{2}$ 
\begin{equation}
\omega _{2}=Pdy\Lambda dz+Qdz\Lambda dx+Rdx\Lambda dy\text{ \ \ \ \ ,}
\label{Obz135}
\end{equation}%
defined on the manifold $M$. The essence of the Gauss-Ostrogradsky theorem
is that the integration in the right-hand side of (\ref{Obz102}) may be
considered as an integration on some boundary $\partial M$ of the manifold,
which after applying a new differentiation $d\omega _{2}$ will become an 
integral over a $3-$form.

The Gauss-Ostrogradsky theorem is based on the general equation 
\begin{equation}
\int\limits_{\partial M}\omega _{2}=-\int\limits_{\partial M}i^{\ast }\omega
_{2}=\int\limits_{M}d\omega _{2}  \label{Obz136}
\end{equation}%
\begin{equation}
=\int\limits_{M}\left[ \frac{\partial P}{\partial x}+\frac{\partial Q}{%
\partial y}+\frac{\partial R}{\partial z}\right] dx\Lambda dy\Lambda dz\text{
\ \ \ \ \ .}  \label{Obz137}
\end{equation}%
We shall not give the detailed calculation how the three-form is derived -
this can be found in almost all textbooks on differential geometry such as 
\cite{SchwarzPhys}, \cite{MishtFomCourse}, \cite{Boothby}, \cite{Sternberg}, 
\cite{DubrFomVol1} and others. Let us remind the functions $P,Q$ and $R$,
calculated in (\ref{Obz110}) and (\ref{Obz110A1})%
\begin{equation}
P\equiv \widetilde{P}_{0}+P_{1}+T^{P}+\widetilde{G}_{1}\text{ \ ,}.Q\equiv 
\widetilde{Q}_{0}+Q_{1}+T^{Q}+\widetilde{G}_{2}\text{ \ , \ }R\equiv 
\widetilde{R}_{0}+R_{1}+T^{R}+\widetilde{G}_{3}\text{ \ \ \ .}
\label{Obz138}
\end{equation}
The function $T^{P}$ previously was defined as \bigskip $\overset{.}{T^{P}}%
\equiv \lambda _{3}T_{y}-\lambda _{2}T_{z}$, analogously $T^{Q},T^{R}$ are
defined.

\subsection{Calculation of the functions in the right-hand side of the
Gauss-Ostrogradsky equation}

\begin{proposition}
The following two equalities are fulfilled 
\begin{equation}
\frac{\partial P_{1}}{\partial x}+\frac{\partial Q_{1}}{\partial y}+\frac{%
\partial R_{1}}{\partial z}=0\text{ \ \ ,}  \label{Obz139}
\end{equation}%
\begin{equation}
\frac{\partial T^{P}}{\partial x}+\frac{\partial T^{Q}}{\partial y}+\frac{%
\partial T^{R}}{\partial z}=0\text{ \ }  \label{Obz140}
\end{equation}%
but only provided the partial derivatives $\frac{\partial }{\partial x},%
\frac{\partial }{\partial y},\frac{\partial }{\partial z}$ commute. These
two equalities will not be equal to zero, if expressed in terms of the
Kepler parameters. In case the partial derivatives do not commute,
expression (\ref{Obz122}) for $\widetilde{P}_{0}$ should be used
(analogously, $\widetilde{Q}_{0}$ and $\widetilde{R}_{0}$ are constructed)
for calculation of the expressions for $P,Q,R$ in (\ref{Obz110}) and (\ref%
{Obz110A1}).
\end{proposition}

\ The expression for the function $\ P$ (\ref{Obz138}) in the Stokes and
Gauss-Ostrogradsky theorem can be calculated by means of the partial
derivatives of the cartesian coordinates $(x,y,z)$ with respect to the two
sets $u_{i}$ and $v_{j}$ of Kepler parameters 
\begin{equation*}
u_{i}=(f_{.},a_{.},e_{.})\text{ \ \ }i=1,2,3\text{ and }v_{j}=(\Omega
_{.},i_{.},\omega _{.})\text{ \ }j=1,2,3
\end{equation*}%
to be 
\begin{equation*}
P=\frac{\partial C}{\partial y}-\frac{\partial B}{\partial z}=\frac{1}{3}%
\sum\limits_{i=1}^{3}(\frac{\lambda _{3}}{y_{u_{i}}}-\frac{\lambda _{2}}{%
z_{u_{i}}})\frac{\partial \widetilde{T}}{\partial u_{i}}+\frac{1}{3}%
\sum\limits_{j=1}^{3}(\frac{\lambda _{3}}{y_{v_{j}}}-\frac{\lambda _{2}}{%
z_{v_{i}}})\frac{\partial \widetilde{T}}{\partial v_{j}}
\end{equation*}%
\begin{equation}
+\frac{1}{3}(\frac{\lambda _{2}}{z_{f}}-\frac{\lambda _{3}}{y_{f}})\frac{%
\partial G}{\partial f}+\left[ K_{1}(y,z)-K_{1}(z,y)\right] \text{ \ \ \ \ .}
\label{Obz141}
\end{equation}%
In case the calculation for $P$ is performed for the two sets of Kepler
parameters $(f_{.1},a_{1.},e_{1},\Omega _{1.},i_{1.},\omega _{1})$ and $%
(f_{2.},a_{2.},e_{2},\Omega _{2.},i_{2.},\omega _{2})$, the summation in the
two sums will be from $1$ to $6$ (respectively there will be $6$ derivatives 
$\frac{\partial }{\partial u_{i}}$ and $6$ derivatives $\frac{\partial }{%
\partial v_{j}}$) and the coefficients in front of the sums will be $\frac{1%
}{6}$.

\begin{proposition}
The antisymmetric function $K_{1}(y,z)$ with respect to the variables $y$
and $z$ can be represented as a sum of \ ten variational derivatives of the
first, second and the third order 
\begin{equation*}
K_{1}(y,z)\equiv P_{ij}^{(1)}\frac{\partial ^{2}L}{\partial u_{i}\partial s}%
+P_{ij}^{(2)}\frac{\partial ^{2}L}{\partial v_{i}\partial s}+P_{ij}^{(3)}%
\frac{\partial ^{3}L}{\partial u_{i}\partial u_{j}\partial s}
\end{equation*}%
\begin{equation*}
+P_{ij}^{(4)}\frac{\partial ^{3}L}{\partial u_{i}\partial v_{j}\partial s}%
+P_{ij}^{(5)}\frac{\partial ^{3}L}{\partial v_{i}\partial v_{j}\partial s}%
+P_{ij}^{(6)}\frac{\partial ^{2}L}{\partial u_{i}\partial u_{j}}
\end{equation*}
\end{proposition}

\begin{equation}
+P_{ij}^{(7)}\frac{\partial ^{2}L}{\partial u_{i}\partial v_{j}}+P_{ij}^{(8)}%
\frac{\partial ^{2}L}{\partial v_{i}\partial v_{j}}+P_{ij}^{(9)}\frac{%
\partial L}{\partial u_{i}}+P_{ij}^{(10)}\frac{\partial L}{\partial v_{i}}%
\text{ \ \ \ ,}  \label{Obz142}
\end{equation}%
where the coefficient functions $%
P_{ij}^{(1)},P_{ij}^{(2)}.......P_{ij}^{(10)}$ are given in Appendix C.

\begin{proposition}
\ The exterior product \ $dy\Lambda dz$ in terms of the variables, related
to the Kepler parameters can be expressed as 
\begin{equation*}
dy\Lambda dz=\sum\limits_{\underset{i\pm j}{i,j=1}}^{3}\left[ \frac{D(y,z)}{%
D(u_{i},u_{j})}du_{i}du_{j}+\frac{D(y,z)}{D(v_{i},v_{j})}dv_{i}dv_{j}\right]
\end{equation*}%
\begin{equation}
+\sum\limits_{\underset{i\pm j,also\text{ }i=j}{i,j=1}}^{3}\frac{D(y,z)}{%
D(u_{i},v_{j})}du_{i}dv_{j}\text{ \ \ ,}  \label{Obz143}
\end{equation}%
where $\frac{D(y,z)}{D(u_{i},u_{j})}$ \ is the Jacobian determinant 
\begin{equation}
\frac{D(y,z)}{D(u_{i},u_{j})}\equiv 
\begin{vmatrix}
\frac{\partial y}{\partial u_{i}} & \frac{\partial y}{\partial u_{j}} \\ 
\frac{\partial z}{\partial u_{i}} & \frac{\partial z}{\partial u_{j}}%
\end{vmatrix}%
=y_{u_{i}}z_{u_{j}}-y_{u_{j}}z_{u_{i}}\text{ \ \ \ \ .}  \label{Obz144}
\end{equation}%
Taking into account the similar expressions for $\frac{D(y,z)}{D(v_{i},v_{j})%
}$ and $\frac{D(y,z)}{D(u_{i},v_{j})}$ in (\ref{Obz143}), it is easily
proved that 
\begin{equation}
dy\Lambda dz=N_{1}(y,z)-N_{1}(z,y)\text{ \ \ \ \ ,}  \label{Obz145}
\end{equation}%
\begin{equation}
N_{1}(y,z)\equiv \sum\limits_{i,j=1}^{3}\left[
y_{u_{i}}z_{u_{j}}du_{i}du_{j}+y_{v_{i}}z_{v_{j}}dv_{i}dv_{j}+y_{u_{i}}z_{v_{j}}du_{i}dv_{j}%
\right]  \label{Obz146}
\end{equation}%
and $N_{1}(y,z)$ is an antisymmetric function with respect to the variables $%
y$ and $z$.. Respectively, 
\begin{equation}
\frac{D(y,z)}{D(v_{i},v_{j})}\text{, \ }\frac{D(y,z)}{D(u_{i},v_{j})}\text{
\ \ \ , \ \ \ }N_{1}(z,x)\text{\ , \ }N_{1}(x,y)\text{ }  \label{Obz147}
\end{equation}%
are determined similarly to formulae (\ref{Obz144}).
\end{proposition}

It is important to note that with account of (\ref{Obz146}), expression (\ref%
{Obz145}) can be written as 
\begin{equation*}
dy\Lambda dz=\sum\limits_{i,j=1}^{3}\left[
(y_{u_{i}}z_{u_{j}}-z_{u_{i}}y_{u_{j}})du_{i}du_{j}+(y_{v_{i}}z_{v_{j}}-z_{v_{i}}y_{v_{j}})dv_{i}dv_{j}%
\right]
\end{equation*}%
\begin{equation}
+(y_{u_{i}}z_{v_{j}}-z_{u_{i}}y_{v_{j}})du_{i}dv_{j}]\text{ \ \ \ \ .}
\label{Obz148}
\end{equation}%
Each of the first and the second terms in the square brackets consists of an
antisymmetric term in the round brackets, multiplied by a symmetric term ($%
du_{i}du_{j}$ or $dv_{i}dv_{j}$). It is known that such a multiplication of
an antisymmetric function with a symmetric function will give zero.
Consequently, only the last (third) term will remain  
\begin{equation}
dy\Lambda
dz=\sum\limits_{i,j=1}^{3}(y_{u_{i}}z_{v_{j}}-z_{u_{i}}y_{v_{j}})du_{i}dv_{j}%
\text{ \ \ .}  \label{Obz149}
\end{equation}%
This is a second order differential form, based on the different
differentials $du_{i}$ and $dv_{j}$. For the case of two satellites, the
summation will be from $1$ to $6$. Similarly one can write the expressions
for $dz\Lambda dx$ and $dx\Lambda dy$. Thus, there will be $6\times 6=36$ $\ 
$ terms in the right-hand side of the two-form $\omega _{2}=Pdy\Lambda
dz+Qdz\Lambda dx+Rdx\Lambda dy$ (\ref{Obz102}), where formulae (\ref{Obz141}%
) for $P$ will be used ($Q$ and $R$ are obtained after cyclic permutation $%
x\rightarrow y\rightarrow z\rightarrow x$) and also (\ref{Obz142}) for the
antisymmetric function $K_{1}(y,z)$ (respectively after cyclic permutation,
also for $K_{1}(z,x)$ and $K_{1}(x,y)$). The $36$ terms in the sum $%
\sum\limits_{i,j=1}^{6}$ will be composed of all the combinations of
multiplication of differentials from the first set of Kepler parameters $%
(df_{.1},da_{1.},de_{1},d\Omega _{1.},di_{1.},d\omega _{1})$ with the
differentials $(df_{2.},da_{2.},de_{2},d\Omega _{2.},di_{2.},d\omega _{2})$
from the second set. The detailed calculations will be performed in another
paper.

\section{Conclusion}

In this paper a new model is presented for transmitting signals between
satellites on two different space-distributed orbits, characterized by the
two sets of Kepler parameters $(f_{.1},a_{1.},e_{1},\Omega
_{1.},i_{1.},\omega _{1})$ and $(f_{2.},a_{2.},e_{2},\Omega
_{2.},i_{2.},\omega _{2})$, with account also of the General Relativity
effect of curving the signal trajectory and the Shapiro delay formulae.
First, it has been proved on the base of the Whitney theorem from
differential geometry and the topological notions of immersions, embeddings
and submersions that it is possible to construct a manifold, which will
depend on all the $12$ Kepler parameters plus the parameter $s$, which
parametrizes cartesian coordinates $x,y,z$ in the transformations (\ref{K1})
- (\ref{K3}). In fact, two sets of transformations (\ref{K1}) - (\ref{K3})
should be written for the variables $x_{1},y_{1},z_{1}$ and $%
x_{2},y_{2},z_{2}$ , but since it turned out to be possible to introduce the
variables $x=x_{2}-x_{1}$, $y=y_{2}-y_{1}$, $z=z_{2}-z_{1}$, the variational
formalism was performed only by variation of the $x,y,z$ variables.
Secondly, it was proved that the first term in the Shapiro delay formulae $%
\frac{R_{AB}}{c}$ ($R_{AB}-$the Euclidean distance between the satellites)
should be replaced by the term $\frac{1}{c}\int\limits_{path}\sqrt{\left( 
\overset{.}{x}\right) ^{2}+\left( \overset{.}{y}\right) ^{2}+\left( \overset{%
.}{z}\right) ^{2}}ds$, which shows that the distance travelled by the signal
is greater than the distance between points on the straight line. Here an
important clarification should be made. In this paper a variational
formalism was performed, taking into account only the term $\frac{1}{c}%
\int\limits_{path}\sqrt{\left( \overset{.}{x}\right) ^{2}+\left( \overset{.}%
{y}\right) ^{2}+\left( \overset{.}{z}\right) ^{2}}ds$, and the result was
that the straight line is the optimal trajectory for propagation of the
signal. However, this does not mean that the signal propagates along a
straight line. In previous papers \cite{Bog2}, \cite{Bog3} and \cite{Bog4}
it has been proved that geodesic distance, travelled by the signal is
greater than the Euclidean distance. Here the straight line appears in the
consideration because of the condition for "optimality" of the distance and
the positions of the two satellites and is fully consistent with theorems
from differential geometry. So the result by itself does not contradict to
any physical laws about propagation of signals, it only suggests that in
order to find any optimal positions or trajectories, one should take into
account the second term $\frac{2G_{\oplus }M_{\oplus }}{c^{3}}\frac{\sqrt{%
\left( \overset{.}{x}\right) ^{2}+\left( \overset{.}{y}\right) ^{2}+\left( 
\overset{.}{z}\right) ^{2}}}{\left[ x^{2}+y^{2}+z^{2}\right] }\left[ x%
\overset{.}{x}+y\overset{.}{y}+z\overset{.}{z}\right] $ (see (\ref{Obz20}))
in the proposed modified Shapiro delay formulae.  Particular attention in the
paper is paid also to the approximations in deriving the standard Shapiro
formulae, since the differential in the null cone equation can be determined
in several ways. As a result of a not so restrictive condition, the modified
Shapiro delay formulae was obtained in the form (\ref{Obz19}).

Taking into account of this second term (\ref{Obz20}) has not been performed
in this paper (in the sense of solving the corresponding Euler-Lagrange
equations with account of the second term), but the variational formalism is
proposed in its most general form.

The main goal of \ this paper was to propose a variational formalism, based
on higher order differential forms and the Stokes and Gauss Ostrogradsky
formalism. Since the Stokes theorem relates an integral on a manifold $M$ to
an integral on its boundary $\partial M$, it was necessary to clarify the
notion of a boundary, which in its classical formulation is related to the
inverse function theorem \cite{MishtFomCourse} and sub-manifolds, convexity
and semi-spaces. The last two notions are frequently used also in
differential geometry of surfaces \cite{Thorpe} and optimization theory \cite%
{Suharev}. The definition for a boundary turned out to be useful for finding
the boundary of the manifold, given by the transformations (\ref{K1}) - (\ref%
{K3}). For this case, there are three boundaries - the $(x,y)-$boundary,
given by the plane (\ref{Obz66}) $x\sin (\omega +f+\Omega )-y\cos (\omega
+f+\Omega )=0$ or the plane (\ref{Obz67}) $x\sin (\omega +f-\Omega )+y\cos
(\omega +f-\Omega )=0$ (and the condition for a boundary $z=0$); the $(x,z)$%
-boundary (\ref{Obz74}) $x$ .$\sin i\frac{\left[ \sin ^{2}(\omega +f)\left(
1+\cos ^{2}i\right) -1\right] }{\sin ^{2}(\omega +f)\cos ^{2}i}+z$ $\frac{%
\cos \Omega \left[ 1-\cos ^{4}i\text{ }tg^{4}(\omega +f)\right] }{\cos ^{2}i%
\text{ }tg^{3}(\omega +f)}=0$ (and the condition $y=0$) and the $(y,z)$
boundary (\ref{Obz81A}) $Ny-z\cos (\omega +f)=0$ with $N(\omega ,f,\Omega
)\equiv \sqrt{-\cos (\omega +f+\Omega )\cos (\omega +f-\Omega )}$ (\ref%
{Obz81}) (and the condition $x=0$). These boundaries naturally are $2-$%
dimensional hyperplanes, because the initial manifold is a $3-$dimensional
one, and consequently the boundaries are $1$ dimension less. The boundaries
were found only for the case of one $3D$ manifold $M_{1}$. The boundary of
the manifold $M=M_{1}\times M_{2}$ can be found if each of these equations
is written for the corresponding manifolds $M_{1}$ and $M_{2}$ with indices "%
$1$" and "$2$" respectively, and the system, determining the $(x,y,z)$%
-coordinates of the $M-$manifold \ $x=x_{2}-x_{1}$, $y=y_{2}-y_{1}$ , $%
z=z_{2}-z_{1}$ is added to them. Thus there will be a system of $9-$%
equations for the $9-$variables $(x_{1},y_{1},z_{1},x_{2},y_{2},z_{2},x,y,z)$
and the $x,y,z$ coordinates will be expressed through the two sets of Kepler
parameters.

Another peculiarity in the variational formalism, based on second- and
third-order differential forms and the Stokes theorem is the
non-commutativity of the partial derivatives $\frac{\partial }{\partial x},%
\frac{\partial }{\partial y},\frac{\partial }{\partial x}$, which means that
\ the commutator $\ $(\ref{Obz131}) $\left[ \frac{\partial }{\partial x_{\mu
}},\frac{\partial }{\partial x_{\nu }}\right] L=\frac{1}{9}%
\sum\limits_{i,k=1}^{3}\left[ \xi _{ik}^{(1)}\frac{\partial }{\partial u_{k}}%
+\xi _{ik}^{(2)}\frac{\partial }{\partial v_{k}}\right] L\neq 0$ is
nonzero.This is natural from the viewpoint of differential geometry since a
zero commutator means that the condition for integrability and the Frobenius
theorem are satisfied \cite{Boothby}. In the case, however, the additional
terms in the variational action due to the non-commutativity have to be
accounted.

Further, the function $P(x,y,z)=\frac{\partial C}{\partial y}-\frac{\partial
B}{\partial z}$ has been calculated in (\ref{Obz141}), which is in the
right-hand side of the Stokes theorem (\ref{Obz102}) $\int\limits_{\partial
M}\left[ Adx+Bdy+Cdz\right] =\int\limits_{M}\left[ Pdy\Lambda dz+Qdz\Lambda
dx+Rdx\Lambda dy\right] $ and also in the two-form (\ref{Obz135}) $\omega
_{2}=Pdy\Lambda dz+Qdz\Lambda dx+Rdx\Lambda dy$. The function $P$ is
expressed in terms of the functions $(x,y,z)$ and their derivatives with
respect to the Kepler parameters. Again it is useful to remind that when the
summation is changed \ from $1$ to $6$ in both sums, which include the
derivatives with respect to $u_{i}$ and $v_{j}$, the expression for $P$ will
contain derivatives with respect to both sets of Kepler parameters.
Particularly interesting and rather complicated is the antisymmetric
function $K_{1}(y,z)$ in (\ref{Obz142}), which contains derivatives (of
first, second and third order) of the original Lagrangian with respect to $%
u_{i}$,$v_{j}$ and the parameter $s$.

There is one more interesting problem, concerning the right-hand side in the
Stokes theorem (\ref{Obz102})(the integral $\int\limits_{M}$ $\left[
Pdy\Lambda dz+Qdz\Lambda dx+Rdx\Lambda dy\right] $). The functions $P$ in (%
\ref{Obz141}), $Q$ and $R$ contain the function $K_{1}(y,z)$ (\ref{Obz142})
(respectively $K_{1}(z,x)$ and $K_{1}(x,y)$), and these functions contain
first-order derivatives with respect to $u_{i}$ and $v_{j}$ and second -
order derivatives with respect to $u_{i},u_{j\text{ }}$or $v_{i},v_{j}$ or $%
u_{i},v_{j}$. In the derived expression (\ref{Obz149}) for $dy\Lambda dz$, a
sum over the different combinations $du_{k}dv_{l}$ was obtained, so each
derivative with respect to $u_{i}$ or second-order derivative with respect
to $u_{i}$ and $u_{j}$ is multiplied by $du_{k}$. In the variational
formalism, we have not imposed any restrictions which are the concrete
Kepler parameters corresponding to each $u_{i}$ or to each $v_{j}$. So the
problem is: in how many ways one can choose two variables $u_{i},u_{k}$ from
among the $6$ parameters for the first set. Respectively, for the case of
the second derivatives in $K_{1}(y,z)$ (respectively $K_{1}(z,x)$ and $%
K_{1}(x,y)$), the similar problem is: in how many ways three variables $u_{i}
$,$u_{j}$,$u_{k}$ can be chosen from among the total $6$-ones? Note that the
variable $u_{k}$ can be chosen independently from $u_{i}$ and $u_{j}$ (but
can be also equal to one of them). Consequently, after obtaining the final
formulaes in the variational formalism, one should in addition sum up over
all the possible choices (permutations) of $u_{i}$,$u_{j}$,$u_{k}$ (or $u_{i}
$,$u_{k}$, depending on the terms in $K_{1}(y,z)$ ). The same should be
performed with respect to the possible permutations of $v_{i}$,$v_{j}$,$v_{k}
$ (or $v_{i}$,$v_{k}$). This would require the implementation of group
theory and shall be developed in further publications.

\section{Appendix A: Expression for the square of the Euclidean distance
between two satellites on two different space-distributed orbits}

The differential $dR_{AB}^{2}$ of the square of the Euclidean distance,
where 
\begin{equation}
R_{AB}^{2}=(x_{1}-x_{2})^{2}+(y_{1}-y_{2})^{2}+(z_{1}-z_{2})^{2}  \label{A00}
\end{equation}%
and $(x_{1},y_{1},z_{1})$ and $(x_{2},y_{2},z_{2})$ are given by the
transformations (\ref{K1}), (\ref{K2}) and (\ref{K3}) with indices "$1$" and
"$2$" correspondingly: 
\begin{equation*}
dR_{AB}^{2}=\{\frac{2e_{1}\sin f_{1}a_{1}^{2}(1-e_{1}^{2})}{(1+e_{1}\cos
f_{1})^{3}}+\frac{2a_{1}a_{2}(1-e_{1}^{2})(1-e_{2}^{2})}{(1+e_{1}\cos
f_{1})^{2}(1+e_{2}\cos f_{2})}[K_{1}(1,2)
\end{equation*}%
\begin{equation}
+K_{2}(1,2)+K_{3}(1,2)+K_{4}(1,2)]\}df_{1}+\{1\longleftrightarrow 2\}df_{2}%
\text{ \ ,}  \label{A1}
\end{equation}%
where the expression $\{1\longleftrightarrow 2\}$ is the same as in the
first bracket $\{\}$ on the first line of (\ref{A1}), but with interchanged
indices $1$ and $2$, i.e.$1\longleftrightarrow 2$. Formulae (\ref{A1}) in
fact is the expression for the differential $dR_{AB}^{2}$ in terms of the
differentials $df_{1}$ and $df_{2}$ (provided that only $f_{1}$ and $f_{2}$
change, the remaining five Kepler parameters in the two sets $(\Omega
_{1},\omega _{1},i_{1},a_{1},e_{1})$ and $(\Omega _{2},\omega
_{2},i_{2},a_{2},e_{2})$ are constant, but are present in the formulae),
i.e. 
\begin{equation}
dR_{AB}^{2}=\frac{\partial \left( R_{AB}^{2}\right) }{\partial f_{1}}df_{1}+%
\frac{\partial \left( R_{AB}^{2}\right) }{\partial f_{2}}df_{2}=\widetilde{K}%
_{1}(\Gamma ^{(1)},\Gamma ^{(2)})df_{1}+\widetilde{K}_{2}(\Gamma
^{(2)},\Gamma ^{(1)})df_{2}\text{ \ \ \ ,}  \label{A0}
\end{equation}%
where again fo convenience the notations $\Gamma ^{(1)}=(f_{1},\Omega
_{1},\omega _{1},i_{1},a_{1},e_{1})$ and $\Gamma ^{(2)}=(f_{2},\Omega
_{2},\omega _{2},i_{2},a_{2},e_{2})$ have been introduced. The expressions
for $K_{1}(1,2),K_{2}(1,2),K_{3}(1,2),K_{4}(1,2)$ in (\ref{A1}) are given
below. The corresponding expressions for $%
K_{1}(2,1),K_{2}(2,1),K_{3}(2,1),K_{4}(2,1)$ in $\{1\longleftrightarrow
2\}df_{2}$ are the same, but with interchanged indices $1\longleftrightarrow
2$ \ \ 

The terms $K_{1}(1,2)$ do not contain the inclination angles $i_{1}$ and $%
i_{2}$ 
\begin{equation*}
K_{1}(1,2)\equiv -e_{1}\sin f_{1}\cos (\omega _{1}+f_{1})\cos (\omega
_{2}+f_{2})\cos (\omega _{1}+f_{1}+\Omega _{1}-\Omega _{2})
\end{equation*}%
\begin{equation*}
+e_{1}\cos f_{1}\sin (\omega _{1}+f_{1})\cos (\omega _{2}+f_{2})\cos (\Omega
_{1}-\Omega _{2})
\end{equation*}

\begin{equation}
+\sin (\omega _{1}+f_{1})\cos (\omega _{2}+f_{2})\cos (\Omega _{1}-\Omega
_{2})\text{ \ \ \ .}  \label{A2}
\end{equation}%
The terms $K_{2}(1,2)$ contain the inclination angles $i_{1}$ and $i_{2}$ in
the terms with multipliers $\cos i_{1}\cos i_{2}$ and $\sin i_{1}\sin i_{2}$%
\begin{equation*}
K_{2}(1,2)\equiv -e_{1}\sin (\omega _{2}+f_{2})\cos i_{1}\cos i_{2}[\sin
f_{1}\cos (\Omega _{1}-\Omega _{2})\sin (\omega _{1}+f_{1})
\end{equation*}%
\begin{equation}
+\cos (\omega _{1}+f_{1})\cos f_{1}\cos (\Omega _{1}-\Omega _{2})+\cos
(\omega _{1}+f_{1})\sin f_{1}\sin (\Omega _{1}-\Omega _{2})]\text{ \ \ .}
\label{A3}
\end{equation}%
Terms $K_{3}(1,2)$ contain $\cos i_{1}$ and do not contain $\cos i_{2}$ 
\begin{equation}
K_{3}(1,2)\equiv \lbrack e_{1}\cos f_{1}+1]\cos (\omega _{1}+f_{1})\cos
(\omega _{2}+f_{2})\sin (\Omega _{1}-\Omega _{2})\cos i_{1}\text{ \ \ .}
\label{A4}
\end{equation}%
The terms $K_{4}(1,2)$ contain the multiplier $\cos i_{2}$ and do not
contain $\cos i_{1}$ 
\begin{equation}
K_{4}(1,2)\equiv \lbrack -e_{1}\cos f_{1}+1]\sin (\omega _{1}+f_{1})\sin
(\omega _{2}+f_{2})\sin (\Omega _{1}-\Omega _{2})\cos i_{2}\text{ \ \ .}
\label{A5}
\end{equation}

\section{Appendix B: The propagation time for a signal, emitted by a
satellite on a space-distributed orbit, expressed in terms of higher-rank
elliptic integrals}

Formulae (\ref{K4}) $\widetilde{T}=G(f)=-\frac{2G_{\oplus }M_{\oplus }}{c^{3}%
}.\frac{n}{(1-e^{2})^{\frac{3}{2}}}\widetilde{T}_{1}+T_{2}$ for the
propagation time has been derived in \cite{Bog1}, but in the review paper 
\cite{Bog2} the different parts of the expression have been collected in one
formulae. The expression for $\widetilde{T}_{1}$ in \ (\ref{K4}) is

\begin{equation}
\widetilde{T}_{1}=-2i\frac{na}{c}q^{\frac{3}{2}}\int \frac{y^{2}dy}{\sqrt{%
(1-y^{2})(1-q^{2}y^{2})}}\text{ \ \ \ \ ,}  \label{B1}
\end{equation}%
where $y=\sqrt{\frac{1}{q}.\frac{(1+e\cos E)}{(1-e\cos E)}}$, \ $q=\frac{1-e%
}{1+e}$ and $\tan \frac{f}{2}=\sqrt{\frac{1-e\cos f}{1+\cos f}}=\sqrt{\frac{%
1+e}{1-e}}\tan \frac{E}{2}$. The second integral $T_{2}$ is \ 
\begin{equation}
T_{2}=-\frac{2G_{\oplus }M_{\oplus }}{c^{3}}.\frac{n}{(1-e^{2})^{\frac{3}{2}}%
}\widetilde{T}_{1}+T_{2}^{(2)}\text{ \ , }  \label{B2}
\end{equation}%
where $T_{2}^{(2)}$ is a combination of the elliptic integral $\widetilde{J}%
_{2}^{(4)}(\widetilde{y},q)$ of the second order and the elliptic integral
of the fourth order $\widetilde{J}_{4}^{(4)}(\widetilde{y},q)$ 
\begin{equation}
T_{2}^{(2)}=-inq^{\frac{5}{2}}\frac{2G_{\oplus }M_{\oplus }}{c^{3}}(1+e^{2})%
\widetilde{J}_{2}^{(4)}(\widetilde{y},q)+inq^{\frac{3}{2}}\frac{2G_{\oplus
}M_{\oplus }}{c^{3}}\frac{(1+e^{2})}{(1-e^{2})}\text{ }\widetilde{J}%
_{4}^{(4)}(\widetilde{y},q)\text{\ .}  \label{B3}
\end{equation}%
The lower indice denotes the order of the integral (determined by the degree
of $\widetilde{y}$ or $\widehat{y}$ in the nominator) and the upper indice
in the bracket $()$ is related to the degree of the polynomial under the
square root in the denominator - in the case it is $\left( 1-\widetilde{y}%
^{2}\right) \left( 1-q^{2}\widetilde{y}^{2}\right) $ or $\left( \widehat{y}%
^{2}-1\right) \left( 1-q^{2}\widehat{y}^{2}\right) $ and consequently the
polynomial is of the fourth degree. The corresponding elliptic integrals of
the second and of the fourth order are 
\begin{equation}
\widetilde{J}_{2}^{(4)}(\widetilde{y},q)=\int \frac{\widetilde{y}^{2}d%
\widetilde{y}}{\sqrt{\left( 1-\widetilde{y}^{2}\right) \left( 1-q^{2}%
\widetilde{y}^{2}\right) }}\text{ \ \ \ \ ,\ \ }\widetilde{J}_{4}^{(4)}(%
\widetilde{y},q)=\frac{q^{5}}{i}\int \frac{\widehat{y}^{4}d\widehat{y}}{%
\sqrt{\left( \widehat{y}^{2}-1\right) \left( 1-q^{2}\widehat{y}^{2}\right) }}%
\text{ \ \ \ ,\ }  \label{B4}
\end{equation}%
In \cite{Bog2} it was proved that the integrals $\widetilde{T}_{1}$ and $%
T_{2}^{(2)}$ are imaginary quantities, so that the propagation time $T=G(f)$
in expression (\ref{K4}) will be a real quantity, as it should be.

\section{Appendix C: Coefficient functions $%
P_{ij}^{(1)},P_{ij}^{(2)}.......P_{ij}^{(10)}$ in the expression for the
function $K_{1}(y,z)$ in (\protect\ref{Obz141})}

\begin{equation}
P_{ij}^{(1)}(y,z)\equiv \frac{1}{9}\sum\limits_{i,j=1}^{3}\left[ \frac{1}{%
z_{u_{j}}}\frac{\partial }{\partial u_{j}}(\frac{\overset{.}{y}}{\overset{..}%
{y}}.\frac{1}{y_{u_{i}}})+\frac{1}{z_{v_{j}}}\frac{\partial }{\partial v_{j}}%
(\frac{\overset{.}{y}}{\overset{..}{y}}.\frac{1}{y_{u_{i}}})\right]
\Longrightarrow \text{variat.derivative \ }\frac{\partial ^{2}L}{\partial
u_{i}\partial s}\text{ ,}  \label{KK4}
\end{equation}%
\begin{equation}
P_{ij}^{(2)}(y,z)\equiv \frac{1}{9}\sum\limits_{i,j=1}^{3}\left[ \frac{1}{%
z_{u_{j}}}\frac{\partial }{\partial u_{j}}(\frac{\overset{.}{y}}{\overset{..}%
{y}}.\frac{1}{y_{v_{i}}})+\frac{1}{z_{v_{j}}}\frac{\partial }{\partial v_{j}}%
(\frac{\overset{.}{y}}{\overset{..}{y}}.\frac{1}{y_{v_{i}}})\right]
\Longrightarrow \text{variat.derivative \ }\frac{\partial ^{2}L}{\partial
v_{i}\partial s}\text{ \ ,}  \label{KK5}
\end{equation}%
\begin{equation}
P_{ij}^{(3)}(y,z)\equiv \frac{1}{9}\sum\limits_{i,j=1}^{3}\frac{1}{z_{u_{i}}}%
\frac{\overset{.}{y}}{\overset{..}{y}}.\frac{1}{y_{u_{j}}}\Longrightarrow 
\text{variat.derivative \ }\frac{\partial ^{3}L}{\partial u_{i}\partial
u_{j}\partial s}\text{ \ \ ,}  \label{KK6}
\end{equation}%
\begin{equation}
P_{ij}^{(4)}(y,z)\equiv \frac{1}{9}\sum\limits_{i,j=1}^{3}\left[ \frac{1}{%
z_{v_{j}}}\frac{\overset{.}{y}}{\overset{..}{y}}.\frac{1}{y_{u_{i}}}+\frac{1%
}{z_{u_{i}}}\frac{\overset{.}{y}}{\overset{..}{y}}.\frac{1}{y_{v_{j}}}\right]
\Longrightarrow \text{variat.derivative \ }\frac{\partial ^{3}L}{\partial
u_{i}\partial v_{j}\partial s}\text{ \ ,}  \label{KK7}
\end{equation}%
\begin{equation}
P_{ij}^{(5)}(y,z)\equiv \frac{1}{9}\sum\limits_{i,j=1}^{3}\frac{1}{z_{v_{i}}}%
\frac{\overset{.}{y}}{\overset{..}{y}}.\frac{1}{y_{v_{j}}}\Longrightarrow 
\text{variat.derivative \ }\frac{\partial ^{3}L}{\partial v_{i}\partial
v_{j}\partial s}\text{ \ ,}  \label{KK8}
\end{equation}%
\begin{equation}
P_{ij}^{(6)}(y,z)\equiv \frac{1}{9}\sum\limits_{i,j=1}^{3}\frac{1}{y_{u_{i}}}%
T^{(j)}(\overset{.}{z},\overset{..}{z},\overset{...}{z},z_{v_{j}})%
\Longrightarrow \text{variat.derivative \ }\frac{\partial ^{2}L}{\partial
u_{i}\partial u_{j}}\text{ \ \ ,}  \label{KK9}
\end{equation}%
\begin{equation}
P_{ij}^{(7)}(y,z)\equiv \frac{1}{9}\sum\limits_{i,j=1}^{3}\left[ \frac{1}{%
y_{v_{j}}}T^{(i)}(\overset{.}{z},\overset{..}{z},\overset{...}{z},z_{u_{i}})+%
\frac{1}{y_{u_{i}}}T^{(j)}(\overset{.}{z},\overset{..}{z},\overset{...}{z}%
,z_{v_{j}})\right] \Longrightarrow \text{variat.derivative \ }\frac{\partial
^{2}L}{\partial u_{i}\partial v_{j}}\text{ \ ,}  \label{KK99}
\end{equation}%
\begin{equation}
P_{ij}^{(8)}(y,z)\equiv \frac{1}{9}\sum\limits_{i,j=1}^{3}\frac{1}{y_{v_{j}}}%
T^{(i)}(\overset{.}{z},\overset{..}{z},z_{u_{i}})\Longrightarrow \text{%
variat.derivative \ }\frac{\partial ^{2}L}{\partial v_{i}\partial v_{j}}%
\text{ \ ,}  \label{KK10}
\end{equation}%
\begin{equation}
P_{ij}^{(9)}(y,z)\equiv \frac{1}{3}\sum\limits_{i,j=1}^{3}\left[ \frac{1}{%
y_{v_{j}}}\frac{\partial T^{(i)}(\overset{.}{z},\overset{..}{z},\overset{...}%
{z},z_{u_{i}})}{\partial v_{j}}+\frac{1}{y_{u_{i}}}\frac{\partial T^{(j)}(%
\overset{.}{z},\overset{..}{z},\overset{...}{z},z_{u_{i}})}{\partial u_{j}}%
\right] \Longrightarrow \text{variat.derivative \ .}\frac{\partial ^{.}L}{%
\partial u_{i}}\text{ \ ,}  \label{KK11}
\end{equation}%
\begin{equation}
P_{ij}^{(10)}(y,z)\equiv \frac{1}{3}\sum\limits_{i,j=1}^{3}\left[ \frac{1}{%
y_{u_{j}}}\frac{\partial T^{(i)}(\overset{.}{z},\overset{..}{z},\overset{...}%
{z},z_{u_{i}})}{\partial u_{j}}+\frac{1}{y_{v_{j}}}\frac{\partial T^{(j)}(%
\overset{.}{z},\overset{..}{z},\overset{...}{z},z_{v_{i}})}{\partial v_{j}}%
\right] \Longrightarrow \text{variat.derivative \ }\frac{\partial ^{.}L}{%
\partial v_{i}}\text{ \ ,}  \label{KK12}
\end{equation}%
everywhere by $\ T^{(i)}(\overset{.}{z},\overset{..}{z},\overset{...}{z}%
,z_{u_{i}})$ we have denoted the expression

\begin{equation}
T^{(i)}(\overset{.}{z},\overset{..}{z},\overset{...}{z},z_{u_{i}})\equiv 
\frac{1}{3}\frac{\overset{.}{z}}{\overset{..}{z}}\frac{1}{\left(
z_{u_{i}}\right) ^{2}}\overset{.}{z}_{u_{i}}+\frac{1}{3}\frac{\overset{.}{z}%
\overset{...}{z}}{(\overset{..}{z})}\frac{1}{z_{u_{i}}}\text{ \ .}
\label{KK13}
\end{equation}

\section{14.3. Acknowledgements}

The author is grateful to the organizers of the Sixteenth Conference of the
Euro-American Consortium for Promoting the Application of Mathematics in
Technical and Natural Sciences (AMITaNS24) June 21-26, 2024, Albena resourt,
Bulgaria and especially to Prof. Michail Todorov for the opportunity to
participate in this conference.

\end{document}